\PassOptionsToPackage{prologue, x11names}{xcolor}

\documentclass[conference]{IEEEtran}

\PassOptionsToPackage{numbers, compress}{natbib}

\usepackage{natbib}
\usepackage[utf8]{inputenc} 
\usepackage[T1]{fontenc}    
\usepackage{hyperref}       
\usepackage{url}            
\usepackage{booktabs}       
\usepackage{amsfonts}       
\usepackage{nicefrac}       
\usepackage{microtype}      
\usepackage{xcolor}         
\usepackage{amsthm}
\usepackage{algorithm}
\usepackage{amssymb}
\usepackage{algpseudocode}
\usepackage{enumitem}
\usepackage{xspace}
\usepackage{graphicx}
\usepackage{multirow} 
\usepackage[flushleft]{threeparttable}
\usepackage{adjustbox}
\usepackage{subcaption}
\usepackage{pifont}
\usepackage{longtable}
\usepackage{mathtools}
\theoremstyle{plain}
\newtheorem{theorem}{Theorem}
\newtheorem{remark}{Remark}

\newtheorem{corollary}{Corollary}
\newtheorem{lemma}{Lemma}
\newtheorem{proposition}{Proposition}

\theoremstyle{definition}
\newtheorem{definition}{Definition}

\def\ie{\textit{i.e.}\xspace} 
\def\etc{\textit{etc.}\xspace}
\def\eg{\textit{e.g.}\xspace}
\def\aka{\textit{a.k.a}\xspace}
\def\etal{{et~al.}\xspace}

\def\Ebb{\mathbb{E}}
\def\Rbb{\mathbb{R}}
\def\Mcal{\mathcal{M}}
\def\Ncal{\mathcal{N}}
\def\Dcal{\mathcal{D}}

\def\DAE{\mathsf{DAE}}
\def\CM{\mathsf{CM}}

\def\mp{\mathrm{MP}}

\def\gcn{\mathsf{CARIBOU}} 

\def\sage{\textsf{PertGraph}}
\def\dpdgc{\textsf{DPDGC}}
\def\gap{\textsf{GAP}}
\def\mlpE{\textsf{MLP}}

\def\Clip{C_{\mathsf{L}}}

\def\photo{\texttt{Photo}}
\def\facebook{\texttt{Facebook}}
\def\computers{\texttt{Computers}}
\def\cora{\texttt{Cora}}
\def\pubmed{\texttt{PubMed}}

\usepackage{color, colortbl}
\definecolor{greyL}{RGB}{235,235,235}

\definecolor{colorbest}{named}{Purple4}
\definecolor{colorsecond}{named}{MediumPurple1}
\definecolor{colortext}{named}{white}

\newcommand{\greenup}{\rotatebox{90}{\color{green}\ding{225}}}
\newcommand{\reddown}{\textcolor{red}{$\blacktriangledown$}}
\newcommand{\tbspace}{\mathrel{\hspace{0.18cm}}}

\usepackage{amsmath,amsfonts,bm}

\def\eqref#1{equation~\ref{#1}}

\def\1{\bm{1}}

\def\mA{{\bm{A}}}

\def\mD{{\bm{D}}}
\def\mE{{\bm{E}}}

\def\mG{{\bm{G}}}

\def\mI{{\bm{I}}}

\def\mN{{\bm{N}}}

\def\mV{{\bm{V}}}
\def\mW{{\bm{W}}}
\def\mX{{\bm{X}}}
\def\mY{{\bm{Y}}}
\def\mZ{{\bm{Z}}}

\DeclareMathAlphabet{\mathsfit}{\encodingdefault}{\sfdefault}{m}{sl}
\SetMathAlphabet{\mathsfit}{bold}{\encodingdefault}{\sfdefault}{bx}{n}

\def\gK{{\mathcal{K}}}

\def\gN{{\mathcal{N}}}

\newcommand{\E}{\mathbb{E}}

\newcommand{\R}{\mathbb{R}}

\def\chainS{\texttt{Chain-S}}
\def\chainM{\texttt{Chain-M}}
\def\chainL{\texttt{Chain-L}}
\def\chainX{\texttt{Chain-X}}

\newcommand{\yuz}[1]{{\color{magenta} [YuZ: #1]}}

\newcommand{\zl}[1]{{\color{purple} [ZL: #1]}}

\usepackage{tikz}
\usetikzlibrary{shapes.geometric}
\newcommand{\Stars}[2][fill=orange,draw=lightgray]{\begin{tikzpicture}[baseline=-0.35em,#1]
\foreach \X in {1,...,5}
{\pgfmathsetmacro{\xfill}{min(1,max(1+#2-\X,0))}
\path (\X*1.1em,0) 
node[star,draw,star point height=0.25em,minimum size=1em,inner sep=0pt,
path picture={\fill (path picture bounding box.south west) 
rectangle  ([xshift=\xfill*1em]path picture bounding box.north west);}]{};
}
\end{tikzpicture}}

\usepackage[many]{tcolorbox}

\newtheorem{find}{Takeaway}
\tcolorboxenvironment{find}{colback=gray!10,
 colframe=black,
  boxrule=0.5pt,    
  arc=4pt,          
  left=2mm,         
  right=2mm,
  top=1mm,
  bottom=1mm,}

\usepackage{listings}
\usepackage{xcolor} 
\begin{document}
%


\title{Convergent Privacy Framework  for Multi-layer GNNs  through Contractive Message Passing}

 
\author{%
  \IEEEauthorblockN{%
    Yu Zheng\IEEEauthorrefmark{1},
    Chenang Li\IEEEauthorrefmark{1},
    Zhou Li\IEEEauthorrefmark{1}\textsuperscript{\textsection} and
    Qingsong Wang\IEEEauthorrefmark{2}\textsuperscript{\textsection}%
  }%
  \IEEEauthorblockA{\IEEEauthorrefmark{1} University of California, Irvine,
    Email: \{yu.zheng, chenangl, zhou.li\}@uci.edu}

  \IEEEauthorblockA{\IEEEauthorrefmark{2} University of California, San Diego,
    Email: qiw072@ucsd.edu}%
}
 

%


\IEEEoverridecommandlockouts
\makeatletter\def\@IEEEpubidpullup{6.5\baselineskip}\makeatother
\IEEEpubid{\parbox{\columnwidth}{
		Network and Distributed System Security (NDSS) Symposium 2026\\
		24-28 February 2026, San Diego, CA, USA\\
		ISBN 979-8-9894372-8-3\\
		https://dx.doi.org/10.14722/ndss.2026.240255\\
		www.ndss-symposium.org
}
\hspace{\columnsep}\makebox[\columnwidth]{}}

\maketitle

\begingroup\renewcommand\thefootnote{\textsection}
\footnotetext{Co-corresponding authors.}
\endgroup

\begin{abstract}
Differential privacy (DP) has been integrated into graph neural networks (GNNs) to protect sensitive structural information, e.g., edges, nodes, and associated features across various applications.
A  prominent approach is to perturb the message-passing process, which forms the core of most GNN architectures. However, existing methods typically incur a privacy cost that grows linearly with the number of layers (e.g., GAP published in Usenix Security'23), ultimately requiring excessive noise to maintain a reasonable privacy level.
This limitation becomes particularly problematic when multi-layer GNNs, which have shown better performance than one-layer GNN, are used to process graph data with sensitive information. 


In this paper, we theoretically establish that the privacy budget converges with respect to the number of layers by applying privacy amplification techniques to the message-passing process, exploiting the contractive properties inherent to standard GNN operations. Motivated by this analysis, we propose a simple yet effective \emph{Contractive Graph Layer (CGL)} that ensures the contractiveness required for theoretical guarantees while preserving model utility.
Our framework, $\gcn$, supports both training and inference, equipped with a contractive aggregation module, a privacy allocation module, and a privacy auditing module. Experimental evaluations demonstrate that $\gcn$ significantly improves the privacy-utility trade-off and achieves superior performance in privacy auditing tasks.
  


\end{abstract}

%
\IEEEpeerreviewmaketitle

\pagestyle{plain}


\section{Introduction}
Graph neural networks (GNNs)~\cite{iclr/KipfW17,nips/GuC0SG20,iclr/VelickovicCCRLB18}, designed for operating over structural data, have achieved success in various domains, including social networks~\cite{kdd/PerozziAS14,www/ZhangGPH22} and recommendation systems~\cite{www/ChenBSXZHHWH24,sigir/WuWF0CLX21,csur/WuSZXC23}.
At their core, many GNN architectures are built upon the \emph{message-passing paradigm}, where node representations are iteratively updated by aggregating information from their neighbors. However, graph structures often encode sensitive information about relationships and attributes. As a result, GNNs are vulnerable to privacy attacks, including membership inference~\cite{sp/Wulinkteller22,ccs/0001MMBS22,corr/abs-2102-05429} and attribute inference~\cite{uss/000100S022,tkde/SunDYZWYHL23}. These vulnerabilities highlight the urgent need for robust privacy protection mechanisms in graph learning.

Differential privacy (DP)~\cite{csf/mironov2017renyi,wasserman2010statistical,journal/dong2022gaussian} has emerged as a foundational framework to provide formal guarantees against data leakage over the past two decades, with widespread applications in machine learning~\cite{nips/TangPSM23,sp/Yu0PGT19,nips/YangZZ0PL023}, synthetic data generation~\cite{imc/SunCGWL24,iclr/LiWC23,jpc/McKennaMS21,pvldb/PatwaSGMR23}, and beyond. In the context of GNNs, recent works~\cite{arxiv/sajadmanesh2020differential,arxiv/Wuprivacy23} have advanced privacy protection through edge-level DP (EDP)~\cite{sp/0011L0022} and node-level DP (NDP)~\cite{pakdd/IftikharW21,sigmod/DayLL16} guarantees. The primary approach employs \emph{perturbed message passing}, which injects calibrated Gaussian or Laplace noise into aggregation layers to protect the edge or node memberships in the training graph. While these approaches have provided formal privacy guarantees, they share a critical limitation: \emph{the privacy loss grows linearly with the number of layers} $K$ or graph hops. In other words, each additional aggregation layer compounds the privacy cost, ultimately requiring large amounts of noise to maintain a reasonable level of privacy protection. This, in turn, severely degrades model utility. 

Recent advances have shown that multi-layer GNNs, especially deeper GCNs, are essential in capturing complex relationships~\cite{feng2022powerful} and analyzing graphs with long-range interactions~\cite{nips/DwivediRGPWLB22,corr/abs-2406-03386,nips/WuJWMGS21,tmlr/TonshoffRRG24,iclr/ThurlemannR23}.
{In fact, informative long-range interactions exist in a lot of real-world datasets. For instance, in large biological networks, long-range dependencies influence protein functions, requiring more than 10 hops of message passing~\cite{icml/Li0GK21}. In social networks, privacy-sensitive relationships propagate through multi-hop neighborhoods~\cite{nguyen2017probability, liu2014assessment}.
}
As reported in~\cite{icml/Li0GK21}, increasing the network depth leads to a substantial improvement in accuracy from $72.5\%$ to $88.2\%$. 
However, the aforementioned linear dependence on $K$ is particularly challenging for multi-layer GNNs as larger $K$ leads to larger privacy parameter $\epsilon$, \aka weak privacy guarantee. 

Interestingly, empirical studies~\cite{sp/Wulinkteller22} have shown that membership inference attacks are not particularly more successful against multi-layer GNNs, suggesting that the linear dependency of privacy cost on network depth might be an overestimation. This observation aligns with the phenomenon known as ``over-smoothing''~\cite{aaai/ChenLLLZS20} in GNNs, where node representations become increasingly homogeneous as network depth increases and consequently making membership inference more challenging. This homogenization effect might actually provide inherent privacy benefits due to the contractive nature of GNN aggregation operations.
This observation motivates our central research question:

\textit{Can we achieve differentially private graph learning with a \textbf{convergent} (bounded) privacy budget, thereby improving the privacy-utility trade-off for deeper GNNs?}

 \begin{figure}[!t]
	\centering
	\includegraphics[width = 0.38\textwidth]{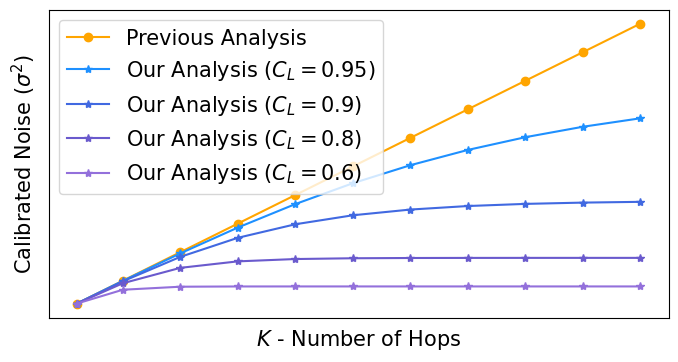}
	\caption{Comparing Calibrated Noise for Perturbed Message Passing. Previous analysis requires $\sigma^2\propto  O(K)$, and our analysis demands $\sigma^2\propto O\Bigl((1-\Clip^K)(1+\Clip)/((1+\Clip^K)(1-\Clip))\Bigl)$, where $\Clip$ is Lipchitz constant. With sensitivity constrained to norm $1$, the signal-to-noise ratio is markedly low as $K$ increases, severely impacting utility.}
	\label{img:privacy_K_summary}
\end{figure}

In this work, we answer this question affirmatively.
{Prior perturbed message-passing mechanisms~\cite{sp/0011L0022, pakdd/IftikharW21,sigmod/DayLL16} assume that privacy loss grows linearly with depth, yet empirical results show that deeper GNNs can be less vulnerable to membership inference. 
We attribute this phenomenon to the contractive nature of common aggregation operators.
In theory, we analyze this contractive property underlying over-smoothing, which leads to bounded sensitivity, so the privacy budget converges with the number of layers  $K$ (see Figure~\ref{img:privacy_K_summary}) instead of growing linearly. This motivates $\gcn$, a privacy-preserving GNN framework that enforces \textit{contractiveness} to mirror real-world GNN behaviors while achieving \emph{convergent privacy budget} through privacy amplification.}

\subsection{Overview of Convergent Privacy Analysis}
\label{sec:intro_converloss}
For a perturbed message-passing GNN with $K$ layers, the standard approach analyzes the privacy loss at each step and then applies the DP composition theorem to aggregate the total privacy cost.
This approach is common in existing privacy analyses of perturbed message-passing GNNs~\cite{arxiv/sajadmanesh2020differential,arxiv/Wuprivacy23,nips/chien2023differentially}, resulting in a privacy budget that scales linearly with $K$, specifically $\epsilon = O(K/\sigma^2) + O(\sqrt{K/\sigma})$.
Thus, as $K$ increases, the amount of injected noise $\sigma$ must also grow, leading to degraded utility particularly when a small $\epsilon$ is desired; see Section~\ref{sec::cha_pro} for a motivating empirical study. 

Inspired by recent advances in privacy amplification~\cite{nips/ChourasiaYS21,icml/BokShifted24} through hidden states and contractive iterative processes, we observe that a similar amplification effect can be exploited in GNNs. 
Here, \emph{contractiveness} refers to the property that the distance between two inputs is reduced after applying the operation, implying reduced distinguishability of outputs from a privacy perspective. The potential privacy amplification in GNNs arises from the following two observations: (1) GNNs typically do not expose intermediate node embeddings during training or inference, focusing only on the final node representations. (2) Standard message-passing operations, such as those used in Graph Convolutional Networks (GCNs)~\cite{iclr/KipfW17} (the dominant model in practice and in empirical studies~\cite{sp/Wulinkteller22}), are inherently contractive, a property that also underlies the over-smoothing phenomenon~\cite{aaai/ChenLLLZS20}.
We theoretically validate this insight by showing that the privacy loss of a $K$-layer perturbed message-passing process with contractive layers satisfies a \emph{convergent privacy budget}.
Specifically, instead of growing linearly with $K$, we show that the privacy budget follows a convergent form:
$
\epsilon = O\left(\frac{(1-\Clip^K)(1+\Clip)}{(1+\Clip^K)(1-\Clip)}\right),
$
where $\Clip$ is the Lipschitz constant of the message-passing operator (see Theorem~\ref{thm:privacy_convergence_general} for details). 
Our improved privacy analysis is achieved by recasting the multi-layer perturbed GNN as a \emph{Contractive Noisy Iteration (CNI)} process~\cite{icml/BokShifted24} and applying the privacy convergence results established for CNIs.

To leverage this analysis in practice, we design a simple yet effective \emph{Contractive Graph Layer (CGL)} that enforces the contractiveness required for our theoretical guarantees while maintaining model expressivity. 
The CGL layer builds upon standard GCN-style aggregation, augmented with residual connections~\cite{icml/ChenWHDL20} and mean aggregation normalization~\cite{alon2020bottleneck}, ensuring expressiveness across many layers without exposing sensitive edge information.

We quantitatively characterize the privacy guarantees of CGL by carefully bounding the Lipchitz constant of the perturbed message-passing operation (Proposition~\ref{prop:lip}) and the sensitivity of the perturbed message-passing with respect to both edge-level privacy (Theorem~\ref{thm:edge-sensitivity}) and node-level privacy (Theorem~\ref{thm:node-sensitivity}). Together, these results allow us to explicitly quantify the privacy budget of the CGL layer using our general theory, culminating in the final privacy guarantee stated in Theorem~\ref{the:overall_epsilon}.

\subsection{$\gcn$: Framework and Evaluation}
\label{sec:intro_framework} Building on perturbed CGL, we realize a private framework $\gcn$ for GNN inference and training.
$\gcn$ includes contractive aggregation module, privacy allocation module, and privacy auditing module.
Together, these design choices enable us to achieve convergent privacy guarantees while maintaining strong GNN performance across graphs with varying interaction ranges.

To evaluate $\gcn$, we conduct extensive experiments over nine graph datasets, including commonly-used real-world datasets and synthetic chain-structured datasets for developing configurable interaction ranges.
The experimental results demonstrate that $\gcn$
improves non-trivial utility over standard graph and chain-structured datasets.
Compared with several SOTA baselines, $\gcn$'s EDP and NDP show significant utility improvements, especially in high privacy regimes, and reasonable computational overhead. 
Table~\ref{tab:compare_design} presents a comprehensive comparison, which is explained in Section~\ref{related_main}.
Ablation studies are provided to understand the relation between privacy-utility hops and various ranges of graph, and the choice of hyper-parameters of CGL.
In addition to privacy verification, we perform auditing experiments based on two membership inference attacks~\cite{sp/Wulinkteller22,tpsisa/OlatunjiNK21}, demonstrating $\gcn$'s robustness. 

\vspace{1mm}
\noindent\textit{\textbf{Contribution.}} In terms of our new insights (Section~\ref{sec::cha_pro}), our contribution includes:
 \begin{enumerate}
    \item A novel privacy analysis for GNNs that leverages the contractiveness of message-passing operations to achieve convergent privacy costs, even for deep networks; (Section~\ref{sec:theory})
    \item The design of perturbed CGL and a practical differentially private GNN framework -- $\gcn$ with provable privacy guarantees and superior utility-privacy tradeoffs; (Section~\ref{sec:framework})
    \item Extensive experimental validation across multiple graph datasets with varying structural properties, demonstrating significant improvements over state-of-the-art private GNN approaches. (Section~\ref{sec:exp})
\end{enumerate}

 \begin{table*}[t]
\small
    \caption{Comparison between Private GNNs. EDP and NDP summarizes the results of private GNNs in Table~\ref{tab:reulst_table_overall_acc_top1}.   
 }
    \label{tab:compare_design}
    \centering
    \setlength\tabcolsep{4pt}
    \adjustbox{max width=\textwidth}{
    \begin{tabular}{c|c|c|c|c|c}
    \toprule[1pt]   \textbf{Framework} &    \textbf{Mechanism} & \textbf{Complexity per Layer}& \textbf{Calibrated Noise ($\sigma$)} & \textbf{EDP Utility}& \textbf{NDP Utility}\\
    \hline
    \sage~\cite{ccs/KolluriBHS22,sp/0011L0022} & Graph perturbation&    $O(|\mV|^2)$ & $\propto 1$  &\Stars{2}&\Stars{1}\\\hline
    \dpdgc~\cite{nips/chien2023differentially} & Decoupled graph with perturbation &$O(|\mE|)$&$\propto\sqrt{K}$& \Stars{3} &\Stars{3.5}\\\hline
       \gap~\cite{uss/sajadmanesh2023gap} & Perturbed message passing  & $O(|\mE|)$  &   $\propto\sqrt{K}$ &\Stars{4} &\Stars{3.5}\\\hline
        $\gcn$    & Perturbed message passing &$O(|\mE|)$  &  $\propto\sqrt{\min(K,\frac{1-\Clip^K}{1+\Clip^K}\frac{1+\Clip}{1-\Clip} )}$&\Stars{5} &\Stars{5}\\
    \bottomrule
    \end{tabular}
    }
\end{table*}


\section{Preliminary}
\subsection{Message Passing Graph Neural Networks}
Graph Neural Networks (GNNs) are a class of neural networks that operate on graph-structured data. Most GNNs follow the message-passing paradigm~\cite{feng2022powerful}, where nodes iteratively aggregate information from their neighbors to update their representations.

\subsubsection{Message Passing Layers} Let $\mG=(\mV, \mE)$ be a graph, where $\mV$ denotes the set of vertices (or nodes) and $\mE$ denotes the set of edges. Let $\mX^{(k)} \in \Rbb^{|\mV| \times d}$ be the node feature matrix at layer $k$, where $d$ is the dimension of the node features. Additionally, we use $\mX_u^{(k)}\in \R^d$ to denote the feature vector of node $u$ at layer $k$. Each layer of a message passing GNN can be generally written as,
\begin{equation}
\label{equ:MP_GNN}
  \text{MP}_G(\mX_u^{(k)}) := \sigma\left(\psi\left(\mX^{(k)}, \oplus_{ v \in \gN(u)} \phi(\mX_u^{(k)}, \mX_v^{(k)}) \right)\right),
\end{equation}
where $\sigma$ is a non-linear activation function, $\gN(u)$ is the set of neighbors of node $u$, $\phi$ is a function that computes the message from node $v$ to node $u$, $\oplus$ represents the aggregation function that processes all messages from the neighbors of node $u$, and $\psi$ is a function that updates the node feature vector of node $u$ with the aggregated messages. 
Examples of message passing GNNs, such as GCN~\cite{iclr/kipf2017semi}  and its variant, are in Appendix~\ref{app:mpgcn}.

\subsubsection{Applications of Message Passing GNNs} Message passing GNNs leverage GNN layers to iteratively refine node representations, which are then employed in tasks like node classification~\cite{ccs/KolluriBHS22}, link prediction~\cite{www/RaoKR24}, and graph classification~\cite{icdm/LiW024}. 
Multi-layer GNNs like deep GNNs~\cite{iclr/ThurlemannR23,nips/WuJWMGS21}  are especially suitable to process
long-range graphs~\cite{nips/DwivediRGPWLB22,corr/abs-2406-03386,tmlr/TonshoffRRG24} by capturing dependencies between distant nodes, which is crucial for tasks like molecular property prediction~\cite{bioinformatics/MaBRHXXYH22}, protein interaction modeling~\cite{bmcbi/ZhongHXLQY22}, and complex node interaction modeling~\cite{feng2022powerful}.

\subsection{Differential Privacy for GNNs}

\begin{definition}[Differential Privacy~\cite{tcc/DworkMNS06}]
Given a data universe $\Dcal$, two datasets $D,D'\subseteq\Dcal$ are adjacent if they differ by only one data instance.
A random mechanism $\Mcal$ is $(\epsilon, \delta)$-differentially private if for all adjacent datasets $D,D'$ and for all events $S$ in the output space of $\Mcal$, we have $\mathrm{Pr}(\Mcal(D)\in S)\leq e^\epsilon \mathrm{Pr}(\Mcal(D')\in S) + \delta$.
\end{definition}

Intuitively, DP~\cite{tcc/DworkMNS06} theoretically quantifies the privacy of a model by measuring the indistinguishability of the outputs of a mechanism $\Mcal$ on two adjacent datasets $D$ and $D'$.
It can be classified into bounded DP and unbounded DP depending on the construction of $D'$, where the former is by replacing a data instance of $D$ and the latter is by addition / removal of a data sample of $D$.
The privacy budget $\epsilon$ is smaller representing a stronger privacy guarantee, while $\delta$ is a slackness quantity that relaxes the pure DP constraint.

\subsubsection{Privacy Definition on Graphs}
In the context of graph data, the notion of adjacency refers to the graph structure, which can be defined as edge-level adjacency (Definition~\ref{def::edp}) and node-level adjacency (Definition~\ref{def::ndp}). 

 
\begin{definition}[Edge-level adjacency~\cite{pvldb/KarwaRSY11}]
\label{def::edp}
	Two graphs $\mG_1=\{\mV_1,\mE_1\}$ and $\mG_2=\{\mV_2,\mE_2\}$ are considered as edge-level neighboring if they differ in a single edge (through addition or removal of the edge), \ie, ($\mV_2=\mV_1) \wedge (\lnot(\mE_2\cap\mE_1)= e_i$) where $e_i\in\mE_1$.
\end{definition}

\begin{definition}[Node-level adjacency~\cite{pvldb/KarwaRSY11}]
\label{def::ndp}
	Two graphs $\mG_1=\{\mV_1,\mE_1\}$ and $\mG_2=\{\mV_2,\mE_2\}$ are considered as node-level neighboring if they differ in a single node and its incident edges (through addition or removal of the node and its incident edges), \ie, $\lnot(\mV_2\cap \mV_1)  =\{n_i, \{e_{ij}\}_{\forall j}\}$ where $n_i\in(\mV_1\cup \mV_2)$ and $ \{e_{ij}\}_{\forall j}$ connects to $n_i$.
\end{definition}



\subsubsection{Perturbed Message Passing with DP}
To incorporate differential privacy into GNNs, one can add noise to the message passing layer, following the perturbed message passing approach~\cite{uss/sajadmanesh2023gap}.
Given a graph $\mG$ and message passing function $\mp$, we define a sequence $\{\mX^{(k)}\}_{k=0}^K$ of node feature matrices by:
\begin{equation}\label{equ:mp_sequence}
  \mX^{(k+1)} = \Pi_\gK(\mp_{\mG}(\mX^{(k)}) + \mZ^{(k)})
\end{equation}
where $\mX^{(0)} = \mX$ is the input feature matrix, $\mZ^{(k)} \sim \mathcal{N}(0, \sigma^2)$ is Gaussian noise, and $\Pi_\gK$ projects features back to bounded set $\mathcal{K}$ (typically constraining $\|\mX_u\|_2 \leq 1$ for each node $u$).
The privacy guarantees of perturbed message passing depend on the sensitivity of the mechanism:

\begin{definition}[Sensitivity of Perturbed Message Passing]
  Let $\mp_{\mG}$, $\mp_{\mG'}$ be the perturbed message passing mechanisms applied to neighboring graphs $G, G'$. The sensitivity is defined as:
  \begin{equation}\label{equ:sensitivity}
    \Delta(\mp) = \max_{G, G'} \max_{\mX \in \mathcal{K}} \|\mp_{\mG}(\mX) - \mp_{\mG'}(\mX)\|_F
  \end{equation}
  where the maximum is taken over all adjacent graphs and all node feature matrices in $\mathcal{K}$.
\end{definition}

The sensitivity determines the scale of noise required for privacy guarantees. Lower sensitivity allows for less noise addition while maintaining the same privacy level, directly affecting the utility-privacy trade-off in differentially private GNNs.

\subsubsection{Privacy Accounting}
Privacy accounting process analyzes the total privacy budget for the composition of several (adaptive) private algorithms. A common approach for analyzing the Gaussian mechanism in perturbed message passing is through R\'enyi differential privacy (RDP)~\cite{csf/mironov2017renyi} and its composition theorem.
\begin{definition}[R\'enyi differential privacy~{\citep{csf/mironov2017renyi}}]
  A randomized algorithm $\mathcal{M}$ is $(\alpha, \epsilon)$-RDP for $\alpha > 1$, $\epsilon > 0$ if for every adjacent dataset $X, X'$, we have $D_{\alpha}(\mathcal{M}(X) \| \mathcal{M}(X')) \leq \epsilon$, where $D_{\alpha}(P \| Q)$ is the Rényi divergence of order $\alpha$ between probability distributions $P$ and $Q$ defined as:
  $$D_{\alpha}(P \| Q) = \frac{1}{\alpha - 1} \log \mathbb{E}_{x \sim Q} \left[ \left(\frac{P(x)}{Q(x)}\right)^{\alpha} \right]$$
\end{definition}

\begin{theorem}[Composition of RDP~\cite{csf/mironov2017renyi}]
  \label{the:rdp_composition}
  If \(\Mcal_1, \dots, \Mcal_k\) are randomized algorithms satisfying, respectively, \((\alpha, \epsilon_1)\)-RDP, \dots, \((\alpha, \epsilon_k)\)-RDP, then their composition defined as \((\Mcal_1(S), \dots, \Mcal_k(S))\) satisfies \((\alpha, \epsilon_1 + \cdots + \epsilon_k)\)-RDP.
\end{theorem}

In this work, we present our privacy results in terms of RDP for ease of interpretation, while our underlying analysis employs the tighter $f$-DP framework. This analysis leverages recent advances in privacy amplification techniques~\cite{nips/SchuchardtSKG24,icml/BokShifted24} to achieve stronger privacy guarantees. The technical details of our convergent privacy analysis are discussed in \S\ref{sec:appro_overview}.

\begin{figure}[!t]
	\centering
	\includegraphics[width = 0.48\textwidth]{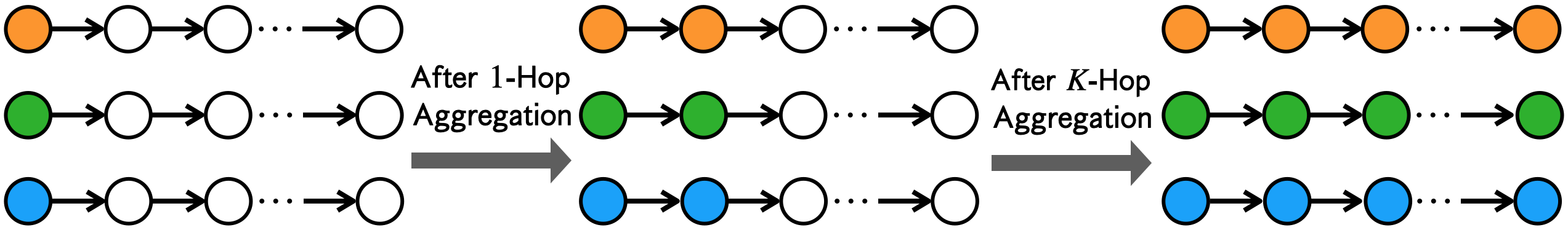}
	\caption{Message Passing on Chain-structured Dataset}
	\label{img:chain}
\end{figure}

\section{Private Multi-layer GNNs Initiative}

\label{sec::cha_pro}
Multi-layer GNNs~\cite{iclr/ThurlemannR23,nips/WuJWMGS21} are vital in tasks like modeling molecular structures, where some properties depend on long range interactions~\cite{nips/DwivediRGPWLB22,corr/abs-2406-03386,tmlr/TonshoffRRG24}  of the nodes~\citep{zhang2022deep,li2023neural}. 
Specifically, they  require messages  flowing across multiple hops before reaching a target node through  stacking multiple layers to exchange information across $K$-hop neighborhoods.
However, ensuring privacy in multi-layer GNNs poses key challenges.
In this section, we outline and illustrate these challenges via a motivating case study (Section~\ref{subsec:gap-chain}), ultimately motivating our new design insights (Section~\ref{sec:appro_overview}) based on \emph{contractive message passing}.

\subsection{Observations on Privacy Accumulation and Performance Degradation}
\label{subsec:gap-chain}

\subsubsection{A Case Study on $\gap$'s Performance on Learning Long-Range Interactions}
To further investigate this phenomenon, we evaluate $K$-layer \gap\ on the \emph{chain-structured} dataset (see Figure~\ref{img:chain}), which is adopted in~\cite{nips/GuC0SG20, liu2021eignn} to 
examine long-range interaction learning capabilities. Specifically, this dataset creates a controlled environment for evaluating GNN performance on long-range interactions. A model without learning from graph structure such as \mlpE\ would fail since most nodes in the dataset have zero-valued feature vectors regardless of their chain type. For a message-passing GNN to correctly classify nodes positioned $K$ hops away from the informative first node, it must perform at least $K$ propagation steps to transfer the meaningful features across the chain as shown in Figure~\ref{img:chain}. This requirement becomes particularly challenging in the private setting as noise must be injected after each message passing layer, potentially overwhelming the signal being propagated.


Figure~\ref{fig:motivating} compares $\gap$' accuracy of its non-private version ({blue solid line}) with its private versions ({dashed lines}), under different privacy budgets $\epsilon \in \{1,2,4,8,16,32\}$. 
It shows a  binary node classification over multiple $8$-node chains as a case study. 
The red solid line in Figure~\ref{fig:motivating} represents random guessing (50\% accuracy).

\begin{figure}[!t]
	\vspace{-12pt}
	   \subfloat[Strong Privacy]{
	   \hspace{-0.2cm}
	   \includegraphics[width = 4.2cm]{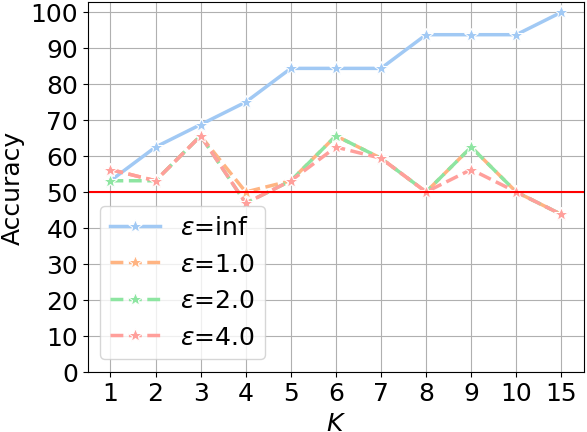}
	   \label{fig:motivating_1}
	   }
	   \subfloat[Weak Privacy]{
	   \hspace{-0.2cm}
	   \includegraphics[width = 4.2cm]{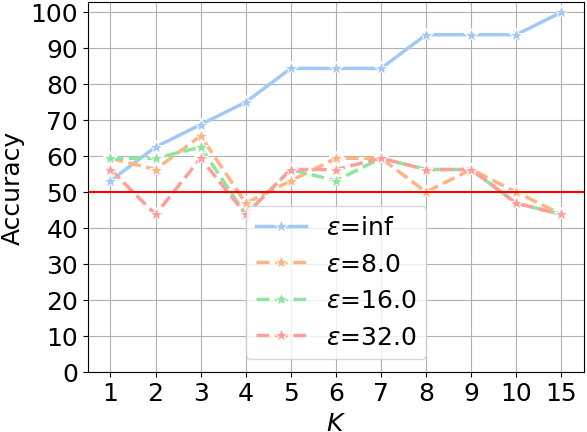}
	   \label{fig:motivating_2}
	   }
	   \caption{Motivating Experiments of Classification Model over Chain-structured Datasets. }
	   \label{fig:motivating}
\end{figure}

 Our exploration reveals a stark contrast between private and non-private settings:
\begin{itemize}[leftmargin=*]
	\item \emph{Non-private setting.} $\gap$'s accuracy consistently improves with increasing layer depth, as message passing enables feature propagation across the chain. The model achieves satisfactory performance after sufficient depth ($K \geq 8$), ultimately reaching perfect classification (100\% accuracy) at $K=15$ hops—demonstrating the necessity of deep architectures for capturing long-range dependencies.
	
	\item \emph{Private setting.} Privacy protection dramatically degrades model utility. As shown in Figure~\ref{fig:motivating}, even with relatively generous privacy budgets ($\epsilon=16$ or $\epsilon=32$), performance remains marginally above random guessing (50\%). Crucially, increasing depth offers no benefit and often harms performance, as noise accumulates exponentially across layers. This confirms our theoretical concerns: standard approaches to privacy in GNNs fundamentally limit the ability to learn long-range interactions.
\end{itemize}

\begin{figure*}[!t]
	\centering
	\includegraphics[width = 0.98\textwidth]{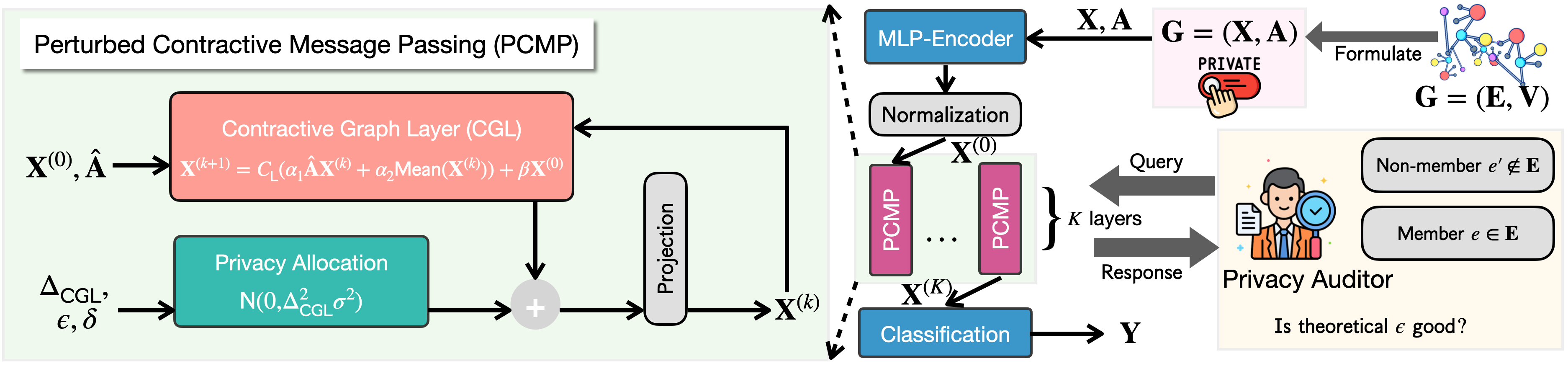}
    \caption{Overview of $\gcn$. The framework integrates CAM, PAM, PDM (Section~\ref{sec:framework}). Combining CAM and PAM can form PCMP. Together, these modules enable scalable and accurate private GNN learning under a convergent privacy budget.}
	\label{img:overview}
\end{figure*}

\subsubsection{Empirical (Counter-Intuitive) Observation: Deeper GNNs May Enhance Privacy}
Recent work by \cite{wu2024provable} revealed a counter-intuitive phenomenon: deeper GNNs empirically exhibit lower vulnerability to membership and link inference attacks. This challenges standard privacy composition analysis, which suggests privacy risk increases with model depth due to multiple queries of the graph data. 
The insight stems from the \emph{over-smoothing}\footnote{Incidentally, another similar phrase is ``over-fitting'', which refers to a model performing well on training data but poorly on unseen data due to memorization. We emphasize that they are different concepts to avoid confusion.}~\cite{aaai/ChenLLLZS20} phenomenon,
where node representations become indistinguishably homogeneous phenomenon, where node representations converge toward similar values as depth increases, making it inherently difficult for adversaries to distinguish individual nodes or infer sensitive relationships.

This observation suggests that standard privacy analysis may be overly pessimistic. The conventional approach of adding noise that scales linearly with depth may be unnecessarily conservative, as the natural privacy amplification properties of over-smoothing could enable tighter privacy bounds with less noise per layer.

\subsection{Core Idea for Convergent Privacy}
\label{sec:appro_overview}

Our insight is to identify and leverage the inherent privacy amplification that occurs in multi-layer GNNs through \emph{contractiveness} (Definition~\ref{def:contractive}). 
When nodes aggregate information from their neighbors (e.g., graph convolution), the resulting representations necessarily become more similar to each other. 

\begin{definition}[Contractive Map]
\label{def:contractive}
A map $f: \mathbb{R}^d \to \mathbb{R}^d$ is said to be contractive with respect to a norm $\|\cdot\|$ if there exists a constant $c < 1$ such that for all $x, y \in \mathbb{R}^d$:
$\|f(x) - f(y)\| \leq c \|x - y\|$
where $c$ is the contractiveness coefficient that governs the rate of contraction.
\end{definition}

\textit{Designing Private GNNs with Convergent Privacy.}
We aim to leverage the privacy amplification properties of contractive map to design a new framework for private GNNs.
This framework is inspired by the recent advances in differentially private gradient descent (DP-GD)~\cite{nips/ChourasiaYS21,icml/BokShifted24}, which has shown that the privacy cost can \emph{converge to a finite value} even with arbitrarily many iterations.
Motivated by it, we aim to translate the advanced privacy analysis techniques from DP-GD to GNNs.
We observe that perturbed message passing (Equation~\ref{equ:mp_sequence}) in GNNs follows a strikingly mathematical parallel pattern as DP-GD:
This parallel structure  enables to leverage the insight of convergent privacy analysis in the context of GNNs, provided we ensure two critical conditions:
\begin{enumerate}[leftmargin=1em]
    \item \textbf{Hidden intermediate embeddings:} Release only the final node representations $\mX^{(K)}$ after $K$ layers, concealing all intermediate states; (Section~\ref{sec:security_model})
    \item \textbf{Contractive message passing:} Design the message passing operation $\mp_{\mG}$ to be provably contractive with coefficient $c < 1$, ensuring $\|\mp_{\mG}(\mX) - \mp_{\mG}(\mY)\|_F \leq c \|\mX - \mY\|_F$ for all node feature matrices $\mX, \mY$. (Section~\ref{sec:cont_layer})
\end{enumerate}
When the perturbed message passing step is contractive with respect to the $\ell_2$ norm, the distance between GNNs trained on neighboring datasets shrinks at each step. Consequently, the influence of individual data points diminishes, leading to the amplified privacy rooted from ``over-smoothing'' effect.
Accordingly, our insight challenges the previous analysis on private GNNs that accumulates the privacy loss from multi-hop GNN aggregations linearly, and simultaneously removes the over-estimated  privacy loss of finally released GNN model to derive a much tighter bound.

\section{Private GNNs with Contractiveness}
 \label{sec:framework}

Building on insights in Section~\ref{sec::cha_pro},
we propose $\gcn$, a modular private GNN framework with convergent privacy budget (Figure~\ref{img:overview}), including three key modules:

\begin{itemize}
    \item[1)] \textbf{Contractive Aggregation Module (CAM, \S~\ref{sec:cont_layer}):} We propose a \textit{new design of message passing layers} with carefully controlled Lipschitz constants to ensure contractiveness so that we can add a reasonable, small yet sufficient noise to protect the computed message passing; Then, we realize a \textit{new DP mechanism of perturbed message passing that only releases final node representations $\mX^{(K)}$}, preventing adversary from exploiting intermediate states and thus amplifying the privacy guarantee rooted in hidden node embeddings.
    \item[2)] \textbf{Privacy Allocation Module (PAM, \S~\ref{sec:pamp}):} PAM ensures efficient privacy budget allocation for edge- and node-level DP guarantees, in which noise calibration is based on \textit{new convergent privacy analysis}.
    \item[3)] \textbf{Privacy Auditing Module (PDM, \S~\ref{sec:privacy_veri}):} PDM empirically audits graph DP through tests such as membership inference attacks, ensuring theoretical guarantees align with practical deployments.
\end{itemize}

\textbf{Security Model.}
\label{sec:security_model}
Aggregation-based GNNs such as GCN~\cite{iclr/kipf2017semi}, GCN-II~\cite{icml/ChenWHDL20}, and SAGE~\cite{nips/HamiltonYL17} reveal only the final node embeddings, keeping intermediate embeddings hidden. 
The final node embeddings learned in this DP manner can be subsequently applied to downstream tasks such as node classification~\cite{www/WangLDLW24}, link prediction~\cite{www/RaoKR24}, and graph classification~\cite{icdm/LiW024}. 
Since DP introduces utility loss, \gap~\cite{uss/sajadmanesh2023gap} seeks to improve model utility by concatenating all intermediate node embeddings from the $K$ layers. 
However, it causes the noise scale $\sigma$ to grow linearly as $O(\sqrt{{K}/{\epsilon^2 }})$, and this privacy bound is loose. 
This work takes a step further beyond \gap, removing the assumption that all intermediate node embeddings must be revealed under a more realistic security model~\cite{nips/0001S22}. 

Overall, $\gcn$ (Figure~\ref{img:overview}) achieves accurate and private multi-layer GNNs, 
which theory will be established in Section~\ref{sec:theory}.
Table~\ref{tab:syn_datasets} summarizes acronyms and symbols.


\begin{table}[t]
\setlength\tabcolsep{10pt}
\begin{center}
\begin{tabular}{|c|l|}
\hline\hline
CGL & Contractive Graph Layer.\\\hline
RDP & R\'enyi differential privacy.\\\hline
MIA &Membership Inference Attacks.\\\hline\hline
$K$	& The number of layers or hops.	 \\\hline
$\epsilon,\sigma$	 & Privacy parameters: budget, noise scale. \\\hline
$\Clip,\Delta$	 & Lipchitz constant, sensitivity.  \\\hline
$\alpha_1,\alpha_2,\beta$& Hyper-parameters of CGL.\\\hline\hline
\end{tabular}
\end{center}
\caption{Acronyms \& Symbols }
\label{tab:syn_datasets}
\end{table}
 
\subsection{Contractive Aggregation Module}
\label{sec:cont_layer}
Contractive operations in GNNs, such as the graph convolution (GConv) layer~\cite{iclr/KipfW17}, inherently reduce privacy risks by mitigating the memorization of GNN parameters through the over-smoothing phenomenon~\cite{aaai/ChenLLLZS20}. This property aligns well with the need for $K$-layer aggregation in long-range graphs, where a target node aggregates embeddings from distant source nodes. However, directly stacking $K$ GConv layers, as shown in Figure~\ref{fig:motivating}, is suboptimal due to limited expressive power~\cite{icml/ChenWHDL20} and heightened sensitivity to DP noise~\cite{WuSZFYW19}. 

To address these limitations, we propose the Contractive Aggregation Module (CAM), centered on the Contractive Graph Layer (CGL). The CGL introduces adjustable coefficients to enhance flexibility and utility while maintaining privacy. Formally, the CGL is defined as:
\begin{equation}
\label{equ:contractlayer}
\mX^{(k+1)} = \Clip\big(\alpha_1 \hat{\mA}\mX^{(k)} + \alpha_2\mathsf{Mean}(\mX^{(k)})\big) + \beta \mX^{(0)},
\end{equation}
where $0 \leq \Clip < 1$ ensures contractiveness, and $\alpha_1, \alpha_2, \beta > 0$ with $\alpha_1 + \alpha_2 = 1$ are hyperparameters. Here, $\hat{\mA}$ is the symmetrically normalized adjacency matrix~\footnote{$\hat{\mA} = \mD^{-\frac{1}{2}}(\mA + \mI)\mD^{-\frac{1}{2}}$ has been widely adopted after GCN emerged, where $\mD$ is the degree matrix and $\mI$ is the identity matrix.},
$\mathsf{Mean}(\mX^{(k)})$ computes the mean of node embeddings, and $\mX^{(0)}$ represents the initial node features.

\textit{Analysis of hyperparameters.} The hyperparameter $\Clip$ ensures the contractiveness by bounding the magnitude of the output embeddings. Specifically, $\Clip$ enforces a Lipschitz constraint, ensuring that small changes in the input embeddings $\mX^{(k)}$ result in proportionally small changes in the output $\mX^{(k+1)}$. 
The coefficients $\alpha_1,\alpha_2$ govern the relative contributions of the graph-based aggregation $\hat{\mA}\mX^{(k)}$ and the mean-based aggregation $\mathsf{Mean}(\mX^{(k)})$, respectively. By satisfying the constraint $\alpha_1 + \alpha_2 = 1$, these coefficients ensure a convex combination of the two components. A higher $\alpha_1$ emphasizes the influence of the graph topology, leveraging structural information from the adjacency matrix $\hat{\mA}$. 
The parameter $\beta$ controls the residual connection $\beta \mX^{(0)}$ calculated independently of the graph topology, incorporating the characteristics of the initial node $\mX^{(0)}$ into the output. 
This residual term mitigates the vanishing gradient problem by preserving a direct path for gradient flow. 

Together, the CGL combines three key components: (1) the graph-based aggregation $\hat{\mA}\mX^{(k)}$, (2) the mean-based aggregation $\mathsf{Mean}(\mX^{(k)})$, and (3) the residual connection $\beta \mX^{(0)}$. This design balances the trade-off between privacy and utility by limiting the propagation of noise while preserving expressive power.

\textit{Comparison with existing design.} 
\gap~\cite{uss/sajadmanesh2023gap} enhances expressiveness by connecting all intermediate embeddings but lacks contractiveness, compromising utility and privacy. In contrast, CGL introduces adjustable coefficients, achieving a balanced trade-off between privacy, utility, and generalizability.
Table~\ref{tab:compare_design} summarizes the general design characteristics between \gap\ and $\gcn$. 

\subsection{Privacy Allocation Module}
\label{sec:pamp}
For multi-hop aggregation, the features from the previous hop $\mX^{(k-1)}$ are aggregated using the adjacency matrix $\hat{\mA}$ to enable message passing to neighboring nodes. To ensure privacy, Gaussian noise $\mathcal{N}(\mu, \sigma^2)$ is added to the aggregated features, where the noise variance $\sigma^2$ is determined by the privacy budget $\epsilon$, ensuring compliance with DP guarantees. 

Building on Section~\ref{sec:cont_layer}, we integrate CAM and PAM to design the Perturbed Contractive Message Passing (PCMP) (Algorithm~\ref{alg:PLMA}). By leveraging contractiveness, PCMP ensures bounded privacy loss for long $K$-hop graphs, eliminating the linear growth of privacy loss with $K$. A subsequent projection step enforces Lipschitz constraints, maintaining consistent scaling across hops. After $K$ iterations, PCMP outputs private feature matrices $\hat{\mathbf{X}}^{(K)}$, which are passed to the classification module.

\begin{algorithm}[t]
\caption{Perturbed Contractive Message Passing (PCMP)}
\label{alg:PLMA}
\begin{algorithmic}[1]
\Require Graph $\mathcal{G} = (\mathcal{V}, \mathcal{E})$ with adjacency matrix $\mA$; The number of hops $K$; Lipchitz constant $\Clip$, DP parameters $\epsilon,\delta$; Initial  normalized features $\mathbf{X}^{(0)}$; 
\Ensure Private aggregated node embeddings $\hat{\mX}^{(K)}$.
\State $ $
\State \Comment{\underline{\textit{Calculate the Required Noise Calibration (from PAM).}}}
\If{Edge-level privacy}
\State Calculate $ \Delta(\mathsf{CGL})$ through Equation~\ref{equ:sensivity_edge}
\ElsIf{Node-level privacy}
\State Calculate $\Delta(\mathsf{CGL})$ through Equation~\ref{equ:sensivity_node}
\EndIf
\State Calculate $\sigma^2$ through Theorem~\ref{the:overall_epsilon} 
\State $ $
\State \Comment{\underline{\textit{Perturbed Contractive Message Passing (from CAM).}}}
\For{$k = 0, \dots, K-1$}
    \State ${\mX}^{(k+1)} \gets\Clip(\alpha_1 \hat{A}\mX^{(k)}+\alpha_2\mathsf{Mean}(\mX^{(k)}))+\beta\mX^{(0)}$
    \State \Comment{{\color{gray}Contractive graph layer: compute node embeddings.}}
    \State ${\mX}^{(k+1)} \gets {\mX}^{(k+1)} + \mathcal{N}(\mu, (\Delta(\mathsf{CGL}))^2\sigma^2)$ 
    \State \Comment{{\color{gray}DP Perturbation.}}
    \State ${\mX}^{(k+1)}  \gets \Pi_{\mathcal{K}}({\mX}^{(k)})$ \Comment{{\color{gray}Projection with norm $1$.}}
\EndFor
\State $ $
\State \textbf{Return}: ${\mathbf{X}}^{(K)}$ 
\end{algorithmic}
\end{algorithm}

\textit{Privacy Budget Allocation.} The privacy budget is distributed across $K$ hops to ensure the total privacy loss adheres to the specified $\epsilon$ and $\delta$. The noise scale is calibrated based on the sensitivity of graph operations and the desired DP guarantees, using a noise allocation mechanism (NAM) that limits noise accumulation under Lipschitz constraints. 
The maximum allowable noise calibration (see Corollary~\ref{the:CGL_rdp}) is constrained by
$K' = \min(K, (1 + \Clip) / (1 - \Clip))$.

PCMP integrates contractive aggregation and privacy-preserving perturbation for private message passing. Noise calibration begins by determining the sensitivity of graph operations and computing the noise scale. At each hop, the CGL aggregates features while maintaining contractiveness through adjustable coefficients. Gaussian noise is added to ensure privacy, and embeddings are normalized to enforce Lipschitz constraints. This iterative process produces private embeddings suitable for downstream tasks.

\subsection{Privacy Auditing Module}   
\label{sec:privacy_veri}
Message passing mechanisms integrate graph structures by recursively aggregating and transforming node features from neighbors. Membership Inference Attacks (MIA) on graph examine the vulnerability of message passing to infer whether specific nodes or edges were part of a GNN's training set~\cite{tpsisa/OlatunjiNK21}. The adversaries exploit black-box access to the GNN, querying it with selected data and analyzing outputs (e.g., class probabilities) to infer membership.
To align theoretical privacy guarantees with practical deployments, we propose an empirical auditing module. This module simulates an adversary to evaluate GNN privacy before deployment. Extending Carlini~\etal~\cite{sp/CarliniCN0TT22}'s MIA framework to graphs, we define a graph-specific privacy auditing game (Definition~\ref{def:mia-game}) and implement two real-world attacks~\cite{sp/Wulinkteller22,tpsisa/OlatunjiNK21} for privacy verification.

\begin{definition}[Graph-based Privacy Auditing Game]
\label{def:mia-game}
The game proceeds between a challenger $\mathcal{C}$ and an privacy auditor $\mathcal{A}$:

\begin{enumerate}
    \item The challenger samples a training graph in the transductive setting (a set of subgraphs in the inductive setting) $G \leftarrow \mathbb{G}$ and trains a model $f_\theta \leftarrow \mathcal{T}(G)$ on the dataset $G$.
    \item The challenger flips a bit $b$, and if $b = 0$, samples a fresh challenge point from the distribution $(x, y) \leftarrow \mathbb{G}$ (such that $(x, y) \notin G$). Otherwise, the challenger selects a point from the training set $(x, y) \overset{\$}{\leftarrow} G$.
    \item The challenger sends $(x, y)$ to the adversary.
    \item The adversary gets query access to the distribution $\mathbb{G}$, and to the model $f_\theta$, and outputs a bit $\hat{b} \leftarrow \mathcal{A}^{\mathbb{G}, f}(x, y)$.
    \item Output 1 if $\hat{b} = b$, and 0 otherwise.
\end{enumerate}
\end{definition}

The attacker can output either a ``hard prediction,'' or  a continuous \textit{confidence score}, thresholded as a reference to yield a membership prediction. 

\subsection{Putting it Together}

As shown in Figure~\ref{img:overview}, before and after the perturbed message passing, $\gcn$ employs an encoder and a classification module (CM) to process the node features. To ensure compliance with DP guarantees, the framework utilizes standard DP-SGD during pre-training. 

Upon completing the private $K$-hop aggregations, the resulting private graph embeddings are passed to the CM. The CM integrates two key components: (1) the graph-agnostic node features $\mX^{(0)}$, which capture individual node characteristics independent of the graph structure, and (2) the private, topology-aware aggregated features $\mX'^{(K)}$, which encode structural information from the graph. This dual integration enhances the model's expressiveness while preserving privacy.

To improve classification accuracy, $\gcn$ adopts a head \mlpE\ architecture proposed by Sajadmanesh \etal~\cite{uss/sajadmanesh2023gap} as the CM. This design ensures that the CM effectively combines the information from both feature sets, enabling robust node classification. Furthermore, the CM guarantees efficient training by leveraging the graph-agnostic features $\mX^{(0)}$, ensuring a lower-bound performance even in scenarios where graph topology is unavailable.

\section{Convergent Privacy Analysis}
\label{sec:theory}

In this section, we present a convergent privacy analysis for perturbed message-passing GNNs with respect to the number of hops. 
In Appendix~\ref{sec:one-layer-privacy}, we review the standard privacy analysis for a one-layer perturbed message-passing GNN, and then illustrate that the privacy cost grows linearly with the number of layers under standard composition theorems.
We then shift our focus to a \emph{convergent} privacy analysis for perturbed GNNs whose message-passing layers are contractive. In particular, we draw upon the framework of \emph{contractive noisy iteration} (CNI) from \cite{icml/BokShifted24}, recasting the multi-layer perturbed GNN as a CNI process. Our analysis reveals that, under hidden intermediate states and contractive message-passing layers, the privacy cost converges as the number of hops increases.
Finally, we specialize this result to our proposed $\gcn$.
We show that $\gcn$'s message-passing operation is contractive, derive its sensitivity, and thereby establish concrete bounds on both edge-level and node-level differential privacy for arbitrarily many hops. All proofs are deferred to the Appendices~\ref{app:convergent-privacy},\ref{app:gcn_privacy}.

\subsection{Contractive Noisy Iteration and Convergent Privacy}
\label{sec:convergent-privacy}
Many GNNs, such as GCN, use "mean-type" aggregation, mixing a node's representation with that of its neighbors. Intuitively, iterative mixing could "amplify" privacy, but existing analyses yield only linear compositions. Our key observation is that perturbed message passing in a contractive GNN layer behaves analogously to noisy gradient descent or \emph{noisy iterative maps} \citep{nips/ChourasiaYS21,nips/AltschulerT22,siamcomp/AltschulerBT24}, where recent work has demonstrated \emph{privacy amplification via iteration}.

Below, we introduce the framework of \emph{contractive noisy iteration} (CNI) from \cite{icml/BokShifted24}, and the meta theorem proved by \cite{icml/BokShifted24} that provides a tight privacy guarantee for CNI processes.  

\begin{definition}[Contractive Noisy Iteration (CNI){\cite[Definition~3.1]{icml/BokShifted24}}]
    \label{def:CNI}
    Consider a sequence of random variables
    \begin{equation}
        \mX^{(k+1)} = \Pi_{\mathcal{K}}\!\bigl(\,\phi_{k+1}\bigl(\mX^{(k)}\bigr)+\mZ^{(k)}\bigr),
    \end{equation}
    where each map $\phi_{k}$ is Lipchitz continuous, $\mZ^{(k)} \sim \mathcal{N}\bigl(0,\,\sigma^2 \,\mathrm{I}\bigr)$ are i.i.d. Gaussian noise vectors independent of $\mX^{(k)}$, and $\Pi_{\mathcal{K}}$ is a projection operator onto a convex feasible set $\mathcal{K} \subseteq \mathbb{R}^d.$ This iterative process is called \emph{contractive noisy iteration} (CNI).
\end{definition}

A special case of CNI considered in~\cite{icml/BokShifted24} is the noisy gradient descent, where the contractive functions are the gradient update steps for some fixed loss function $f$, and the noise distribution $\xi_k$ is Gaussian.
This situation is similar to the perturbed message passing mechanism in a contractive GNN layer, where the contractive function $\phi_k\equiv \mp_G$ is the message passing operation for a fixed graph $\mG$.

The complete privacy analysis of CNI processes in \cite{icml/BokShifted24} involves concepts of \emph{trade-off functions} (Definition~\ref{def:fdp_tradeoff}) and \emph{Gaussian differential privacy} (GDP, Definition~\ref{def:GDP}), which we briefly review in Appendix~\ref{app:pre_fdp}.
GDP can be converted to the more familiar Rényi Differential Privacy (RDP) framework:

\begin{lemma}[GDP Implies RDP~{\cite[Lemma~A.4]{icml/BokShifted24}}]
    \label{lem:gdp_to_rdp}
    If a mechanism is $\mu\text{-GDP},$ then it satisfies
    $\left(\alpha, \frac{1}{2}\,\alpha\,\mu^2\right)\text{-RDP}$
    for all $\alpha > 1.$
\end{lemma}

With this connection established, we now state the key meta theorem from \cite{icml/BokShifted24} that analyzes the CNI process and provides a tight privacy guarantee.

\begin{theorem}[Meta Theorem on CNI~{\citep[Theorem~C.5]{icml/BokShifted24}}]
    \label{thm:privacy_convergence_CNI}
    Let $\{\mX^{(k)}\}_{k=0}^K$ and $\{\mX'^{(k)}\}_{k=0}^K$ represent the output of CNI processes,
    $$
        \begin{aligned}
             & \text{CNI}(\mX^{(0)}, \{\phi_k\}_{k=1}^K, \{\Ncal(0, \sigma^2 \mI_d)\}_{k=1}^K, \gK), \text{ and} \\
             & \text{CNI}(\mX^{(0)}, \{\phi_k'\}_{k=1}^K, \{\Ncal(0, \sigma^2 \mI_d)\}_{k=1}^K, \gK).
        \end{aligned}
    $$
    Assume that:
    \begin{itemize}
        \item $\phi_1$ and $\phi'_1$ are Lipschitz continuous,
        \item each $\phi_k, \phi_k'$ is $\gamma$-Lipschitz, with $\gamma < 1$ for $k = 2, \dots, K$,
        \item $\|\phi_k(x) - \phi_k'(x)\| \leq s$ for all $x$ and $k = 1, \dots, K$,
    \end{itemize}
    Then the tradeoff function $T(\mX^{(K)}, \mX'^{(K)})$ satisfies,
    $$
        T(\mX^{(K)}, \mX'^{(K)}) \geq T(\gN(0, 1), \gN(\mu^{(K)}, 1)),
    $$
    where,
    $$\mu^{(K)} = \sqrt{\frac{1-\gamma^K}{1+\gamma^K}\frac{1+\gamma}{1-\gamma}}\frac{s}{\sigma}.$$
\end{theorem}

\begin{remark}
The theorem above slightly generalizes the original result from \cite{icml/BokShifted24} by relaxing the Lipschitz condition to require $\gamma < 1$ only for $k \geq 2$ rather than for all iterations. This relaxation is critical for analyzing our $\gcn$ architecture, where the first message-passing layer includes the residual term $\beta \mX^{(0)}$, potentially making it non-contractive while subsequent layers remain contractive. The proof extends the original argument by carefully tracking the influence of the first layer on the privacy guarantee. For completeness, we provide the full proof in Appendices~\ref{app:convergent-privacy},\ref{app:gcn_privacy}.
\end{remark}

By utilize the  above meta theorem and property of Gaussian tradeoff function, we can derive the privacy guarantee for a $K$-layer perturbed message-passing GNN with contractive message passing layers.

\begin{theorem}[{\cite[Theorem~4.2]{icml/BokShifted24}} adapted for GNNs]
    \label{thm:privacy_convergence_general} 
    Let $\mp$ be a message passing mechanism of a GNN such that the message passing operator $\mp_G$ for any graph $\mG$ is contractive with Lipschitz constant $\gamma<1$ for layers $k\geq 2$.
    Assume the sensitivity of $\mp$ is $\Delta(\mp)$ and the noise scale is $\sigma$.
    Then, a $K$-layer perturbed message passing GNN with $\mp$ satisfies,
    $$\left(\alpha, \frac{1}{2}\alpha \frac{\Delta^2(\mp)}{\sigma^2} \frac{1-\gamma^K}{1+\gamma^K}\frac{1+\gamma}{1-\gamma}\right)\text{-RDP},$$
    which is equivalent to
    $\left(\frac{1}{2}\alpha \frac{\Delta^2(\mp)}{\sigma^2} \frac{1-\gamma^K}{1+\gamma^K}\frac{1+\gamma}{1-\gamma} + \frac{\log(1/\delta)}{\alpha - 1}, \delta\right)\text{-DP}.$
\end{theorem}

The upshot is that as $K \to \infty$, the privacy parameter converges to $\frac{1}{2}\alpha \frac{\Delta^2(\mp)}{\sigma^2} \frac{1+\gamma}{1-\gamma}$, rather than growing unbounded with $K$ as in standard composition. This result enables meaningful privacy guarantees even for deep GNNs with many message-passing layers.

\subsection{Edge- and Node-Level Privacy of $\gcn$}
\label{sec:gcn_privacy}

We now specialize to the $\gcn$ architecture and establish both edge-level and node-level DP guarantees. Recall the \emph{contractive graph layer (CGL)} of $\gcn$:
\begin{equation*}
    \text{CGL:} \quad \quad   \mX^{(k+1)} = \Clip(\alpha_1 \hat{\mA}\mX^{(k)}+\alpha_2\mathsf{Mean}(\mX^{(k)}))+\beta \mX^{(0)},
\end{equation*}
where $0\leq \Clip < 1$, and hyper-parameters, $\alpha_1, \alpha_2, \beta> 0, \alpha_1 + \alpha_2 = 1$.
In order to establish the privacy guarantees of $\gcn$, we need to determine the sensitivity and the Lipschitz constant of the message passing layer $\mathsf{CGL}$.

\begin{proposition}
    \label{prop:lip}
    The message passing layer $\mathsf{CGL}$ of $\gcn$ is contractive with Lipschitz constant $\Clip<1$ with respect to the input $\mX^{(k)}$, for any $k\geq 2$.
\end{proposition}

The sensitivity of the message passing layer $\mathsf{CGL}$ regarding edge-level and node-level adjacency graphs is determined as follows.
\begin{theorem}[Edge-level Sensitivity of $\mathsf{CGL}$]
    \label{thm:edge-sensitivity}
    Let $\mG$ be a graph and $D_{\min}$ be the minimum node degree of $\mG$.
    The edge level sensitivity $\Delta_e$ of the message passing layer $\mathsf{CGL}$ is
    \begin{equation}
    \label{equ:sensivity_edge}
        \begin{aligned}
            \Delta_e(\mathsf{CGL}) = \sqrt{2}\Clip\alpha_1\Big( & \frac{1}{(D_{\min} +1) (D_{\min} + 2)}+ \frac{C(D_{\min})}{\sqrt{D_{\min}+1}} \\
            & + \frac{1}{\sqrt{D_{\min}+2}\sqrt{D_{\min}+1}}\Big),
        \end{aligned}
    \end{equation}
    where $C(D_{\min})$ is a piecewise function defined as
    \begin{equation}
        C(D_{\min}) = \begin{cases}
            \frac{D_{\min}}{\sqrt{D_{\min}+1}} - \frac{D_{\min}}{\sqrt{D_{\min}+2}} & D_{\min} > 3          \\
            (\frac{3}{\sqrt{4}} - \frac{3}{\sqrt{5}})                               & 1\leq D_{\min} \leq 3 \\
        \end{cases}
        \label{eq:C_Dim_1}
    \end{equation}
\end{theorem}

Intuitively, the edge-level sensitivity $\Delta_e$ will be smaller if the minimum degree $D_{\min}$ is larger or the Lipschitz constant $\Clip$ is smaller.
The node level sensitivity $\Delta_n$ of the message passing layer $\mathsf{CGL}$ is determined as follows.
\begin{theorem}[Node-level Sensitivity of $\mathsf{CGL}$]
    \label{thm:node-sensitivity}
    Let $\mG$ be a graph and $D_{\max}$ be the maximum node degree of $\mG$.
    The node level sensitivity $\Delta_n$ of the message passing layer $\mathsf{CGL}$ is
    \begin{equation}
    \label{equ:sensivity_node}
        \begin{aligned}
            &\Delta_n(\mathsf{CGL}) =  1 + \quad \alpha_2 \Clip\frac{2|V|}{|V|+1}  +\alpha_1 \Clip \left(\frac{\sqrt{D_{\max}}}{(D_{\min} +1) (D_{\min} + 2)} \right.\\
            & \qquad \left. + \frac{C(D_{\min})\sqrt{D_{\max}}}{\sqrt{D_{\min}+1}} + \frac{1}{\sqrt{D_{\min}+2}}\right)
        \end{aligned}
    \end{equation}
    where $C(D_{\min})$ is defined as in \eqref{eq:C_Dim_1}.
\end{theorem}

From the above, we can see that the node-level sensitivity $\Delta_n$ will be smaller if the maximum degree $D_{\max}$ is smaller, the minimum degree $D_{\min}$ is larger, or the Lipschitz constant $\Clip$ is smaller.

Then as a consequence of the contractive result in Proposition~\ref{prop:lip} and the sensitivity results in Theorem~\ref{thm:edge-sensitivity} and Theorem~\ref{thm:node-sensitivity}, we can apply Theorem~\ref{thm:privacy_convergence_general} to derive the following privacy guarantee (Corollary~\ref{the:CGL_rdp}) for the message-passing layer $\mathsf{CGL}$ of $\gcn$.
Specifically, for Lipchitz constant $0.8$, a $10$-layer GNN realizes  $(\alpha,{3.62\alpha\Delta^2}/{\sigma^2})$-RDP (versus $(\alpha,{5\alpha\Delta^2}/{\sigma^2})$-RDP for \gap), meanwhile $2$-layer GNN is $(\alpha,{0.99\alpha\Delta^2}/{\sigma^2})$-RDP (versus $(\alpha,{\alpha\Delta^2}/{\sigma^2})$-RDP for \gap).
That is, more privacy costs are saved as the number of layer increases.

\begin{corollary}[DP Guarantees for $\mathsf{CGL}$ layers]
    \label{the:CGL_rdp}
    Let $\mG$ be a graph and $K$ be the number of hops ($\mathsf{CGL}$ layers) in $\gcn$.
    Then the $K$-hop message passing of $\gcn$ satisfies:

    $$\left(\alpha, \frac{\alpha}{2}\frac{\Delta^2}{\sigma^2} \min\left\{K, \frac{1-\Clip^K}{1+\Clip^K}\frac{1+\Clip}{1-\Clip}\right\}\right)\text{-RDP},$$
    
    where $\Delta = \Delta_e(\mathsf{CGL})$ from Theorem~\ref{thm:edge-sensitivity} for edge-level privacy, or $\Delta = \Delta_n(\mathsf{CGL})$ from Theorem~\ref{thm:node-sensitivity} for node-level privacy. In particular, as $K\to \infty$, it converges to
    $\left(\alpha, \frac{\alpha}{2}\frac{\Delta^2}{\sigma^2} \frac{1+\Clip}{1-\Clip}\right)\text{-RDP}.$
    Also, the RDP guarantees imply the following DP guarantees:
    $$\left(\frac{\alpha}{2}\frac{\Delta^2}{\sigma^2} \min\left\{K, \frac{1-\Clip^K}{1+\Clip^K}\frac{1+\Clip}{1-\Clip}\right\} + \frac{\log(1/\delta)}{\alpha - 1}, \delta\right)\text{-DP}.$$
\end{corollary}

For a complete analysis, we integrate the node-level privacy guarantees into an overall DP bound for the entire training process of $\gcn$. Specifically, we assume that the private DAE and CM satisfy $(\alpha, \epsilon_{\DAE}(\alpha))$-RDP and $(\alpha, \epsilon_{\CM}(\alpha))$-RDP, respectively.
By privacy composition, $\gcn$'s overall privacy budget $\epsilon_{\gcn}$ is then derived as shown in Theorem~\ref{the:overall_epsilon}.
The terms $\epsilon_{\DAE}(\alpha)$ and $\epsilon_{\CM}(\alpha)$ individually quantify the privacy contributions from the DAE and CM modules, while the remaining aggregation term accounts for node-level privacy loss during $K$-hop message passing.
Finally, the $\frac{\log(1 / \delta)}{\alpha - 1}$ term incorporates the privacy failure probability $\delta$, ensuring conventional privacy guarantees across the inference process.

\begin{theorem}[$\gcn$'s Privacy Guarantee]
    \label{the:overall_epsilon}
    For any $\alpha > 1$, let DAE training and CM training satisfy $(\alpha, \epsilon_{\DAE}(\alpha))$-RDP and $(\alpha, \epsilon_{\CM}(\alpha))$-RDP, respectively. Then, for the maximum hop $K \geq 1$, privacy failure probability $0 < \delta < 1$, Lipschitz constant $0 < \Clip < 1$, and noise variance $\sigma^2$, $\gcn$ satisfies $(\epsilon_{\gcn}, \delta)$-DP, where
       $\epsilon_{\gcn} =
         \epsilon_{\DAE}(\alpha) + \epsilon_{\CM}(\alpha) + \frac{\alpha}{2}\frac{\Delta^2}{\sigma^2} \min\left\{K, \frac{1-\Clip^K}{1+\Clip^K}\frac{1+\Clip}{1-\Clip}\right\} + \frac{\log(1/\delta)}{\alpha - 1}.$
    Here, $\Delta = \Delta_e(\mathsf{CGL})$ from Theorem~\ref{thm:edge-sensitivity} for edge-level privacy, or $\Delta = \Delta_n(\mathsf{CGL})$ from Theorem~\ref{thm:node-sensitivity} for node-level privacy.
\end{theorem}

\begin{table*}[!ht]
\small
    \caption{Comparison of Classification Accuracy for EDP and NDP. The \colorbox{colorbest}{\color{colortext}{best accuracy}} and the \colorbox{colorsecond}{\color{colortext}second-best accuracy} are highlighted, respectively. The symbol \greenup\ represents that the best accuracy improves the second-best accuracy by more than $10\%$. The symbol \reddown\  represents the accuracy less than $55\%$, close to random guess on the chain-structured datasets.
    }\label{tab:reulst_table_overall_acc_top1}
    \centering
    \setlength\tabcolsep{14pt}
    \adjustbox{max width=\textwidth}{
    \begin{tabular}{l|c|c|c|c|c|c|c|c|c|c}
    \toprule[1pt]     & \textbf{Dataset}   &  \textbf{Computers} &  \textbf{Facebook}&    \textbf{PubMed} &  \textbf{Cora}&  \textbf{Photo} & \textbf{Chain-S}  & \textbf{Chain-M}& \textbf{Chain-L} &  \textbf{Chain-X}\\
    \cmidrule[1pt]{1-11}
      \multicolumn{11}{c|}{\textbf{EDP}} \\\cmidrule[1pt]{1-11}
  \multirow{4}{*}{$\epsilon=1$}   &    $\gcn$ &  \cellcolor{colorbest}\color{colortext}$92.0\%$&  \cellcolor{colorbest}\color{colortext}$74.0\%\tbspace$&  \cellcolor{colorsecond}\color{colortext}$88.2\%$&  \cellcolor{colorbest}\color{colortext}$85.1\%\tbspace$&  \cellcolor{colorbest}\color{colortext}$95.8\%$&  \cellcolor{colorbest}\color{colortext}$84.4\%$\ \greenup&  \cellcolor{colorbest}\color{colortext}$82.5\%$\ \greenup&  \cellcolor{colorbest}\color{colortext}$70.0\%\tbspace$&  \cellcolor{colorsecond}\color{colortext}$66.0\%\tbspace$\\
    &\dpdgc    & \cellcolor{colorsecond}\color{colortext}$88.0\%$ & $60.6\%\tbspace$& \cellcolor{colorbest}\color{colortext}$88.3\%$&$75.5\%\tbspace$&$92.5\%$& $43.8\%$\ \reddown   & $50.0\%$ \reddown
     &$51.7\%$ \reddown& $39.0\%$ \reddown\\
     &\gap  &  $87.2\%$&  \cellcolor{colorsecond}\color{colortext}$68.5\% \tbspace$&  $87.4\%$&  \cellcolor{colorsecond}\color{colortext}$77.1\%\tbspace$&  \cellcolor{colorsecond}\color{colortext}$93.0\%$&  \cellcolor{colorsecond}\color{colortext}$65.6\% \tbspace$&  \cellcolor{colorsecond}\color{colortext}$67.5\% \tbspace$&  \cellcolor{colorsecond}\color{colortext}$61.7\% \tbspace$&  \cellcolor{colorbest}\color{colortext}$67.0\%\tbspace$\\
   & \sage     &$77.8\%$ & $48.2\%$ \reddown&$85.0\%$&$60.0\%\tbspace$&$82.4\%$ &  $53.1\%$\ \reddown & $45.0\%$ \reddown
     &$53.3\%$ \reddown &$51.0\%$ \reddown\\
    \cmidrule[1pt]{1-11}
      \multirow{4}{*}{$\epsilon=2$}   &    $\gcn$  &  \cellcolor{colorbest}\color{colortext}$92.1\%$&  \cellcolor{colorbest}\color{colortext}$73.8\%\tbspace$&  \cellcolor{colorbest}\color{colortext}$89.1\%$&  \cellcolor{colorbest}\color{colortext}$86.9\%\tbspace$&  \cellcolor{colorbest}\color{colortext}$95.9\%$&  \cellcolor{colorbest}\color{colortext}$90.6\%$ \greenup&  \cellcolor{colorbest}\color{colortext}$82.5\%$ \greenup&  \cellcolor{colorbest}\color{colortext}$73.3\%$ \greenup&  \cellcolor{colorbest}\color{colortext}$68.0\% \tbspace$\\
    &\dpdgc    & $88.2\%$ & $66.7\% \tbspace$&\cellcolor{colorsecond}\color{colortext}$88.2\%$& \cellcolor{colorsecond}\color{colortext}$77.5\% \tbspace$&$93.3\%$ & $43.8\%$ \reddown  & $50.0\%$ \reddown
     &$51.7\%$ \reddown &$39.0\%$ \reddown\\
     &\gap  &  \cellcolor{colorsecond}\color{colortext}$88.3\%$&  \cellcolor{colorsecond}\color{colortext}$71.7\% \tbspace$&  $87.5\%$&  \cellcolor{colorsecond}\color{colortext}$77.5\%\tbspace$&  \cellcolor{colorsecond}\color{colortext}$93.4\%$&  \cellcolor{colorsecond}\color{colortext}$65.6\% \tbspace$&  \cellcolor{colorsecond}\color{colortext}$67.5\% \tbspace$&  \cellcolor{colorsecond}\color{colortext}$61.7\% \tbspace$&  \cellcolor{colorbest}\color{colortext}$68.0\% \tbspace$\\
   & \sage   &$76.1\%$ & $48.1\%$ \reddown &$84.6\%$&$60.1\%\tbspace$&$82.4\%$  &  $43.8\%$ \reddown  & $45.0\%$ \reddown
     &$53.3\%$ \reddown&$51.0\%$ \reddown \\
     \cmidrule[1pt]{1-11}
       \multirow{4}{*}{$\epsilon=4$}   &    $\gcn$  &  \cellcolor{colorbest}\color{colortext}$92.2\%$&  \cellcolor{colorbest}\color{colortext}$74.0\% \tbspace$ &  \cellcolor{colorbest}\color{colortext}$89.5\%$&  \cellcolor{colorbest}\color{colortext}$87.3\%$ \greenup&  \cellcolor{colorbest}\color{colortext}$95.9\%$&  \cellcolor{colorbest}\color{colortext}$90.6\%$ \greenup&  \cellcolor{colorbest}\color{colortext}$82.5\%$ \greenup&  \cellcolor{colorbest}\color{colortext}$73.3\%$ \greenup&  \cellcolor{colorbest}\color{colortext}$69.0\% \tbspace$\\
    &\dpdgc    & \cellcolor{colorsecond}\color{colortext}$88.9\%$ & \cellcolor{colorsecond}\color{colortext}$73.3\% \tbspace$ &\cellcolor{colorsecond}\color{colortext}$88.4\%$&$75.1\% \tbspace$& \cellcolor{colorsecond}\color{colortext}$94.2\%$& $43.8\%$ \reddown  & $50.0\%$ \reddown
     & $51.7\%$ \reddown&$39.0\%$ \reddown\\
     &\gap    &  $87.2\%$&  $68.5\% \tbspace$&  $87.4\%$&  \cellcolor{colorsecond}\color{colortext}$77.1\% \tbspace$&  $93.0\%$&  \cellcolor{colorsecond}\color{colortext}$65.6\% \tbspace$&  \cellcolor{colorsecond}\color{colortext}$67.5\%\tbspace$&  \cellcolor{colorsecond}\color{colortext}$61.7\%\tbspace$&  \cellcolor{colorsecond}\color{colortext}$67.0\%\tbspace$\\
   & \sage   & $79.1\%$ & $50.3\%$ \reddown&$85.8\%$& $63.3\%\tbspace$& $85.7\%$ &  $50.0\%$ \reddown  & $47.5\%$ \reddown
     &$51.7\%$ \reddown &$54.0\%$ \reddown\\
   \cmidrule[1pt]{1-11} 
      \multirow{4}{*}{$\epsilon=8$}   &    $\gcn$  &  \cellcolor{colorbest}\color{colortext}$92.2\%$&  $74.4\%\tbspace$&  \cellcolor{colorbest}\color{colortext}$89.8\%$&  \cellcolor{colorbest}\color{colortext}$88.4\%\tbspace$&  \cellcolor{colorbest}\color{colortext}$96.0\%$&  \cellcolor{colorbest}\color{colortext}$93.8\%$ \greenup&  \cellcolor{colorbest}\color{colortext}$82.5\%$ \greenup&  \cellcolor{colorbest}\color{colortext}$75.0\%$ \greenup&  \cellcolor{colorbest}\color{colortext}$70.0\%\tbspace$\\
    &\dpdgc    & $89.4\%$ & \cellcolor{colorbest}\color{colortext}$78.6\%\tbspace$ &\cellcolor{colorsecond}\color{colortext}$88.6\%$&$76.0\%\tbspace$&\cellcolor{colorsecond}\color{colortext}$94.6\%$&  $43.8\%$ \reddown & $50.0\%$ \reddown
     &$51.7\%$ \reddown &$39.0\%$ \reddown\\
     &\gap  &  \cellcolor{colorsecond}\color{colortext}$90.1\%$&  $75.0\%\tbspace$&  $88.1\%$&  \cellcolor{colorsecond}\color{colortext}$79.0\%\tbspace$&  \cellcolor{colorsecond}\color{colortext}$94.6\%$  &  \cellcolor{colorsecond}\color{colortext}$65.6\%\tbspace$&  \cellcolor{colorsecond}\color{colortext}$60.0\%\tbspace$&  \cellcolor{colorsecond}\color{colortext}$61.7\%\tbspace$&  \cellcolor{colorsecond}\color{colortext}$67.0\%\tbspace$\\
   & \sage   & $87.9\%$ & \cellcolor{colorsecond}\color{colortext}$75.6\%\tbspace$&$84.8\%$&$ 75.7\%\tbspace$& $92.2\%$  &  $43.8\%$ \reddown  & $47.5\%$ \reddown
     &$51.7\%$ \reddown &$61.0\%\tbspace$\\
     \cmidrule[1pt]{1-11}\multirow{4}{*}{$\epsilon=16$}   &    $\gcn$ &  \cellcolor{colorbest}\color{colortext}$92.2\%$&  $74.5\%\tbspace$&  \cellcolor{colorbest}\color{colortext}$89.8\%$&  \cellcolor{colorbest}\color{colortext}$88.6\%\tbspace$&  \cellcolor{colorbest}\color{colortext}$96.0\%$ &  \cellcolor{colorbest}\color{colortext}$93.8\%$&  \cellcolor{colorbest}\color{colortext}$85.0\%$ \greenup&  \cellcolor{colorbest}\color{colortext}$73.3\%$ \greenup&  \cellcolor{colorbest}\color{colortext}$73.0\%\tbspace$\\
    &\dpdgc    & $90.4\%$ &  \cellcolor{colorbest}\color{colortext}$81.2\%\tbspace$&\cellcolor{colorsecond}\color{colortext}$88.9\%$&$77.7\%\tbspace$&$93.8\%$& $43.8\%$ \reddown  & $50.0\%$ \reddown
     &$51.7\%$ \reddown &$40.0\%$ \reddown\\
     &\gap  &  $90.1\%$&  $76.0\%\tbspace$&  $88.5\%$&  $81.7\%\tbspace$&  \cellcolor{colorsecond}\color{colortext}$94.6\%$ &  \cellcolor{colorsecond}\color{colortext}$62.5\%\tbspace$&  \cellcolor{colorsecond}\color{colortext}$70.0\%\tbspace$&  \cellcolor{colorsecond}\color{colortext}$61.7\%\tbspace$&  \cellcolor{colorsecond}\color{colortext}$68.0\%\tbspace$\\
   & \sage   & \cellcolor{colorsecond}\color{colortext}$90.9\%$ &\cellcolor{colorsecond}\color{colortext}$79.5\%\tbspace$&$87.6\%$&\cellcolor{colorsecond}\color{colortext}$84.7\%\tbspace$& $94.1\%$  &  $43.8\%$ \reddown  & $47.5\%$ \reddown
     & $51.7\%$ \reddown& $51.0\%$ \reddown\\
       \cmidrule[1pt]{1-11} \multirow{4}{*}{$\epsilon=32$}   &    $\gcn$ &  \cellcolor{colorbest}\color{colortext}$92.2\%$&  $74.0\%\tbspace$&  \cellcolor{colorbest}\color{colortext}$89.9\%$&  \cellcolor{colorbest}\color{colortext}$88.6\%\tbspace$&  \cellcolor{colorbest}\color{colortext}$95.9\%$&  \cellcolor{colorbest}\color{colortext}$93.8\%$ \greenup&  \cellcolor{colorbest}\color{colortext}$82.5\%$ \greenup&  \cellcolor{colorbest}\color{colortext}$78.3\%$ \greenup&  \cellcolor{colorbest}\color{colortext}$73.0\%\tbspace$\\
    &\dpdgc   & \cellcolor{colorsecond}\color{colortext}$91.3\%$ & \cellcolor{colorbest}\color{colortext}$82.9\%\tbspace$ &\cellcolor{colorsecond}\color{colortext}$88.8\%$&$79.9\%\tbspace$&$94.3\%$ & $43.8\%$ \reddown  & $50.0\%$ \reddown
     & $51.7\%$ \reddown&$40.0\%$ \reddown\\
     &\gap    &  $90.6\%$&  $77.0\%\tbspace$&  \cellcolor{colorsecond}\color{colortext}$88.8\%$&  $82.5\%\tbspace$&  \cellcolor{colorsecond}\color{colortext}$94.8\%$&  \cellcolor{colorsecond}\color{colortext}$59.4\%\tbspace$&  \cellcolor{colorsecond}\color{colortext}$65.0\%\tbspace$&  \cellcolor{colorsecond}\color{colortext}$61.7\%\tbspace$&  \cellcolor{colorsecond}\color{colortext}$68.0\%\tbspace$\\
   & \sage   & $90.6\%$ & \cellcolor{colorsecond}\color{colortext}$79.9\%\tbspace$ &$86.9\%$&\cellcolor{colorsecond}\color{colortext}$85.2\%\tbspace$&$94.4\%$ &  $43.8\%$ \reddown & $47.5\%$ \reddown
     & $51.7\%$ \reddown&$51.0\%$ \reddown\\
     \cmidrule[1pt]{1-11} 
          \multicolumn{11}{c|}{\textbf{NDP (max node degree $=20$)}} \\\cmidrule[1pt]{1-11}
   \multirow{5}{*}{$\epsilon=1$}     
    & $\gcn$ & \cellcolor{colorbest}\color{colortext}$91.91\%$ \greenup & \cellcolor{colorbest}\color{colortext}$64.47\%$ \greenup & \cellcolor{colorsecond}\color{colortext}$72.94\%\tbspace$ & \cellcolor{colorbest}\color{colortext}$80.81\%$ \greenup & \cellcolor{colorbest}\color{colortext}$94.83\%$ \greenup & \cellcolor{colorbest}\color{colortext}$78.12\%$ \greenup & \cellcolor{colorbest}\color{colortext}$85.00\%$ \greenup & \cellcolor{colorbest}\color{colortext}$68.33\%$ \greenup & \cellcolor{colorbest}\color{colortext}$66.00\%\tbspace$ \\
    & \dpdgc & $56.72\%\tbspace$ & $39.09\%$ \reddown & $59.12\%\tbspace$ & $33.58\%$ \reddown & $46.39\%$ \reddown & $58.06\%\tbspace$ & $57.50\%\tbspace$ & $57.63\%\tbspace$ & $54.55\%$ \reddown \\
    & \gap & $36.71\%$ \reddown & $35.07\%$ \reddown & $55.06\%\tbspace$ & $33.95\%$ \reddown & $30.88\%$ \reddown & \cellcolor{colorsecond}\color{colortext}$65.62\%\tbspace$ & $55.0\%$ \reddown & \cellcolor{colorsecond}\color{colortext}$58.33\%\tbspace$ & \cellcolor{colorsecond}\color{colortext}$59.00\%\tbspace$ \\
    & \sage & $29.51\%$ \reddown & $20.73\%$ \reddown & $39.51\%$ \reddown & $18.45\%$ \reddown & $21.07\%$ \reddown & $59.38\%\tbspace$ & \cellcolor{colorsecond}\color{colortext}$60.00\%$ & $56.67\%\tbspace$ & $55.00\%$ \reddown \\
     & \mlpE & \cellcolor{colorsecond}\color{colortext}$63.44\%\tbspace$ & \cellcolor{colorsecond}\color{colortext}$45.39\%$ \reddown& \cellcolor{colorbest}\color{colortext}$73.52\%\tbspace$ & \cellcolor{colorsecond}\color{colortext}$34.13\%$ \reddown& \cellcolor{colorsecond}\color{colortext}$59.44\%\tbspace$ & $56.25\%\tbspace$ & $55.0\%$ \reddown & $55.0\%$ \reddown & $51.0\%$ \reddown \\
    \cmidrule[1pt]{1-11}
  \multirow{5}{*}{$\epsilon=2$}     
    & $\gcn$ & \cellcolor{colorbest}\color{colortext}$92.24\%$ \greenup & \cellcolor{colorbest}\color{colortext}$68.10\%$ \greenup& \cellcolor{colorsecond}\color{colortext}$79.86\%$ \greenup & \cellcolor{colorbest}\color{colortext}$83.39\%$ \greenup & \cellcolor{colorbest}\color{colortext}$95.10\%$ \greenup & \cellcolor{colorbest}\color{colortext}$84.38\%$ \greenup & \cellcolor{colorbest}\color{colortext}$85.00\%\tbspace$  & \cellcolor{colorbest}\color{colortext}$70.00\%$ \greenup & \cellcolor{colorbest}\color{colortext}$65.00\%\tbspace$ \\
    & \dpdgc & $66.07\%\tbspace$ & $45.37\%$ \reddown & $68.60\%\tbspace$ & $31.92\%$ \reddown & $57.92\%\tbspace$ & $58.06\%\tbspace$ & $57.50\%\tbspace$ & $57.63\%\tbspace$ & $54.55\%$ \reddown \\
    & \gap & $47.92\%$ \reddown & $39.97\%$ \reddown & $67.82\%\tbspace$ & $33.39\%$ \reddown & $36.05\%$ \reddown & \cellcolor{colorsecond}\color{colortext}$65.62\%\tbspace$ & $55.00\%$ \reddown & \cellcolor{colorsecond}\color{colortext}$58.33\%\tbspace$ & \cellcolor{colorsecond}\color{colortext}$59.00\%\tbspace$ \\
    & \sage & $34.30\%$ \reddown & $21.43\%$ \reddown & $39.72\%$ \reddown & $19.93\%$ \reddown & $23.46\%$ \reddown & $59.38\%\tbspace$ &  \cellcolor{colorsecond}\color{colortext}$60.00\%\tbspace$ & $56.67\%\tbspace$ & $55.00\%$ \reddown \\
       & \mlpE & \cellcolor{colorsecond}\color{colortext}$69.56\%\tbspace$ & \cellcolor{colorsecond}\color{colortext}$47.95\%$ \reddown & \cellcolor{colorbest}\color{colortext}$80.6\%\tbspace$ & \cellcolor{colorsecond}\color{colortext}$37.82\%$ \reddown & \cellcolor{colorsecond}\color{colortext}$72.96\%\tbspace$ & $56.25\%\tbspace$ & $55.0\%$ \reddown & $55.0\%$ \reddown & $51.0\%$ \reddown \\
    \cmidrule[1pt]{1-11}
  \multirow{5}{*}{$\epsilon=4$}     
    & $\gcn$ & \cellcolor{colorbest}\color{colortext}$92.24\%$ \greenup & \cellcolor{colorbest}\color{colortext}$70.83\%$ \greenup & \cellcolor{colorbest}\color{colortext}$84.66\%\tbspace$ & \cellcolor{colorbest}\color{colortext}$85.61\%$ \greenup & \cellcolor{colorbest}\color{colortext}$95.36\%$ \greenup & \cellcolor{colorbest}\color{colortext}$90.62\%$ \greenup & \cellcolor{colorbest}\color{colortext}$82.50\%\tbspace$  & \cellcolor{colorbest}\color{colortext}$71.67\%$ \greenup & \cellcolor{colorbest}\color{colortext}$69.00\%$ \greenup  \\
    & \dpdgc & $72.35\%\tbspace$ & $48.81\%$ \reddown & $79.71\%\tbspace$ & $32.10\%$ \reddown & $73.49\%\tbspace$ & $58.06\%\tbspace$ & $57.50\%\tbspace$ & $57.63\%\tbspace$ & $54.55\%$ \reddown \\
    & \gap & $61.84\%\tbspace$ & $47.02\%$ \reddown & $79.33\%\tbspace$ & $33.39\%$ \reddown & $45.33\%$ \reddown & \cellcolor{colorsecond}\color{colortext}$65.62\%\tbspace$ & $55.00\%$ \reddown & \cellcolor{colorsecond}\color{colortext}$58.33\%\tbspace$ &  \cellcolor{colorsecond}\color{colortext}$59.00\%\tbspace$ \\
    & \sage & $36.12\%$ \reddown & $23.01\%$ \reddown & $40.30\%$ \reddown & $21.59\%$ \reddown & $25.78\%$ \reddown & $59.38\%\tbspace$ & \cellcolor{colorsecond}\color{colortext}$60.00\%\tbspace$ & $56.67\%\tbspace$ & $55.00\%$ \reddown \\
        & \mlpE & \cellcolor{colorsecond}\color{colortext}$73.94\%\tbspace$ & \cellcolor{colorsecond}\color{colortext}$49.2\%$ \reddown& \cellcolor{colorsecond}\color{colortext}$83.13\%\tbspace$ & \cellcolor{colorsecond}\color{colortext}$50.18\%$ \reddown& \cellcolor{colorsecond}\color{colortext}$78.99\%\tbspace$ & $56.25\%\tbspace$ & $55.0\%$ \reddown & $55.0\%$ \reddown & $51.0\%$ \reddown\\
    \cmidrule[1pt]{1-11}
  \multirow{5}{*}{$\epsilon=8$}     
    & $\gcn$ & \cellcolor{colorbest}\color{colortext}$92.39\%$ \greenup & \cellcolor{colorbest}\color{colortext}$72.21\%$ \greenup & \cellcolor{colorbest}\color{colortext}$87.12\%\tbspace$ & \cellcolor{colorbest}\color{colortext}$87.08\%$ \greenup & \cellcolor{colorbest}\color{colortext}$95.29\%$ \greenup & \cellcolor{colorbest}\color{colortext}$90.62\%$ \greenup & \cellcolor{colorbest}\color{colortext}$87.50\%\tbspace$  & \cellcolor{colorbest}\color{colortext}$73.33\%$ \greenup & \cellcolor{colorbest}\color{colortext}$69.00\%$ \greenup  \\
    & \dpdgc & $76.50\%\tbspace$ & $49.64\%$ \reddown & $83.39\%\tbspace$ & $43.54\%$ \reddown & $79.66\%\tbspace$ & $58.06\%\tbspace$ & $57.50\%\tbspace$ & $57.63\%\tbspace$ & \cellcolor{colorsecond}\color{colortext}$59.60\%\tbspace$ \\
    & \gap & $68.52\%\tbspace$ & $48.33\%$ \reddown & $82.35\%\tbspace$ & $31.55\%$ \reddown & $68.46\%\tbspace$ & \cellcolor{colorsecond}\color{colortext}$65.62\%\tbspace$ & $55.00\%$ \reddown & \cellcolor{colorsecond}\color{colortext}$58.33\%\tbspace$ & $59.00\%\tbspace$ \\
    & \sage & $37.38\%$ \reddown & $24.53\%$ \reddown & $42.07\%$ \reddown & $26.38\%$ \reddown & $28.36\%$ \reddown & $59.38\%\tbspace$ & \cellcolor{colorsecond}\color{colortext}$60.00\%\tbspace$ & $56.67\%\tbspace$ & $55.00\%$ \reddown \\
      & \mlpE & \cellcolor{colorsecond}\color{colortext}$77.65\%\tbspace$ & \cellcolor{colorsecond}\color{colortext}$49.96\%$ \reddown& \cellcolor{colorsecond}\color{colortext}$84.66\%\tbspace$ & \cellcolor{colorsecond}\color{colortext}$61.81\%\tbspace$ & \cellcolor{colorsecond}\color{colortext}$83.96\%\tbspace$ & $56.25\%\tbspace$ & $55.0\%$ \reddown & $55.0\%$ \reddown& $51.0\%$ \reddown\\
    \cmidrule[1pt]{1-11}
  \multirow{5}{*}{$\epsilon=16$}     
    & $\gcn$ & \cellcolor{colorbest}\color{colortext}$92.43\%$ \greenup & \cellcolor{colorbest}\color{colortext}$72.29\%$ \greenup & \cellcolor{colorbest}\color{colortext}$88.41\%\tbspace$ & \cellcolor{colorbest}\color{colortext}$88.01\%$ \greenup & \cellcolor{colorbest}\color{colortext}$95.29\%\tbspace$  & \cellcolor{colorbest}\color{colortext}$90.62\%$ \greenup & \cellcolor{colorbest}\color{colortext}$82.50\%$ \greenup & \cellcolor{colorbest}\color{colortext}$75.00\%$ \greenup & \cellcolor{colorbest}\color{colortext}$70.00\%$ \greenup \\
    & \dpdgc & $78.47\%\tbspace$ & \cellcolor{colorsecond}\color{colortext}$50.85\%$ \reddown & $85.32\%\tbspace$ & $56.09\%\tbspace$ & $83.23\%\tbspace$ & $58.06\%\tbspace$ & $57.50\%\tbspace$ & $57.63\%\tbspace$ & \cellcolor{colorsecond}\color{colortext}$59.60\%\tbspace$ \\
    & \gap & $73.94\%\tbspace$ & $49.98\%$ \reddown & $83.67\%\tbspace$ & $37.45\%$ \reddown & $76.74\%\tbspace$ & \cellcolor{colorsecond}\color{colortext}$65.62\%\tbspace$ & $55.00\%$ \reddown & \cellcolor{colorsecond}\color{colortext}$58.33\%\tbspace$ & $59.00\%\tbspace$ \\
    & \sage & $40.13\%$ \reddown & $26.88\%$ \reddown & $44.33\%$ \reddown & $27.68\%$ \reddown & $33.47\%$ \reddown & $59.38\%\tbspace$ & \cellcolor{colorsecond}\color{colortext}$60.00\%\tbspace$ & $56.67\%\tbspace$ & $55.00\%$ \reddown \\
    & \mlpE & \cellcolor{colorsecond}\color{colortext}$80.25\%\tbspace$ & $49.92\%$ \reddown& \cellcolor{colorsecond}\color{colortext}$85.57\%\tbspace$ & \cellcolor{colorsecond}\color{colortext}$67.53\%\tbspace$ & \cellcolor{colorsecond}\color{colortext}$87.34\%\tbspace$ & $56.25\%\tbspace$ & $55.0\%$ \reddown & $55.0\%$ \reddown & $51.0\%$ \reddown\\
    \cmidrule[1pt]{1-11}
  \multirow{5}{*}{$\epsilon=32$}     
    & $\gcn$ & \cellcolor{colorbest}\color{colortext}$92.39\%$ \greenup & \cellcolor{colorbest}\color{colortext}$73.33\%$ \greenup & \cellcolor{colorbest}\color{colortext}$88.92\%\tbspace$ & \cellcolor{colorbest}\color{colortext}$87.82\%$ \greenup & \cellcolor{colorbest}\color{colortext}$95.36\%\tbspace$ & \cellcolor{colorbest}\color{colortext}$87.50\%$ \greenup & \cellcolor{colorbest}\color{colortext}$82.50\%$ \greenup & \cellcolor{colorbest}\color{colortext}$75.00\%$ \greenup & \cellcolor{colorbest}\color{colortext}$70.00\%$ \greenup \\
    & \dpdgc & $81.33\%\tbspace$ & \cellcolor{colorsecond}\color{colortext}$51.19\%$ \reddown & \cellcolor{colorsecond}\color{colortext}$85.90\%\tbspace$ & $64.02\%\tbspace$ & $86.94\%\tbspace$ & $58.06\%\tbspace$ & $57.50\%\tbspace$ & $57.63\%\tbspace$ & $58.59\%\tbspace$ \\
    & \gap & $76.95\%\tbspace$ & $50.64\%$ \reddown & $85.04\%\tbspace$ & $57.01\%\tbspace$ & $80.25\%\tbspace$ & \cellcolor{colorsecond}\color{colortext}$65.62\%\tbspace$ & $55.00\%$ \reddown & \cellcolor{colorsecond}\color{colortext}$58.33\%\tbspace$ & \cellcolor{colorsecond}\color{colortext}$59.00\%\tbspace$ \\
    & \sage & $46.40\%$ \reddown & $31.30\%$ \reddown & $47.53\%$ \reddown & $29.89\%$ \reddown & $41.68\%$ \reddown & $59.38\%\tbspace$ & \cellcolor{colorsecond}\color{colortext}$60.00\%\tbspace$ & $56.67\%\tbspace$ & $55.00\%$ \reddown \\
    & \mlpE & \cellcolor{colorsecond}\color{colortext}$81.89\%\tbspace$ & $50.27\%$ \reddown & $85.87\%\tbspace$ & \cellcolor{colorsecond}\color{colortext}$70.85\%\tbspace$ & \cellcolor{colorsecond}\color{colortext}$89.07\%\tbspace$ & $56.25\%\tbspace$ & $55.0\%$ \reddown& $55.0\%$ \reddown& $51.0\%$ \reddown \\
         \cmidrule[1pt]{1-11} 
          \multicolumn{11}{c}{\textbf{Non-Private}} \\
         \cmidrule[1pt]{1-11} 
        &    $\gcn$  
     &  \cellcolor{colorsecond}\color{colortext}$92.4\%\tbspace$&  $79.0\%\tbspace$&  \cellcolor{colorbest}\color{colortext}$89.9\%\tbspace$&  \cellcolor{colorbest}\color{colortext}$89.3\%\tbspace$&  \cellcolor{colorsecond}\color{colortext}$96.0\%\tbspace$&  \cellcolor{colorbest}\color{colortext}$100.0\%\tbspace$&  \cellcolor{colorbest}\color{colortext}$100.0\%\tbspace$&  \cellcolor{colorbest}\color{colortext}$100.0\%\tbspace$&  \cellcolor{colorbest}\color{colortext}$100.0\%\tbspace$\\
   &\dpdgc   & \cellcolor{colorbest}\color{colortext}$92.8\%\tbspace$ & \cellcolor{colorbest}\color{colortext}$86.4\%\tbspace$ &$88.1\%\tbspace$&$83.9\%\tbspace$& \cellcolor{colorbest}\color{colortext}$96.2\%\tbspace$  & $59.4\%\tbspace$  & $77.5\%\tbspace$ 
     &$63.3\%\tbspace$ &$73.0\%\tbspace$\\
 Plain     &\gap    &  $91.0\%\tbspace$&  $79.0\%\tbspace$&  \cellcolor{colorsecond}\color{colortext}$89.3\%\tbspace$&  \cellcolor{colorsecond}\color{colortext}$85.2\%\tbspace$&  $95.5\%\tbspace$&  \cellcolor{colorbest}\color{colortext}$100.0\%\tbspace$&  \cellcolor{colorbest}\color{colortext}$100.0\%\tbspace$&  \cellcolor{colorbest}\color{colortext}$100.0\%\tbspace$&  \cellcolor{colorbest}\color{colortext}$100.0\%\tbspace$\\
  $\epsilon=\infty$  & \sage    & $91.6\%\tbspace$ & \cellcolor{colorsecond}\color{colortext}$79.7\%\tbspace$&$87.0\%\tbspace$&$82.1\%\tbspace$&$94.2\%\tbspace$ &$59.4\%\tbspace$    & $55.0\%$ \reddown
     & $60.0\%\tbspace$&$58.0\%\tbspace$ \\
   & \mlpE & $85.82\%\tbspace$ & $51.35\%$ \reddown & $87.45\%\tbspace$ & $75.83\%\tbspace$ & $91.98\%\tbspace$ & $59.38\%\tbspace$ & $55.0\%$ \reddown & $53.33\%$ \reddown & $51.0\%$ \reddown \\
    \bottomrule
    \end{tabular}
    }
\end{table*}


\section{Experimental Evaluation}
\label{sec:exp}

The empirical evaluation includes privacy-utility trade-offs, privacy audits, ablation studies of hops and hyper-parameters, and computational overhead.

\textit{Datasets.}
$\gcn$ was tested over nine graph datasets.
Five of the datasets have been broadly used to evaluate GNN methods, including \photo\ and \computers\ \citep{corr/abs-1811-05868}, \cora\ and \pubmed\ \citep{sen2008collective}, \facebook\ \citep{corr/abs-1102-2166}. We also adjusted the synthetic chain-structured dataset developed under \cite{nips/GuC0SG20} into various scales, termed \chainS, \chainM, \chainL\ and \chainX.
The chain-structured dataset has been used to understand the relations between privacy/utility and hops, as described in Section~\ref{subsec:gap-chain}. 
It is considered as an important benchmark to assist the development of long-range interaction GNNs by the ML community~\cite{nips/DwivediRGPWLB22}.
More details on datasets, model configurations, and privacy configurations are specialized in Appendix~\ref{app:dataset}.

\textit{Baselines.}
We compare $\gcn$ with multiple baselines~\cite{nips/chien2023differentially,uss/sajadmanesh2023gap,sp/0011L0022,arxiv/DaigavaneNode21,ccs/KolluriBHS22} about edge-level private GNNs and vanilla GNNs. 
To our best knowledge, \gap~\cite{uss/sajadmanesh2023gap} and \dpdgc~\cite{nips/chien2023differentially}  are the strongest baselines for perturbed message passing under Gaussian mechanisms in both
edge/node-level DP GNNs. 
In addition, we consider another research direction, i.e.,  first perturbing the graph through randomized response and then training GNNs over the perturbed graph~\sage~\cite{sp/0011L0022}. 
For a comprehensive evaluation, we adopt both research lines of works as our baselines. 
\mlpE\ is a baseline commonly compared with GNNs to demonstrate how well GNNs utilize graph structures, as it provides a reference counterpart to which GNNs learn the only node features without graph topology. 
More details of the baseline methods are provided in Appendix~\ref{app:baseline}.

\textit{Computational overhead.} We elaborate the overhead analysis, including latency and memory, in Appendix~\ref{app:overhead}.

\subsection{Trade-offs of Privacy and Accuracy}
\label{subsec:tradeoff}

We first compare  $\gcn$ with the baseline methods on all datasets for their downstream classification tasks and report top GNN model utility of both EDP and NDP.
We run each model $3$ times for each group of hyper-parameters, reporting the top classification accuracy in Table~\ref{tab:reulst_table_overall_acc_top1} and
the mean accuracy over the 3 runs in Table~\ref{tab:reulst_table_overall_acc_mean} of Appendix~\ref{app:endp_supple}.
For the experiments about NDP, we set the max node degree $D_{\max}$  to $20$, following the experiment setup of \gap.
We leave more reference results (e.g., different maximum node degrees) in Appendix~\ref{app:endp_supple}.

Regarding the standard graph datasets, for \computers, \pubmed, \cora\ and \photo, $\gcn$ can outperform all the other baselines in most cases with varying $\epsilon$, no matter for EDP or NDP. 
As established in Theorem~\ref{thm:node-sensitivity}, NDP requires injecting more noise compared to EDP under the same privacy budget, hence, the accuracy of NDP is often lower than EDP for standard datasets. 
In particular, $\gcn$ is the only framework that can surpass \mlpE\  of most cases in NDP settings, showing effective GNN learning over structural graphs.
We leave more detailed analysis on EDP and NDP in Appendix~\ref{app:endp_supple} and ablation study for different max node degrees in Section~\ref{sec:abla_degreenode}.

For standard benchmark datasets with informative node features, the utility of our model approaches that of non-private methods as the privacy  $\epsilon$ increases. 
For chain-structured graphs,  the learning task primarily relies on the underlying graph topology, which is more challenging.
Accordingly, model utility is more sensitive to the added noise realized by perturbed message passing.
{This sensitivity is due to their sparse chain structure: non-zero features are present only at the first node of each chain. Information must propagate from this source, and it can be degraded by noise accumulation during propagation.
In this case if the small training set, by chance, contains an imbalanced selection of nodes (e.g., sampling nodes only near the end of a chain, far from the feature source), the task becomes significantly more difficult. This can lead to higher variance in results. 
To ensure a fair comparison, we use the exact same data split  for all models mentioned above.
}

{In Appendix~\ref{subsec:larger}, we assess whether $\gcn$ scales with larger real-world datasets, e.g., Reddit2, and report our results.
}

 \begin{find}
\label{ref:find1}
  $\gcn$  achieves a more favorable privacy-utility trade-off than other baselines across standard graph datasets, chain-structured datasets with various parameter settings.
\end{find}

\subsection{Privacy Auditing}
\label{subsec:auditing}

{
Following the PAM outlined in Section~\ref{sec:privacy_veri}, we focus on black-box, membership-based privacy audits that match the theoretical guarantees of our DP mechanisms and those of prior perturbed message-passing methods~\cite{nips/chien2023differentially,uss/sajadmanesh2023gap}.
Under this threat model, LinkTeller~\cite{sp/Wulinkteller22} and G-MIA~\cite{tpsisa/OlatunjiNK21} are canonical and complementary:
LinkTeller targets edge-level membership by asking whether a specific edge exists in the training graph; G-MIA targets node-level membership by deciding whether a node and its connected edges were used during training.
We adopt G-MIA's attacking settings of TSTF, where models have been trained on subgraph and tested on full graph.
The adversary knows the whole graph $\mG$ and all edges contained in $\mG$ but has no access to the subgraphs used for early training.
Both attacks (i) are specifically designed for GNNs, (ii) operate in the transductive setting considered in our analysis, (iii) require only query access to GNN models, and (iv) are publicly available and already used to evaluate DP-GNN defenses.
This makes them ideal choices as mechanism-level auditing tools in $\gcn$.}

In Table~\ref{tab:privacy_audit} (in Appendix), we report the AUC  (Area Under the Curve) score about the attack effectiveness, when $\gcn$ is being attacked. AUC is a major metric to evaluate the  membership inference attack~\cite{sp/Wulinkteller22,he2021stealing}.
Specifically, AUC measures true positive rate against the false positive rate on various classification thresholds, 
and a score of $0.5$ suggesting random guessing.
We found $\gcn$ is very effective against LinkTeller, by dropping the attack AUC from between 0.86 to 0.998 across all standard datasets ($\epsilon = \inf$) to less than 0.500 ($\epsilon$ ranges from 1 to 32). 
The similar effect was also observed by Wu et al. (e.g., less than 0.5 attack precision for 3-layer GCN for high density belief, shown in Table IX)~\cite{sp/Wulinkteller22} and Tang et al. (e.g., less than 0.5 attack AUC sometimes in Figure 10)~\cite{popets/TangNAA24}.
For G-MIA, its AUC on the 5 datasets are already lower than LinkTeller by a notable margin when $\epsilon = \inf$ (between 0.567 to 0.702 for the 5 datasets), so the impact of $\gcn$ is relatively small. But we observe on \cora, $\gcn$ is able to drop AUC from 0.645 to 0.500.

\begin{find}
 In privacy auditing, $\gcn$'s shows effective resistance to membership inference attacks.
\end{find}

\subsection{Ablation Study}

\subsubsection{Impact of $K$}
Both $\gcn$ and \gap\ perform $K$-hop aggregations under $K$ aggregation layers. Here we evaluate the impact of $K$ on accuracy on the chain-structured datasets (\chainS, \chainM\ and \chainL), as their classification results highly depend on long-range interactions.  In Figure~\ref{fig:k-study} and Figure~\ref{fig:K-study} of Appendix, we compare $\gcn$ and \gap\ on varying $\epsilon$ and varying $K$, respectively. The result of \gap\ is drawn with lines and the result of $\gcn$ is illustrated with the colored boxes, because $\gcn$ also depends on other hyper-parameters  $\Clip, \alpha_1,\beta$ and we use the colored boxes to represent the interquartile range over 5 runs of their different value combinations.

In Figure~\ref{fig:acc5_chains1_e1}, we show the result of one setting ($\epsilon=4$ on \chainM), and $\gcn$ achieves higher accuracy at \textit{every} $K$. This observation is consistent in other settings as shown in Figure~\ref{fig:k-study}.
In addition,  the highest accuracy happens at $K=10$ (close to the number of nodes per chain) for $\gcn$, and the classification accuracy fluctuates when $K$ varies for both  $\gcn$  and \gap.
 Across all datasets, \gap’s accuracy degrades monotonically with depth, consistent with its privacy noise variance growing linearly in the number of layers ($\sigma^2 \propto K$). In contrast, $\gcn$ benefits from additional depth and then plateaus, owing to its convergent privacy cost with respect to depth. This behavior shows that $\gcn$ can leverage deeper architectures and realize the contractive privacy amplification guaranteed by our analysis.

\begin{figure}[!t]

		\subfloat[Different $K$ when $\epsilon=4$]{
		\includegraphics[width = 3.8cm]{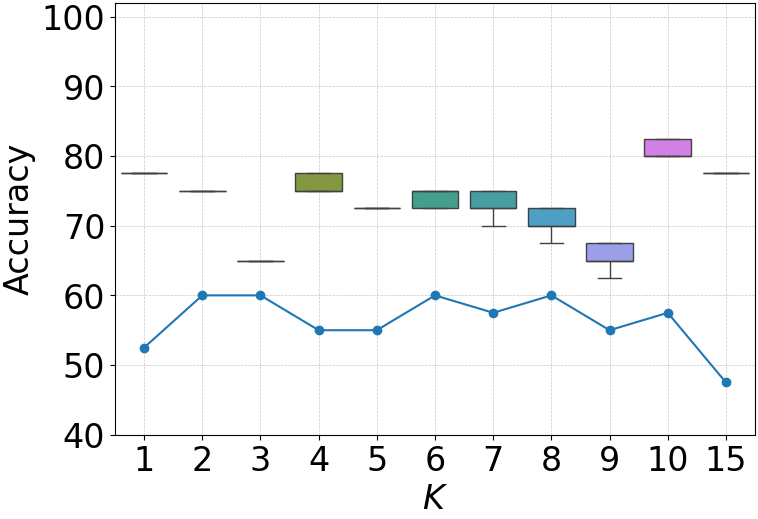}
		\label{fig:acc5_chains1_e1}
		}
      \subfloat[Different $\epsilon$ when $K=10$]{
		\includegraphics[width = 3.8cm]{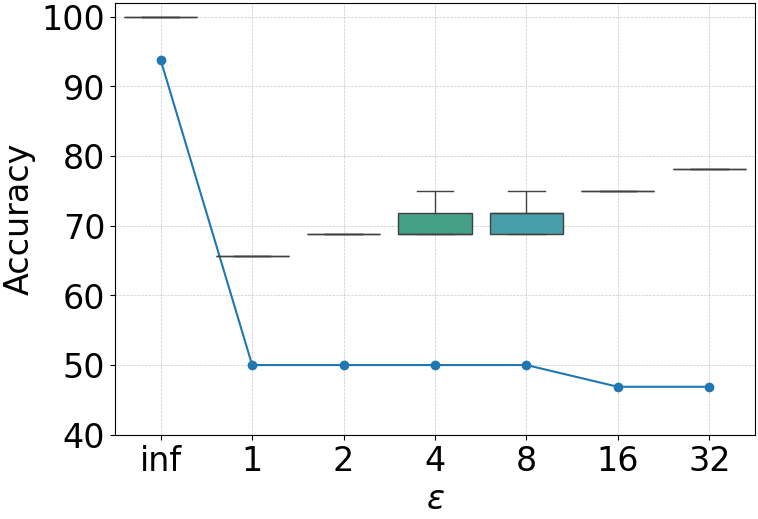}
		\label{fig:acc5_chains1_hops10}
		}
		\caption{Comparison between $\gcn$ (colored boxes) and \gap\ (blue lines) for ablation study. }
        \label{fig:acc5_chains1_e_all}
        \vspace{-5mm}
		\end{figure}

\subsubsection{Impact of $\epsilon$}
We assess how $\epsilon$ affects the performance of $\gcn$ and \gap, by flipping $K$ and $\epsilon$ from the previous ablation study.
Specifically, 
we evaluate the three chain-structured datasets (\chainS, \chainM\ and \chainL), and for each dataset, we use the same $K$ for both $\gcn$ and \gap, and then change $\epsilon$ from $1$ to $32$.

We present the full results in Figure~\ref{fig:K-study} in Appendix and one setting in Figure~\ref{fig:acc5_chains1_hops10} ($K=10$ on \chainS). Figure~\ref{fig:non-priv} in Appendix shows their non-private versions.
When $K$ is much less than the number of nodes per chain, \eg, $K=1$ vs. 8-nodes chain,  $\gcn$ and \gap\ cannot realize satisfactory accuracy (both under 75\%) even for non-private settings, as features from distant nodes cannot be effectively learnt. 
If $K$ is near to or larger than the number of nodes per chain (e.g., when $K=10$ for \chainS, as shown in Figure~\ref{fig:acc5_chains1_hops10}), though both \gap\ and $\gcn$ see very high accuracy for non-private mode, the accuracy of \gap\ drops to 50\% at $\epsilon=1$ and further decreases with increased $\epsilon$, suggesting the noise magnitude are not properly controlled. On the other hand, $\gcn$ sees steady growth of accuracy along with increased $\epsilon$, which is a desired outcome for privacy protection. 

\begin{find}
  Private GNNs face fundamental trade-offs between privacy, utility, and model depth $K$.
Model utility becomes more susceptible to the DP noise if  GNNs are tightly coupled with the underlying graph structure.
 \end{find}


{
\subsubsection{Noise scaling with depth}
Here we analyze how the noise scale changes with the number of layers $K$ under standard linear composition (Corollary~\ref{cor:k-layer-privacy}) and under  convergent privacy analysis (Theorem~\ref{thm:privacy_convergence_general}).
For a general study, we remove the effect of node degree $D_{\min},D_{\max}$ derived from a particular dataset and fix the target DP parameters $\epsilon,\delta$.  
In this case, the dependence of the calibrated noise on depth is governed by depth $K$. 
To make this comparison concrete, we instantiate a representative setting by normalizing  the sensitivity: $\Delta^2(\mathrm{MP}) = 1, \alpha = 6, \Clip = 0.9.$
Setting $\Delta^2(\mathrm{MP})=1$ removes a common multiplicative factor and highlights 
the qualitative dependence on $K$. The choice $\alpha=6$ is  common
and simplifies the expressions, while $\Clip=0.9$ represents a standard contractive layer.

\begin{table}[H]
\centering
\caption{
Noise Scale $\sigma$ under Different $K$.
We set $\Delta^2(\mathrm{MP})=1,\alpha=6,\Clip=0.9$ and $ \epsilon=4,\delta=0.001$.
}
\label{tab:noise-theory-compare}
\setlength\tabcolsep{5pt}
\begin{tabular}{c|cccccccc}
\toprule
$K$ & $1$ & $2$ & $4$ & $8$ & $16$ & $32$&$64$&$128$ \\
\midrule
Linear 
& $1.07$& $1.52$ & $2.15$ & $3.04$ & $4.30$ & $6.07$ & $8.56$ & $12.11$\\
Convergent
& $1.07$ & $1.52$ & $2.14$ & $3.00$ & $4.07$ & $4.66$ &$4.66$ & $4.67$ \\
\bottomrule
\end{tabular}
\end{table}
Table~\ref{tab:noise-theory-compare} reports the proportional values of $\sigma$ 
for several representative depths $K$, assuming the same target privacy budget $\epsilon=4$.
Under linear composition, the required  $\sigma$ grows proportionally to $K$,  becoming large for deep GNNs.
In contrast, under $\gcn$'s analysis, the required $\sigma$ grows from $1$ to a bounded constant (here approximately $4.7$) and then saturates. This study highlights $\gcn$ can support  deep architectures without unbounded noise growth.
In addition, a failure case of divergent noise allocation has been shown in Appendix~\ref{app:exp_failure}.
}

\subsubsection{hyper-parameters related to contractiveness}
We studied the impact of hyper-parameters $\Clip, \alpha_1,\beta$ in CGL. In Appendix, we draw Figure~\ref{fig:heatmaps} of classification accuracy using \chainS  and \cora.
As $\Clip$ constrains the features learned at each aggregation, in the relatively weak privacy guarantee ($\epsilon=16, 32$),  Figures~\ref{fig:heatmap_chains1_epsilon_cl} and \ref{fig:heatmap_cora_epsilon_cl} empirically confirms that the accuracy improves with $\Clip$ increases. 
In contrast, for strong privacy guarantee ($\epsilon=1, 2$), larger $\Clip$ reduces the model accuracy due to the accumulated large noise. 
{Small $\Clip$ enforces strong contraction, accelerating privacy convergence and reducing effective sensitivity, but overly small values may reduce expressive power.  
Larger $\Clip$ increases representational capacity but slows contraction and slightly increases noise amplification. }

Figures~\ref{fig:heatmap_chains1_epsilon_alpha} and \ref{fig:heatmap_cora_epsilon_alpha} describe the ratio ($\alpha_1$) of learning from the graph, where $\alpha_1=1.0 (\alpha_2=0.0)$ means the information from adjacent matrix is utilized at the maximum degree.
Larger $\alpha_1$ leads to higher accuracy across varying $\epsilon$ in general, suggesting $\gcn$ is able to achieve good balance between graph connectivity and privacy.
The impact of $\beta$, which decides the power of residual connection between node features and CGL, is different on the two datasets.
Since $\chainS$ is designed to tailor graph topology over node features, increasing $\beta$ to a large value (e.g., 15) might hurt accuracy. For \cora\ with rich node features, the model accuracy is generally increased along with $\beta$.

\begin{find}
All parameters $\sigma,\Clip,\beta,\alpha_1$ in perturbed CGL contribute to the privacy-utility trade-off.
\end{find}

\smallskip
\subsubsection{Impact of $D_{\max}, D_{\min}$}
\label{sec:abla_degreenode}
Figure~\ref{fig:acc_nodedegree_e2} shows an example ($\epsilon=2$) of the classification accuracy of NDP under different maximum node degrees, and more ablation study results are deferred to  Figure~\ref{fig:abla_nodedeg}.
As shown in Figure~\ref{fig:acc_nodedegree_e2_photo_main} on the \photo\ dataset, $\gcn$ consistently realizes the highest accuracy under different maximum node degrees ranging from $5$ to $100$.
Improving maximum node degree for \dpdgc\ can slightly increase the classification accuracy when maximum node degree is $20$, while the number of maximum node degree does not help for \sage\ and \gap.
For the \chainS\ dataset in Figure~\ref{fig:acc_nodedegree_e2_chains1_main}, classification accuracy of $\gcn$ is improved when the maximum node degree is increased and relatively small.
In addition, $\gcn$ outperforms baseline works significantly, \ie, approximated $20\%$-$25\%$ higher than the second best $\gap$.

{The sensitivity formulas in Theorems~\ref{thm:edge-sensitivity},\ref{thm:node-sensitivity} explicitly depend on structural properties of the graph, particularly minimum degree $D_{\min}$ and maximum degree $D_{\max}$.
Our empirical results (Table~\ref{tab:reulst_table_overall_acc_top1}, Figure~\ref{fig:abla_nodedeg}) reflect these theoretical dependencies:
(i) High-degree graphs such as \photo\  exhibit lower noise and higher accuracy; (ii) Low-degree or chain-like graphs incur higher sensitivity and lower accuracy, but $\gcn$ mitigates the impact.
}

\begin{figure}[!t]
	\subfloat[$\epsilon=2$ (\photo)]{
		\includegraphics[width = 3.9cm]{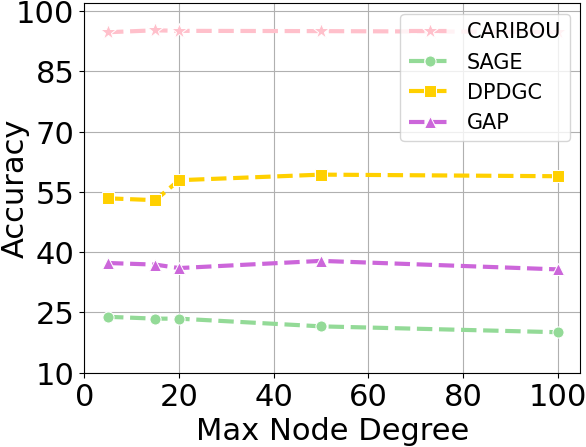}
		\label{fig:acc_nodedegree_e2_photo_main}
		}
	   \subfloat[$\epsilon=2$ (\chainS)]{
		\includegraphics[width = 3.9cm]{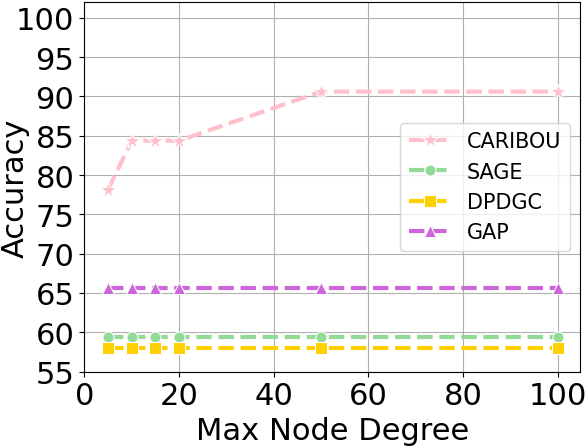}
		\label{fig:acc_nodedegree_e2_chains1_main}
		}
	   \caption{NDP Accuracy with Varying Max Node Degree.}
	   \label{fig:acc_nodedegree_e2}
\end{figure}

{
\section{Future Works and Discussions}


\textit{More GNN models.}
$\gcn$ is instantiated and evaluated primarily with commonly used message-passing architectures, where contractiveness naturally emerges or can be enforced by design. Extending our convergent privacy framework to a broader class of GNNs, including attention-based models (e.g., GAT), spectral convolution methods, and emerging graph transformers would be a natural next step, but new research problems will emerge.
For instance, these models differ in their Lipschitz properties, aggregation operators, and feature mixing patterns, which may influence the achievable privacy amplification and the expressiveness–contractiveness trade-off. Developing a unified analysis for these graph families, or designing contractive variants of non-message-passing architectures, is an open and promising direction.

\textit{White-box attacks and defenses.}
Our privacy auditing focuses on black-box membership threats, which align with the theoretical guarantees of perturbed message passing. However, stronger adversaries with white-box access to gradients, intermediate embeddings, or partial training states can mount reconstruction, inversion, or property inference attacks that fall outside our current threat model. Prior work has shown that gradient-based attacks can recover fine-grained structural information, especially in over-parameterized models. Investigating the extent to which contractiveness mitigates these stronger threats, and designing DP mechanisms that remain robust under partial or full white-box exposure, warrant future research. Such analyses may require combining $\gcn$ with complementary techniques such as gradient perturbation, secure aggregation, or private feature compression.

}

\section{Related Works}
\label{related_main}
\noindent\textbf{Multi-layer GNNs.}
Recent literature shows that multi-layer GNNs hold significant potential for modeling long-range dependencies and complex relational structures crucial for many real-world applications.
Node labels and attributes may depend on distant nodes, necessitating the aggregation of information over larger receptive fields~\cite{tnn/WuPCLZY21} through multi-layer GNNs.
Notably, Li~\etal~\cite{icml/Li0GK21} demonstrated, through the 1000-layer GNN, that increasing the network depth  attains substantial gains in accuracy, \eg, from $72\%$ with shallow GNNs to $88\%$ with  hundred- and thousand-layer GNNs, by capturing distant features.
However, enforcing DP in multi-layer GNNs is particularly challenging, as these GNN models aggregate node embeddings over deeper layers and broader neighborhoods. 
Current research still lacks an effective solution to injecting DP noise to multi-layer GNNs with privacy-utility balance.

 \smallskip

\noindent\textbf{Differentially private GNNs.}
Graphs consist of edges and nodes.
Corresponding to instance-level DP~\cite{csf/mironov2017renyi,ccs/AbadiCGMMT016,nips/0001S22}, the ``instance'' of graphs can be an edge or a node,  naturally called edge-level DP (EDP) and node-level DP (NDP). 
GNNs have emerged as a key approach for applications over graph-structured data, such as intrusion detection~\cite{eurosp/ApruzzeseLS23,sp/SommerP10}, social recommendation~\cite{icde/XiaSH0XP23}, and drug discovery~\cite{kdd/ZhongM23}.
Sharing trained GNN model can lead to privacy risks~\cite{uss/000100S022,ccs/0001MMBS22,corr/abs-2102-05429}, typically membership inference attack (MIA)~\cite{sp/CarliniCN0TT22,ccs/0001MMBS22}. 
MIA stems from ``overfitting'', where models can memorize training memberships~\cite{zhang2021understanding}, either an edge or a node. 
Consequently, GNNs can leak sensitive information about their edge- or node-level neighbors. 

To address these risks, existing research works~\cite{arxiv/sajadmanesh2020differential,arxiv/Wuprivacy23,nips/chien2023differentially,arxiv/DaigavaneNode21,corr/abs-2109-08907,uss/sajadmanesh2023gap} have integrating DP with GNNs to achieve EDP and NDP.
One research direction is to utilize graph perturbation (e.g., LPGNet~\cite{ccs/KolluriBHS22} and LapGraph~\cite{sp/0011L0022}) through  a randomized response mechanism and adding discrete DP noise to the adjacency matrix.
Then, the perturbed graph is passed to GNNs for subsequent training tasks, where the graph perturbation is required only once and also irrelevant to the GNN architectures.
 However, the GNN model utility is low when being trained over a perturbed graph when the privacy budget is tight, for example, $<40\%$ accuracy of $\epsilon=1,2,3,4$ reported in LPGNet~\cite{ccs/KolluriBHS22}.

To improve utility, \emph{perturbed message-passing mechanism} (PMP)~\cite{uss/sajadmanesh2023gap} has been proposed by adding the calibrated Gaussian noise to the message-passing layer, and \dpdgc\ perturbs the decoupled graph convolution~\cite{nips/chien2023differentially}.
As PMP realizes a better trade-off of privacy and utility, our work extends the research line of PMP.
Table~\ref{tab:compare_design} presents a comprehensive comparison.
Albeit their efforts on EDP and NDP, leveraging the contractive hidden node embeddings in private GNNs for amplifying privacy  remains an underexplored avenue; thus, $\gcn$ fills this gap.
More related works are detailed in Appendix~\ref{app:related}.

\section{Conclusion}
\label{sec:conclusion}
In this study, we provide a theoretical analysis establishing a convergent privacy budget for private deeper GNNs. Our analysis addresses a longstanding limitation in perturbed message-passing architectures, namely, the linear accumulation of noise with depth, by showing that privacy loss can remain bounded as the number of layers increases. Consequently, deeper models can be deployed with a significantly improved privacy-utility trade-off. Our analysis is broadly applicable, requiring only two conditions that are commonly satisfied in practice: the use of hidden intermediate states (also a standard design choice) and contractive message passing layers, which are often observed empirically.

To demonstrate the practical implications of our theory, we introduce a novel private GNN framework, $\gcn$, which incorporates a simple yet effective Contractive Graph Layer (CGL) that theoretically guarantees the contractiveness required by our analysis. $\gcn$ further integrates optimized privacy budgeting, and modular auditing mechanisms to deliver strong privacy guarantees while preserving model utility. Empirical results show that $\gcn$ substantially improves the privacy-utility trade-off and enhances robustness to membership inference attacks.

\section*{Acknowledgment}

We want to thank the reviewers for their insightful and constructive comments. 
The authors of University of California, Irvine are partially supported by NSF CNS-2220434 and TIP-2453148.
Any opinions, findings, and conclusions expressed in this material are those of the author(s) and do not necessarily reflect the views of the sponsors.

\bibliographystyle{IEEEtran}
\bibliography{references}


\appendix

\subsection{Review on Message Passing GNNs}
\label{app:mpgcn}
Every layer in the Graph Convolutional Network (GCN)~\cite{iclr/kipf2017semi} can be expressed as a message passing layer, where the aggregation function simply computes weighted sums of the features of the neighbors. The GCN layer applied to graph $\mG$ with a node feature matrix $\mX^{(k)}$ can be expressed as:
\begin{equation}
  \text{GCN}_G(\mX_u^{(k)}) = \sigma\left(\hat{\mA}\mX^{(k)}\mW^{(k)}\right),
\end{equation}
where $\hat{\mA} = \mD^{-\frac{1}{2}}(\mA + \mI)\mD^{-\frac{1}{2}}$ is the symmetrically normalized adjacency matrix of the graph, $\mI$ is the identity matrix, $\mD$ is the degree matrix of the graph (a diagonal matrix where $d_{ii} = \sum_j a_{ij}$), $a_{ij}$ is the $(i,j)$-th entry of the adjacency matrix $\mA$, $\mW^{(k)}$ is a learnable weight matrix at layer $k$, and $\sigma$ is a non-linear activation function.

Other simple variants include replacing the mean-type aggregation $\hat{\mA}\mX^{(k)}$ in GCN with a random walk adjacency matrix $\tilde{\mA}\mX^{(k)}$, where $\tilde{\mA} = \mD^{-1}\mA$ is the random walk normalized adjacency matrix. Alternative aggregation functions include sum or max aggregation, where $\hat{\mA}\mX^{(k)}$ is replaced with $\sum_{v \in \gN(u)} \mX_v^{(k)}$ (like in $\gap$) or $\max_{v \in \gN(u)} \mX_v^{(k)}$, respectively.
\section{Preliminary on $f$-DP and Related Definitions}
\label{app:pre_fdp}

In this section, we review the $f$-differential privacy (DP) framework with its definition using trade-off functions, and its special case of Gaussian differential privacy (GDP). We will also review some related results that will be used in our analysis.

The $f$-differential privacy framework~\cite{journal/dong2022gaussian,tit/KairouzOV17} is based on hypothesis testing where $f$ denotes a trade-off function between type I and type II errors. Given two output distributions $P$ and $Q$ of a mechanism $\Mcal$ on neighboring datasets $\Dcal$ and $\Dcal'$, the problem is to distinguish between them: $H_0: \text{dataset is } \Dcal$ vs. $H_1: \text{dataset is } \Dcal'$. For a rejection function $\phi: \mathcal{X} \to [0, 1]$, the type I error is $\E_P \phi$ and the type II error is $1 - \E_Q \phi$. The privacy guarantee is formalized through the trade-off function defined below.

This function characterizes the minimum Type II error (false negative rate) as a function of the maximum Type I error (false positive rate) in hypothesis testing between $P$ and $Q$.
A special case of the trade-off function is the \emph{Gaussian tradeoff function}, which is used to characterize privacy in the Gaussian differential privacy (GDP) setting.

\begin{definition}[Trade-off Function {\cite[Theorem~4.2]{icml/BokShifted24}}]
    \label{def:fdp_tradeoff}
    For distributions $P$, $Q$ on the same measurable space, the trade-off function $T(P, Q): [0,1] \to [0,1]$ is defined as:
    \begin{equation*}
        T(P, Q)(\alpha) = \inf\{1 - \Ebb_Q \phi : \Ebb_P \phi \leq \alpha,\ 0 \leq \phi \leq 1\}.
    \end{equation*}
    A randomized algorithm $\mathcal{M}$ satisfies $f$-differential privacy if for any adjacent datasets $D$ and $D'$, $T(\mathcal{M}(D), \mathcal{M}(D')) \geq f$.
\end{definition}

Note that a larger tradeoff function $f$ implies it is more difficult to distinguish the two neighboring datasets $\Dcal$ and $\Dcal'$, and thus the mechanism $\mathcal{M}$ is more private.

A special case of $f$-DP is the Gaussian differential privacy (GDP)~\cite{journal/dong2022gaussian}, which is defined as follows.
\begin{definition}[Gaussian Differential Privacy {\cite[Definition~2.1]{icml/BokShifted24}}]
    \label{def:GDP}
    A randomized algorithm $\mathcal{M}$ is $\mu$-GDP if for any adjacent datasets $\Dcal$ and $\Dcal'$, $T(\mathcal{M}(D), \mathcal{M}(D')) \geq G_\mu$,
    where $G_\mu$ is the Gaussian tradeoff function defined as
    $$G_\mu = T(\Ncal(0, 1), \Ncal(\mu, 1)).$$
    More specifically, the values of $G_\mu$ at $\alpha \in [0, 1]$ can be computed as
    \begin{equation*}
        G_\mu(\alpha) = \Phi\left( \Phi^{-1}(1 - \alpha) - \mu \right),
    \end{equation*}
    where $\Phi$ is the cumulative distribution function of the standard normal distribution.
\end{definition}


A common tool for analyzing the composition of $f$-DP is the tensor product of trade-off functions, which is defined as follows.
\begin{definition}[{\cite[Definition~3.1]{journal/dong2022gaussian}}]
    \label{def:tensor_product}
    The tensor product of two trade-off functions $f=T(P, Q)$ and $g=T(P', Q')$ is defined as
    \begin{equation*}
        f \otimes g = T(P \times P', Q \times Q').
    \end{equation*}
\end{definition}

The tensor product of trade-off functions satisfies the following properties as proved in \cite[Proposition~D.1]{journal/dong2022gaussian}.
\begin{lemma}[{\cite[Proposition~D.1]{journal/dong2022gaussian}}]
    \label{lem:tradeoff_composition}
    The tensor product of trade-off functions satisfies the following properties:
    \begin{itemize}
        \item $f \otimes g$ is a well-defined trade-off function.
        \item $f\otimes \mathrm{Id} = f$, where $\mathrm{Id}$ is the identity function.
        \item for GDP, $G_{\mu_1} \otimes G_{\mu_2} \otimes \cdots \otimes G_{\mu_n} = G_{\sqrt{\mu_1^2 + \mu_2^2 + \cdots + \mu_n^2}}$.
    \end{itemize}
\end{lemma}

\begin{lemma}[Post-processing inequality]\label{lemma:post_processing}
    Let $P, Q$ be two probability distributions and $K$ is some map possibly random, then
    \begin{equation}\label{equ:post_processing}
        T(K(P), K(Q)) \geq T(P, Q).
    \end{equation}
\end{lemma}

The following theorem shows the relationship between GDP and RDP, as well as the conversion from RDP to DP.
\begin{theorem}[GDP to RDP {\cite[Lemma~A.4]{icml/BokShifted24}}]
    \label{lem:gdp_to_rdp}
    If a mechanism is $\mu$-GDP, then it satisfies $(\alpha, \frac{1}{2}\alpha\mu^2)$-RDP for all $\alpha > 1$.
\end{theorem}


\begin{theorem}[RDP to DP {\cite[Proposition~3]{csf/mironov2017renyi}}]
    \label{thm:RDP_to_DP}
    If $f$ is an $(\alpha, \epsilon)$-RDP mechanism, then for any \(\delta \in (0, 1)\), it satisfies $(\epsilon + \frac{\log(1/\delta)}{\alpha - 1}, \delta)$-DP.
\end{theorem}

\subsection{Preliminaries: Perturbed Message Passing}\label{sec:one-layer-privacy}

Recall that each layer of a perturbed message-passing GNN updates the node features following Equation~\ref{equ:mp_sequence}, where $\mZ \sim \Ncal\bigl(0, \sigma^2 \mathrm{I}_d\bigr)$ is Gaussian noise added to obscure the true output of $\mp_{\mG},$ and $\Pi_\gK$ projects onto some convex feasible set $\gK.$
As shown in Proposition~\ref{prop:one-layer-privacy}, for a single layer, the privacy cost depends only on $\Delta(\mp)$ and $\sigma.$ However, by the composition theorem of RDP \citep{csf/mironov2017renyi}, stacking $K$ such layers yields a cost that scales linearly in $K$ (Corollary~\ref{cor:k-layer-privacy}).

\begin{proposition}
    \label{prop:one-layer-privacy}
    A one-layer perturbed message-passing GNN with mechanism $\mp$ is
    $
        \Bigl(\alpha, \frac{\alpha\,\Delta^2(\mp)}{2\sigma^2}\Bigr)\text{-R\'enyi DP},
    $
    which implies
    $
        \Bigl(\frac{\alpha\,\Delta^2(\mp)}{2\sigma^2} + \frac{\log(1/\delta)}{\alpha - 1}, \delta\Bigr)\text{-DP}.
    $
\end{proposition}

\begin{corollary}
    \label{cor:k-layer-privacy}
    By applying the composition theorem of R\'enyi differential privacy, a $K\text{-layer}$ perturbed message-passing GNN satisfies
    $
        \Bigl(\alpha, \frac{K\,\alpha\,\Delta^2(\mp)}{2\sigma^2}\Bigr)\text{-RDP},
    $
    which implies
    $
        \Bigl(\frac{K\,\alpha\,\Delta^2(\mp)}{2\sigma^2} + \frac{\log(1/\delta)}{\alpha - 1}, \delta\Bigr)\text{-DP}.
    $
\end{corollary}

This linear growth of the privacy budget is common in multi-hop settings, such as \citep{uss/sajadmanesh2023gap}, but is troublesome for deep GNNs, especially when many message-passing steps are needed for long-range dependencies. Excessively large privacy loss often undermines the model's utility.


\begin{proof}[Proof of Proposition~\ref{prop:one-layer-privacy}]
    By the Gaussian mechanism guarantee under R\'enyi differential privacy (RDP) \citep[Corollary 3]{csf/mironov2017renyi}, the map
    $X \mapsto \mp_{\mG}(X) + Z$
    is $\bigl(\alpha,\alpha\,\Delta^2(\mp)/(2\sigma^2)\bigr)\text{-RDP}.$ The post-processing theorem then ensures that projecting to $\gK$ does not increase the privacy cost. Converting from RDP to $(\epsilon,\delta)\text{-DP}$ \citep[Proposition 3]{csf/mironov2017renyi} completes the proof.
\end{proof}

\begin{proof}
[Proof of Corollary~\ref{cor:k-layer-privacy}]
    The proof follows directly from applying the RDP composition Theorem~\ref{the:rdp_composition} to Proposition~\ref{prop:one-layer-privacy} for $K$ independent layers. The RDP guarantee of the $K$-layer GNN is then converted to $(\epsilon,\delta)$-DP using Theorem~\ref{thm:RDP_to_DP}.
\end{proof}

\subsection{Proof of results in Section~\ref{sec:convergent-privacy}}
\label{app:convergent-privacy}
\begin{proof}[Proof of Theorem~\ref{thm:privacy_convergence_CNI}]
    The proof is reproduced from the original proof in \cite{icml/BokShifted24}, we include it here for self-containedness and to clarify that the contractive condition (Lipschitz constant bounded by $\gamma < 1$) is only required for iteration steps $k \geq 2$.

    Let $\mX^{(k)}$ and $\mX'^{(k)}$ be the outputs of the two CNI processes, that is,
    \begin{align*}
        \mX^{(k+1)}  & = \Pi_{\gK}\bigl(\phi_{k+1}(\mX^{(k)}) + \mZ^{(k+1)}\bigr),    \\
        \mX'^{(k+1)} & = \Pi_{\gK}\bigl(\phi'_{k+1}(\mX'^{(k)}) + \mZ'^{(k+1)}\bigr).
    \end{align*}
    The key idea in the proof, following \cite{icml/BokShifted24}, is to construct an auxiliary interpolating sequence $\{\widetilde{\mX}^{(k)}\}_{k=0}^K$ between the two CNI processes. For each step $k$, we define: \begin{equation}\label{eq:interpolated_cni}
        \begin{aligned}
        \widetilde{\mX}^{(k+1)} = &\Pi_{\gK}\bigl(\lambda_{k+1}\phi_{k+1}(\mX^{(k)}) \\
        & \quad + (1-\lambda_{k+1})\phi'_{k+1}(\widetilde{\mX}^{(k)}) + \mZ^{(k+1)}\bigr),
        \end{aligned}
    \end{equation}
    That is, $\widetilde{\mX}^{(k+1)}$ interpolates between using $\phi_{k+1}$ and $\phi'_{k+1}$ at each step with a mixing parameter $\lambda_{k+1} \in [0, 1]$, where $\lambda_{K} = 1$ so that $\widetilde{\mX}^{(K)} = \mX^{(K)}$. Note that this interpolation process uses the same noise vector $\mZ^{(k+1)}$ as in the original CNI process for $\mX^{(k+1)}$.

    We now recall a lemma from \cite{icml/BokShifted24} that establishes the trade-off function bound for one step of interpolation.
    \begin{lemma}[{\cite[Lemma~3.2]{icml/BokShifted24}}]
        \label{lem:one-step-tradeoff}
        Suppose that $\phi$ and $\phi'$ are $c$-Lipschitz functions and that $\|\phi(x) - \phi'(x)\| \leq s$ for all $x \in \gK$.
        Then for any $\lambda \geq 0$ and any random variable $\widetilde{\mX}$ satisfying $\|\mX - \widetilde{\mX}\| \leq z$, there is
        \begin{align*}
             & T(\lambda \phi(\mX) + (1-\lambda) \phi'(\widetilde{\mX}) + \gN(0, \sigma^2), \phi'(\mX') + \gN(0, \sigma^2)) \geq \\
             & \quad T(\widetilde{\mX}, \mX') \otimes G\left(\frac{\lambda(cz +s)}{\sigma}\right),
        \end{align*}
    \end{lemma}

    With the help of this lemma, we can then bound the trade-off function between $\mX^{(k)}$ and $\mX'^{(k)}$ by first bound the distance between $\mX^{(k)}$ and $\widetilde{\mX}^{(k)}$ and then applying the above lemma iteratively.

    We let $c_k$ be the maximum of the Lipschitz constants of $\phi_k$ and $\phi'_k$ for $k=1,\dots,K$, and by our assumption, $c_k < 1$ for $k=2,\dots,K$. We now track the distance between $\mX^{(k)}$ and the interpolated sequence $\widetilde{\mX}^{(k)}$.

    \noindent \textbf{Claim.} Let $z_k$ be a sequence of non-negative numbers given by
    $z_0 = 0$ and $z_{k+1} = (1-\lambda_{k+1})(c_{k+1} z_k + s)$ for $k=0, \dots, K-1$. Let $\widetilde{\mX}^{(k)}$ be the output of the interpolated CNI process in Equation~\ref{eq:interpolated_cni}.
    Then we have
    $\|\mX^{(k)} - \widetilde{\mX}^{(k)}\| \leq z_k$ for all $k=0, \dots, K$.

    \noindent \textbf{Proof of Claim.}
    The claim is proved by induction. For $k=0$, we have $\|\mX^{(0)} - \widetilde{\mX}^{(0)}\| = 0$ as $\widetilde{\mX}^{(0)} = \mX^{(0)}$.
    For $k\geq 1$, we have
    \begin{align*}
       & \|\mX^{(k)} - \widetilde{\mX}^{(k)}\| X \leq \|\Pi_{\gK}(\phi_k(\mX^{(k-1)}) + \mZ^{(k)})                                                           \\
                                              & \quad- \Pi_{\gK}(\lambda_k\phi_k(\mX^{(k-1)}) + (1-\lambda_k)\phi'_k(\widetilde{\mX}^{(k-1)}) + \mZ^{(k)})\|     \\
                                              & \leq \|\phi_k(\mX^{(k-1)}) - \lambda_k\phi_k(\mX^{(k-1)}) - (1-\lambda_k)\phi'_k(\widetilde{\mX}^{(k-1)})\| \\
                                              & \leq (1-\lambda_k)\|\phi_k(\mX^{(k-1)}) - \phi'_k(\widetilde{\mX}^{(k-1)})\|                                \\
                                              & \leq (1-\lambda_k)(\|\phi_k(\mX^{(k-1)}) - \phi'_k(\widetilde{\mX}^{(k-1)})\|                               \\
                                              & \quad  + \|\phi_k(\mX^{(k-1)}) - \phi_k(\widetilde{\mX}^{(k-1)})\|)                                         \\
                                              & \leq (1-\lambda_k)(c_k z_{k-1} + s),
    \end{align*}
    where the second inequality follows from the fact that $\Pi_{\gK}$ is a projection onto a convex set and hence has Lipschitz constant $1$ (\citep[Lemma~2.9]{icml/BokShifted24}).
    This concludes the proof of the claim by induction.
    We let $a_{k+1} = \lambda_{k+1}(c_{k+1} z_k + s)$ for $k=0, \dots, K-1$. In particular, $a_1 = \lambda_1 s$ regardless of the value $c_1$ since it is multiplied with $z_0 = 0$. This is the main our observation that the Lipchitz constant of the first iteration does not affect the constant $a_1$ and consequently the privacy guarantee.

    We can now apply Lemma~\ref{lem:one-step-tradeoff} iteratively to conclude the proof.
    There is,
    \begin{align*}
       & T(\mX^{(K)}, \mX'^{(K)})  = T(\widetilde{\mX}^{(K)}, \mX'^{(K)})                                                                                                                               \\
                                 & \leq T(\lambda_K \phi_K(\mX^{(K-1)}) + (1-\lambda_K)\phi'_K(\widetilde{\mX}^{(K-1)}) + \mZ^{(K)},                                                                    \\
                                 & \quad \phi'_K(\mX'^{(K-1)}) + \mZ^{(K)})                                                                                                                             \\
                                 & \leq T(\widetilde{\mX}^{(K-1)}, \mX'^{(K-1)}) \otimes G\left(\frac{a_K}{\sigma}\right)                                                                               \\
                                 & \leq T(\widetilde{\mX}^{(K-2)}, \mX'^{(K-1)}) \otimes G\left(\frac{a_{K-1}}{\sigma}\right) \otimes G\left(\frac{a_K}{\sigma}\right)                                  \\
                                 & \leq T(\widetilde{\mX}^{(0)}, \mX'^{(0)}) \otimes G\left(\frac{a_1}{\sigma}\right) \otimes G\left(\frac{a_2}{\sigma}\right) \cdots G\left(\frac{a_K}{\sigma}\right). \\
                                 & \leq G\left(\frac{1}{\sigma}\sqrt{\sum_{k=1}^Ka_k^2} \right)
    \end{align*}
    where the last inequality follows from the fact that $\widetilde{\mX}^{(0)}= \mX^{(0)}$ and property of Gaussian tradeoff function in Lemma~\ref{lem:tradeoff_composition}.

    Lastly, by the assumption that $\phi_k$ and $\phi'_k$ are $\gamma$-Lipschitz for $k=2,\dots,K$, we can let $c_k = \gamma$ for $k=2,\dots,K$. We set $\lambda_k = \frac{\gamma^{K -k} (1- \gamma^2)}{1 - \gamma^{K -k +2} - \gamma^k + \gamma^{K}}$ for $k=1,\dots,K$ as computed in \cite[Lemma~C.6]{icml/BokShifted24}. Then $z_k = \frac{(1-\gamma^k)(1 - \gamma^{K -k})}{(1 + \gamma^{K})(1-\gamma)}s$, and consequently, $a_k = \frac{\gamma^{K-k}(1+ \gamma)}{1 +\gamma^k} s$.
    In this case, $\sum_{k=1}^K a_k^2 = \frac{(1-\gamma^K)(1+\gamma)}{(1+\gamma^K)(1-\gamma)}s^2$. This concludes the proof of the theorem.
\end{proof}

\begin{proof}[Proof of Theorem~\ref{thm:privacy_convergence_general}]
    Let $G, G'$ be two adjacent graphs dataset. Let $\mp_G$ represent message passing operation for a fixed graph $G$, with projection $\Pi_\gK$ ensuring output lies in convex set $\gK$.
    Then starting from $\mX^{(0)}$, the perturbed message passing GNN with $\mp$ can be represented as the CNI process
    $$\text{CNI}(\mX^{(0)}, \{\mp_G\}_{k=1}^K, \{\Ncal(0, \sigma^2 \mI_d)\}_{k=1}^K, \gK).$$
    Similarly, the perturbed message passing GNN with $\mp'$ for graph $G'$ is represented as
    $$\text{CNI}(\mX^{(0)}, \{\mp_{\mG'}\}_{k=1}^K, \{\Ncal(0, \sigma^2 \mI_d)\}_{k=1}^K, \gK).$$

    By assumptions, both $\mp_G$ and $\mp_{\mG'}$ are contractive with Lipschitz constant $\gamma< 1$ for layers $k=2,\dots,K$ and sensitivity bound of $\mp$ implies that $\|\mp_G(x) - \mp_{\mG'}(x)\| \leq \Delta(\mp)$ for all $x\in \gK$.
    This shows that we can the meta-theorem in Theorem~\ref{thm:privacy_convergence_CNI} to bound the trade-off function between the two CNI processes as
    \begin{align*}
        T(\mX^{(K)}, \mX'^{(K)}) & \geq G\left(\frac{\Delta(\mp)}{\sigma}\sqrt{\frac{1-\gamma^K}{1+\gamma^K} \frac{1+\gamma}{1-\gamma} }\right)
    \end{align*}
    This implies that the $K$-layer perturbed message passing GNN with $\mp$ is $\frac{\Delta(\mp)}{\sigma}\sqrt{\frac{1-\gamma^K}{1+\gamma^K} \frac{1+\gamma}{1-\gamma} }$-Gaussian differential privacy (GDP) defined in Definition~\ref{def:GDP}. By applying Lemma~\ref{lem:gdp_to_rdp}, we then obtain the stated RDP guarantee.
    Then by applying Theorem~\ref{thm:RDP_to_DP}, we can convert the RDP guarantee to $(\epsilon,\delta)$-DP.
\end{proof}

\subsection{Proof of results in Section~\ref{sec:gcn_privacy}}
\label{app:gcn_privacy}
\begin{proof}[Proof of Proposition~\ref{prop:lip}]
    Let $\mY^{(k-1)}, \mY'^{(k-1)} \in \gK$ be two inputs to the message passing operator $\mp_G$ at layer $k$.
    Since $k\geq 2$, the residue term $\beta \mX^{(0)}$ is independent of the input $\mY^{(k-1)}$ and $\mY'^{(k-1)}$, and thus does not affect the Lipschitz constant.
    We can write the difference between the outputs of $\mathsf{CGL}$ as follows:
    \begin{align*}
         & \|\mathsf{CGL}(\mY^{(k-1)}) - \mathsf{CGL}(\mY'^{(k-1)})\|                                                                           \\
         & \leq \|\Clip(\alpha_1 \hat{\mA}\mY^{(k-1)}+\alpha_2\mathsf{Mean}(\mY^{(k-1)})) + \beta \mX^{(0)}                                     \\
         & \quad \quad - \Clip(\alpha_1 \hat{\mA}\mY'^{(k-1)}+\alpha_2\mathsf{Mean}(\mY'^{(k-1)})) - \beta \mX^{(0)}\|                          \\
         & \leq \Clip\|\alpha_1 (\hat{\mA}\mY^{(k-1)}-\hat{\mA}\mY'^{(k-1)})\\
         & \quad \quad +\alpha_2(\mathsf{Mean}(\mY^{(k-1)})-\mathsf{Mean}(\mY'^{(k-1)}))\| \\
         & \leq \Clip(\alpha_1 + \alpha_2)\|\mY^{(k-1)}-\mY'^{(k-1)}\|\\
         & = \Clip\|\mY^{(k-1)}-\mY'^{(k-1)}\|,
    \end{align*}
    where the second line follows from the fact that the operator norms of $\hat{\mA}$ and $\mathsf{Mean}$ are bounded by $1$.
\end{proof}

\begin{proof}[Proof of Theorem~\ref{thm:edge-sensitivity}]
    Let $\mG, \mG'$ be two edge adjacent graphs and $\hat{\mA}, \hat{\mA}'$ be the corresponding adjacency matrices of $\mG, \mG'$ respectively. Without loss of generality, we assume that the edge $e_{uv}$ is added to $\mG$ to form $\mG'$ for two nodes $u$ and $v$. Then the $\mathsf{CGL}$ layer updates the node features as follows:
    \begin{align*}
        \mX^{(k)} & = \Clip(\alpha_1 \hat{\mA}\mX^{(k-1)}+\alpha_2\mathsf{Mean}(\mX^{(k-1)})) + \beta \mX^{(0)}, \\
        \mX'^{(k)} & = \Clip(\alpha_1 \hat{\mA}'\mX'^{(k-1)}+\alpha_2\mathsf{Mean}(\mX'^{(k-1)})) + \beta \mX^{(0)}.
    \end{align*}
    The difference between the two outputs is given by the aggregation of $\hat{\mA}$ and $\hat{\mA}'$. Then the edge-level sensitivity $\Delta_e(\mathsf{CGL})$ is the amount to bound $\|\hat{\mA}\mX^{(k)} - \hat{\mA}'\mX'^{(k)}\|_F$.
    Since only one edge is added, the difference between $\hat{\mA}$ and $\hat{\mA}'$ is only on the row corresponding to $u$ and $v$.

    For row $u$, we need to bound $\|(\hat{\mA} \mX^{(k)})_u - (\hat{\mA}' \mX^{(k)})_u\|_2$. For $(\hat{\mA}\mX^{(k)})_u$, we can write it as
    \begin{equation}
        \begin{aligned}
            (\hat{\mA}\mX^{(k)})_u & = \frac{1}{d_u + 1} \mX^{(k)}_u + \sum_{w \in \mN_u} \frac{1}{\sqrt{d_u+1}\sqrt{d_w+1}} \mX^{(k)}_w \\
        \end{aligned}
    \end{equation}
    where $d_u$ is the degree of node $u$ in graph $\mG$ and $\mN_u$ is the neighbors of node $u$ in graph $\mG$.
    For $(\hat{\mA}' \mX^{(k)})_u$, with the same notation for $d_u$ and $\mN_u$, we can write it as
    \begin{equation}
        \begin{aligned}
            (\hat{\mA}'\mX^{(k)})_u & = \frac{1}{d_u + 2} \mX^{(k)}_u + \sum_{w \in \mN_u} \frac{1}{\sqrt{d_u+2}\sqrt{d_w+1}} \mX^{(k)}_w \\
                                    & + \frac{1}{\sqrt{d_u+2}\sqrt{d_v'+1}} \mX^{(k)}_v                                                   \\
        \end{aligned}
    \end{equation}
    where $d'_v$ is the degree of node $v$ in graph $\mG'$.
    
    Then there is
    \begin{align*}
         & \|(\hat{\mA}\mX^{(k)})_u - (\hat{\mA}' \mX^{(k)})_u\|_2                                                                                                         \\
         & \leq \left\| \frac{1}{d_u + 1} \mX^{(k)}_u - \frac{1}{d_u + 2} \mX^{(k)}_u \right\|_2                                                                           \\
         & + \left\| \sum_{w \in \mN_u} \hspace{-2mm}\frac{1}{\sqrt{d_u+1}\sqrt{d_w+1}} \mX^{(k)}_w - \hspace{-2mm}\sum_{w \in \mN_u}\hspace{-2mm} \frac{1}{\sqrt{d_u+2}\sqrt{d_w+1}} \mX^{(k)}_w \right\|_2 \\
         & + \left\| \frac{1}{\sqrt{d_u+2}\sqrt{d_v'+1}} \mX^{(k)}_v \right\|_2                                                                                            \\
         & \leq \frac{1}{(d_u + 1)(d_u + 2)}  + \sum_{w\in N_u}\frac{1}{\sqrt{d_w+1}}(\frac{1}{\sqrt{d_u+1}} - \frac{1}{\sqrt{d_u+2}})                                     \\
         & + \frac{1}{\sqrt{d_u+2}\sqrt{d_v'+1}}                                                                                                                           \\
         & \leq \frac{1}{(d_u + 1)(d_u + 2)}  + \frac{d_u}{\sqrt{d_w+1}}(\frac{1}{\sqrt{d_u+1}} - \frac{1}{\sqrt{d_u+2}})                                                  \\
         & + \frac{1}{\sqrt{d_u+2}\sqrt{d_v'+1}}
    \end{align*}
    To bound $\frac{d_u}{\sqrt{d_w+1}}(\frac{1}{\sqrt{d_u+1}} - \frac{1}{\sqrt{d_u+2}})$, we study the monotonicity of the function $f(x) = \frac{x}{\sqrt{x+1}} - \frac{x}{\sqrt{x+2}}$ for $x>0$. It turns out that $f(x)$ only has one positive critical point and is around $x=2.9$, and when evaluated on integers, $f(x)$ increases from $x=1$ to $x=3$ and then decreases from $x=3$ to $\infty$.
    Thus, when the minimum degree $D_{\min}$ of $\mG$ is larger than $3$, the function $f(x)$ is maximized at $x= D_{\min}$, and we have
    \begin{align*}
          \|(\hat{\mA}\mX^{(k)})_u &- (\hat{\mA}' \mX^{(k)})_u\|_2                                                                                      \\
         & \leq \frac{1}{(D_{\min} +1) (D_{\min} + 2)}   \\
         &\quad  +  \frac{D_{\min}}{\sqrt{D_{\min}+1}}(\frac{1}{\sqrt{D_{\min}+1}} - \frac{1}{\sqrt{D_{\min}+2}})\\
         &\quad + \frac{1}{\sqrt{D_{\min}+2}\sqrt{D_{\min}+1}}                                                                                               \\
    \end{align*}
    where we use the fact that the minimum degree of $\mG'$ is larger than that of $\mG$.
    When $1\leq D_{\min} \leq 3$, we can bound the function $f(x)$ by $f(3)= \frac{3}{\sqrt{4}} - \frac{3}{\sqrt{5}}$. This results in
    \begin{align*}
         & \|(\hat{\mA}\mX^{(k)})_u - (\hat{\mA}' \mX^{(k)})_u\|_2                                                                                                            \\
         & \leq \frac{1}{(D_{\min} +1) (D_{\min} + 2)}  + (\frac{3}{\sqrt{4}} - \frac{3}{\sqrt{5}})\frac{1}{\sqrt{D_{\min}+1}}\\
        &\quad \quad + \frac{1}{\sqrt{D_{\min}+2}\sqrt{D_{\min}+1}}
    \end{align*}
    For notation convenience, we use $C(D_{\min})$ to denote the piecewise function of $D_{\min}$, which is defined as
    \begin{equation}
        C(D_{\min}) = \begin{cases}
            \frac{D_{\min}}{\sqrt{D_{\min}+1}} - \frac{D_{\min}}{\sqrt{D_{\min}+2}} & D_{\min} > 3          \\
            (\frac{3}{\sqrt{4}} - \frac{3}{\sqrt{5}})                               & 1\leq D_{\min} \leq 3 \\
        \end{cases}
        \label{eq:C_Dim}
    \end{equation}
    
    Therefore, the effect of modifying an edge on a single node $u$ of $\hat{\mA}\mX^{(k)}$ is bounded by
    \begin{equation}
        \begin{aligned}
             & \|(\hat{\mA}\mX^{(k)})_u - (\hat{\mA}' \mX^{(k)})_u\|_2                                                                             \\
             & \leq \frac{1}{(D_{\min} +1) (D_{\min} + 2)}  + \frac{C(D_{\min})}{\sqrt{D_{\min}+1}} \\
             &\quad + \frac{1}{\sqrt{D_{\min}+2}\sqrt{D_{\min}+1}}
        \end{aligned}
        \label{eq:single_node_sensitivity_bound}
    \end{equation}
    
    The same analysis can be applied to the row $v$ of $\hat{\mA}\mX^{(k)}$ and $\hat{\mA}'\mX^{(k)}$. The edge sensitivity $\Delta_e(\gcn)$ can then be bounded as the following:
   \begin{align*}
    & \Delta_e(\gcn) := \max_{\mG, \mG'} \|\hat{\mA}\mX^{(k)} - \hat{\mA}'\mX^{(k)}\|_F \\
    & \leq \alpha_1 \Clip\sqrt{
           \begin{multlined}[t]
           \|(\hat{\mA}\mX^{(k)})_u - (\hat{\mA}' \mX^{(k)})_u\|_2^2 \\
           + \|(\hat{\mA}\mX^{(k)})_v - (\hat{\mA}' \mX^{(k)})_v\|_2^2
           \end{multlined}
    } \\
    & \leq \sqrt{2}\alpha_1 \Clip \left( \frac{1}{(D_{\min}+1)(D_{\min}+2)} + \frac{C(D_{\min})}{\sqrt{D_{\min}+1}} \right. \\
    & \qquad \left. + \frac{1}{\sqrt{D_{\min}+2}\sqrt{D_{\min}+1}} \right)
\end{align*}
\end{proof}

\begin{proof}[Proof of Theorem~\ref{thm:node-sensitivity}]
Let $\mG, \mG'$ be two node adjacent graphs and $\hat{\mA}, \hat{\mA}'$ be the corresponding adjacency matrices of $\mG, \mG'$ respectively. Without loss of generality, we assume that the node $v$ is added to $\mG$ to form $\mG'$ and connected to nodes $N_v$ in $\mG$. The layer updates the node features as follows:
\begin{align*}
    \mX^{(k)} & = \Clip(\alpha_1 \hat{\mA}\mX^{(k-1)}+\alpha_2\mathsf{Mean}(\mX^{(k-1)})) + \beta \mX^{(0)}, \\
    \mX'^{(k)} & = \Clip(\alpha_1 \hat{\mA}'\mX'^{(k-1)}+\alpha_2\mathsf{Mean'}(\mX'^{(k-1)})) + \beta \mX^{(0)}.
\end{align*}
The difference between the two outputs is given by the aggregation of $\hat{\mA}$ and $\hat{\mA}'$ as well as the mean operator $\mathsf{Mean}$ and $\mathsf{Mean'}$ since $\mG'$ has one more node than $\mG$.
Then the node-level sensitivity $\Delta_v(\mathsf{CGL})$ can be bounded as follows:
\begin{align*}
    & \Delta_n(\mathsf{CGL})                                                                                                                                                          \\
    & = \max_{\mG, \mG'} \|\mathsf{CGL}(\mX^{(k)}) - \mathsf{CGL}'(\mX^{(k)})\|_F                                                                                                                      \\
    & \leq \|\mX_v'\|_2 + \sum_{u \in \mN_v} \alpha_1 \Clip\|(\hat{\mA}\mX^{(k)})_u - (\hat{\mA}'\mX^{(k)})_u\|_2                                                                      \\
    & + \quad \sum_{u \in \mN_v} \alpha_2 \Clip\|\mathsf{Mean}(\mX^{(k)})_u - \mathsf{Mean'}(\mX^{(k)})_u\|_2
\end{align*}
For the  first term $\|\mX_v'\|_2$, it is bounded by $1$ by constraint of $\gK$.
For the second term, we can argue similar as in the proof of Theorem~\ref{thm:edge-sensitivity}. For nodes $u \in \mN_v$, there is
\begin{align*}
     & \|(\hat{\mA}\mX^{(k)})_u - (\hat{\mA}' \mX^{(k)})_u\|_2                                                                                                         \\
     & \leq \left\| \frac{1}{d_u + 1} \mX^{(k)}_u - \frac{1}{d_u + 2} \mX^{(k)}_u \right\|_2                                                                           \\
     & + \left\| \sum_{w \in \mN_u} \hspace{-2mm}\frac{1}{\sqrt{d_u+1}\sqrt{d_w+1}} \mX^{(k)}_w - \sum_{w \in \mN_u} \hspace{-2mm}\frac{1}{\sqrt{d_u+2}\sqrt{d_w+1}} \mX^{(k)}_w \right\|_2 \\
     & + \left\| \frac{1}{\sqrt{d_u+2}\sqrt{d_v'+1}} \mX^{(k)}_v \right\|_2                                                                                            \\
     & \leq \frac{1}{(d_u + 1)(d_u + 2)}  + \sum_{w\in N_u}\frac{1}{\sqrt{d_w+1}}(\frac{1}{\sqrt{d_u+1}} - \frac{1}{\sqrt{d_u+2}})                                     \\
     & + \frac{1}{\sqrt{d_u+2}\sqrt{d_v'+1}}                                                                                                                           \\
     & \leq \frac{1}{(d_u + 1)(d_u + 2)}  + \frac{d_u}{\sqrt{d_w+1}}(\frac{1}{\sqrt{d_u+1}} - \frac{1}{\sqrt{d_u+2}})                                                  \\
     & + \frac{1}{\sqrt{d_u+2}\sqrt{d_v'+1}}                                                                                                                          \\
     & \leq \frac{1}{(D_{\min} +1) (D_{\min} + 2)}  + \frac{C(D_{\min})}{\sqrt{D_{\min}+1}} \\
     & + \frac{1}{\sqrt{D_{\min}+2}\sqrt{d_v'+1}}
\end{align*}
Then the summation term $\sum_{u \in \mN_v} \|(\hat{\mA}\mX^{(k)})_u - (\hat{\mA}'\mX^{(k)})_u\|_2$ can be bounded by
\begin{align*}
     & \sum_{u \in \mN_v} \|(\hat{\mA}\mX^{(k)})_u - (\hat{\mA}'\mX^{(k)})_u\|_2                                             \\
     & \leq \sum_{u \in \mN_v} \left(\frac{1}{(D_{\min} +1) (D_{\min} + 2)}  + \frac{C(D_{\min})}{\sqrt{D_{\min}+1}} \right. \\
     & \qquad \left. + \frac{1}{\sqrt{D_{\min}+2}\sqrt{d_v'+1}}\right) \\
     & \leq |d_v'| \left(\frac{1}{(D_{\min} +1) (D_{\min} + 2)}  + \frac{C(D_{\min})}{\sqrt{D_{\min}+1}} \right. \\
     & \qquad \left. + \frac{1}{\sqrt{D_{\min}+2}\sqrt{d_v'+1}}\right)             \\
\end{align*}
\begin{align*}
 & \leq
    \frac{\sqrt{d_v'}}{(D_{\min} +1) (D_{\min} + 2)}  + \frac{C(D_{\min})\sqrt{d_v'}}{\sqrt{D_{\min}+1}} \\
    &\qquad + \frac{\sqrt{d_v'}}{\sqrt{D_{\min}+2}\sqrt{d_v'+1}}\\
     & \leq
    \frac{\sqrt{d_v'}}{(D_{\min} +1) (D_{\min} + 2)}  + \frac{C(D_{\min})\sqrt{d_v'}}{\sqrt{D_{\min}+1}} + \frac{1}{\sqrt{D_{\min}+2}}
    \\
     & \leq
    \frac{\sqrt{D_{\max}}}{(D_{\min} +1) (D_{\min} + 2)}  + \frac{C(D_{\min})\sqrt{D_{\max}}}{\sqrt{D_{\min}+1}} + \frac{1}{\sqrt{D_{\min}+2}},
\end{align*}
where $D_{\max}$ is the maximum degree of the graph $\mG$ and $d_v'$ is the degree of node $v$ in graph $\mG'$.
Additionally, there is
\begin{align*}
     & \|\mathsf{Mean}(\mX^{(k)})_u - \mathsf{Mean'}(\mX^{(k)})_u\|_2                                                                             \\
     & = \left\| \frac{1}{|V|} \sum_{w \in V} \mX^{(k)}_w - \frac{1}{|V|+1} \sum_{w \in V} \mX^{(k)}_w - \frac{1}{|V|+1} \mX^{(k)}_{v'}\right\|_2 \\
     & \leq \left\| \frac{1}{|V|(|V|+1)} \sum_{w \in V} \mX^{(k)}_w - \frac{1}{|V|+1} \mX^{(k)}_{v'}\right\|_2                                    \\
     & \leq \frac{2}{|V|+1},
\end{align*}
where $|V|$ is the number of nodes in graph $\mG$.

Then for the node-level sensitivity of one layer of $\mathsf{CGL}$, we have
\begin{align*}
     & \Delta_n(\mathsf{CGL})                                                                                                                                                          \\
     & = \max_{\mG, \mG'} \|\mathsf{CGL}(\mX^{(k)}) - \mathsf{CGL}'(\mX^{(k)})\|_F                                                                                                                      \\
     & \leq \|\mX_v'\|_2 + \sum_{u \in \mN_v} \alpha_1 \Clip\|(\hat{\mA}\mX^{(k)})_u - (\hat{\mA}'\mX^{(k)})_u\|_2                                                                      \\
     & + \quad \sum_{u \in \mN_v} \alpha_2 \Clip\|\mathsf{Mean}(\mX^{(k)})_u - \mathsf{Mean'}(\mX^{(k)})_u\|_2                                                                          \\
     & \leq 1 +  \alpha_1 \Clip \left(\frac{\sqrt{D_{\max}}}{(D_{\min} +1) (D_{\min} + 2)}  \right.\\
     & \qquad \left. + \frac{C(D_{\min})\sqrt{D_{\max}}}{\sqrt{D_{\min}+1}} + \frac{1}{\sqrt{D_{\min}+2}}\right) \\
     & + \quad \alpha_2 \Clip\frac{2|V|}{|V|+1}
\end{align*}
\end{proof}

\begin{proof}[Proof of Corollary~\ref{the:CGL_rdp}]
    The proof follows from plugging the edge-level and node-level sensitivity of $\gcn$ into Theorem~\ref{thm:privacy_convergence_general}.
\end{proof}

\begin{proof}[Proof of Theorem~\ref{the:overall_epsilon}]
    The result comes from directly applying the composition theorem of DP to different modules of $\gcn$.
\end{proof}

\subsection{Experimental Setup}
\label{app:exp_setup}
This section presents  experiments setups and counterpart baselines in detail.

\subsubsection{Datasets}
\label{app:dataset}
Table~\ref{tab:datasets} lists statistics of these datasets, including Amazon co-purchase networks (\photo,  \computers~\cite{corr/abs-1811-05868}), social network (\facebook~\cite{corr/abs-1102-2166}),  and citation networks (\cora, \pubmed~\cite{sen2008collective}).
These datasets are also widely adopted
as benchmark datasets to evaluate various GNNs~\cite{nips/GuC0SG20,uss/sajadmanesh2023gap,nips/chien2023differentially,nips/YaoJRJ23}.
In addition, we use the synthetic graph generation algorithm developed in IGNN~\cite{nips/GuC0SG20}, to produce chain-structured datasets for convenient and generalized verification of multi-hop aggregations.
The chain-structured dataset can be configured with various number of nodes per chain, number of chains per class, and the number of classes.
In Table~\ref{tab:syn_datasets}, the statistics of the four synthesized chain-structured datasets, namely \chainS, \chainM,  \chainL, and \chainX\, are shown.
We unify the feature dimension to $5$ and focus on the structure of chain-structured dataset, and we configure the number of nodes to $\{8,10,15\}$ per chain, and the number of chains to $\{3,5\}$ per class.

\begin{table}[!h]
\caption{Standard Graph Dataset Statistics. ``\#Tra'' and ``\#Test'' are the ratios of the total nodes used for training and testing.}
\setlength\tabcolsep{2pt}
\begin{center}
\begin{tabular}{l|c|c|c|c|rr}
\hline
    \multicolumn{1}{l|}{\textbf{Dataset}} & \textbf{Node} & \textbf{Edge} & \textbf{Feature} & \textbf{Class} &\textbf{\#Tra} &\textbf{\#Test}
\\ \hline
\computers	 & $13,471$ 	 & $491,722$ 	 & $767$ & $10$ & $10\%$ & $20\%$ \\
\facebook 	 & $26,406$ 	 & $2,117,924$ 	 & $501$ & $6$ & $10\%$ & $20\%$ \\
\pubmed 	 &$19,717$& $88,648$&$500$ &$3$&$10\%$ & $20\%$\\
\cora 	 & $2,708$ & $10,556$ &$1,433$ &$7$&$10\%$ & $20\%$\\
\photo 	 &$7,535$ & $238,162$ & $745$ & $8$&$10\%$ & $20\%$\\
\hline
\end{tabular}
\end{center}
\label{tab:datasets}
\end{table}

\begin{table}[!h]
\caption{Chain-structured Datasets Statistics. Each chain has the same number of nodes and ``Node'' is the sum of nodes across chains.}
\setlength\tabcolsep{2pt}
\begin{center}
\begin{tabular}{l|c|c|c|c|rr}
\hline
    \multicolumn{1}{l|}{\textbf{Dataset}} & \textbf{Node} & \textbf{Chain} &  \textbf{Feature} & \textbf{Class} &\textbf{\#Tra} &\textbf{\#Test}
\\ \hline
\chainS 	 & $48$ 	 & $6$ 	 & $5$ & $2$ & $8$ & $ 32$ \\
\chainM	 & $60$ 	 & $6$ 	 & $5$ & $2$ & $10$ & $40$ \\
\chainL	 & $90$ 	 & $6$ 	 & $5$ & $2$ & $15$ & $60$ \\
\chainX&       $150$& $10$ &$5$ & $2$ & $25$& $100$\\\hline
\end{tabular}
\end{center}
\label{tab:syn_datasets}
\end{table}


\subsubsection{Hardware and software}
Our experiments have been conducted on the Ubuntu 20.04.2 LTS server, with AMD Ryzen Threadripper 3970X 32-core CPUs of 256 GB CPU memory and NVIDIA GeForce RTX 3090 of 24GB memory.
$\gcn$ equipped with edge- and node-level privacy has been implemented in PyTorch using PyTorch-Geometric (PyG)~\cite{corr/abs-1903-02428} framework.
DP implementation adopts the \texttt{autodp} library~\cite{AutoDP2025}, which includes analytical moments accountant and private training with DP-SGD~\cite{ccs/AbadiCGMMT016}.

\subsubsection{Baseline algorithms} 
\label{app:baseline}

The evaluated edge-level private algorithms are introduced below.
\begin{itemize}
    \item \dpdgc: This algorithm decouples the neighborhood aggregation process from the transformation of node features, thereby avoiding direct aggregation of sensitive data.
    This decoupling enables improved privacy guarantees through the DP composition theorem, allowing for a more efficient balance between privacy guarantee and  model performance.
    \item \gap: This algorithm preserves edge privacy via
aggregation perturbation, \ie, adding calibrated Gaussian noise to the output of the aggregation function for hiding the presence of a particular edge.
\gap's architecture involves pre-training encoder, aggregation module, and classification module, so that  \gap\   can reduce the privacy costs of the perturbed aggregations by one-time computation over lower-dimensional embeddings.
\item  \sage: This algorithm adopts the graph perturbation, building on the popular GraphSAGE architecture as its backbone GNN model.
To realize graph perturbation, \sage\ perturbs the adjacency matrix of graph  using the asymmetric randomized response.
\item \mlpE: Typical \mlpE\ model is trained over node features, without referring to graph edges. 
Thus, $\epsilon=0$ always holds for \mlpE, revealing nothing about edges and providing complete edge-level privacy.

\end{itemize}

 By extending edge-level private algorithms, node-level DP detailed below can be realized by protecting full information (\ie, edges, node features) of a node.

\begin{itemize}
    \item \dpdgc:  This algorithm extends \dpdgc\ and bounds the out-degree of nodes to realize the node-level DP.
    Thus, \dpdgc\ reduces the dependency of DP noise variance on the maximum node degree, improving the trade-off between privacy guarantee and model utility.
    \item \gap: This algorithm extends \gap\ for bounded-degree graphs, where each node has controllable influence to its neighbors by sampling a limited number of neighbors.
    \item \sage: This algorithm adapts DP-SGD to the GraphSAGE model, and simultaneously adds the noise to aggregation function constrained by node-level sensitivity.
    \item \mlpE: This algorithm is trained with DP-SGD without accessing the edges.
\end{itemize}

\subsubsection{Vanilla algorithms} The non-private versions of all private algorithms above are used to quantify the accuracy loss of corresponding EDP and NDP algorithms.

 \subsubsection{Model architectures \& configurations of $\gcn$}
We concatenate $K$ layers ranging from $1$ to $20$, with $K, \alpha_1, \beta, \Clip$ as the hyper-parameters. 
If the range between two interactive nodes is long, we use larger values to test the influence of hops.
Specifically, for \chainS\  and \chainM\  datasets, $\gcn$ uses the number of hops $K\in\{15,10,9,8,7,6,5,4,3,2,1\}$, while for \chainS\  and \chainM\  datasets consisting of longer chains, $\gcn$ takes $K\in \{20,17,15,\allowbreak13, 11,9,7,5\}$. 
As for standard datasets listed in Table~\ref{tab:datasets}, $K$ is chosen from the union sets of the aforementioned two sets, i.e., $\{20,17,15,\allowbreak13,11,10,9,8,7,6,5,4,3,2,1\}$.
To align with \gap's configuration~\cite{uss/sajadmanesh2023gap}, we set the number of hidden units to $\{16,64\}$ and use the SeLU activation function~\cite{nips/KlambauerUMH17}
at every layer.
We adopt the Adam optimizer over $100$ epochs with a learning rate $0.001$, and pick the best accuracy to report.

\subsubsection{Privacy configuration and parameters}
We implement convergent privacy allocation based on \gap's privacy budget accounting mechanism~\cite{uss/sajadmanesh2023gap} and numerically calibrate noise level $\sigma$ by setting $\epsilon$.
All edge/node-level private algorithms adopt the Gaussian mechanism to sample the i.i.d. noise to realize the desired $(\epsilon,\delta)$-DP.
For comprehensive evaluation, we consider different choices of $\epsilon\in\{1, 2, 4, 8, 16, 32\}$, while $\delta$ is set to be smaller than the inverse number of edges for EDP or the inverse number of nodes for NDP.

\subsection{Computational Overhead}
\label{app:overhead}

\begin{table}[h]
\small
    \caption{Execution Time of $3$ Training-Test Runs (seconds) for EDP}\label{tab:reulst_table_running_time}
    \vspace{1pt}
    \centering
    \setlength\tabcolsep{3pt}
    \adjustbox{max width=\textwidth}{
    \begin{tabular}{c|c|c|c|c|c}
    \toprule[1pt]
    \textbf{Dataset} & \textbf{$\gcn$} & \dpdgc& \gap & \sage & \mlpE \\
    \cmidrule[1pt]{1-6}
    \computers & $3.70$ & $87.72$ & $6.62$ & $3.37$ & $39.37$ \\
    \facebook & $4.34$ & $233.65$ & $6.36$ & $4.29$ & $70.93$ \\
    \pubmed & $3.46$ & $149.07$ & $6.08$ & $3.78$ & $55.11$ \\
    \cora & $3.05$ & $16.85$ & $5.98$ & $3.10$ & $10.88$ \\
    \photo & $3.07$ & $43.84$ & $6.00$ & $3.33$ & $23.96$ 
    \\\hline
    \chainS & $2.92$ & $5.93$ & $5.54$ & $2.99$ & $4.85$ \\
    \chainM & $2.76$ & $5.62$ & $5.67$ & $3.16$ & $4.76$ \\
    \chainL & $2.77$ & $5.93$ & $5.39$ & $3.15$ & $4.61$ \\
    \chainX & $2.91$ & $5.93$ & $5.23$ & $3.06$ & $5.04$ \\
    \bottomrule
    \end{tabular}
    }
\end{table}

\begin{table}[h]
\small
    \caption{Execution Time of $3$ Training-Test Runs (seconds) for NDP}\label{tab:reulst_table_running_time_ndp}
    \vspace{1pt}
    \centering
    \setlength\tabcolsep{2pt}
    \adjustbox{max width=\textwidth}{
    \begin{tabular}{c|c|c|c|c|c}
    \toprule[1pt]
    \textbf{Dataset} & \textbf{$\gcn$} & \dpdgc& \gap & \sage & \mlpE \\
    \cmidrule[1pt]{1-6}
    \computers                  & $3.23$ & $526.12$ & $351.83$ & $386.04$ & $1452.91$ \\
    \facebook                   & $4.39$ & $2056.61$ & $726.72$ & $767.28$ & $2217.46$ \\
    \pubmed                     & $2.82$ & $1183.27$ & $528.56$ & $418.21$ & $1472.64$ \\
    \cora                       & $2.52$ & $96.29$ & $92.04$ & $62.06$ & $184.39$ \\
    \photo                      & $2.91$ & $298.93$ & $217.79$ & $203.94$ & $329.01$ \\\hline
    \chainS   & $2.26$ & $10.96$ & $9.41$ & $1.51$ & $13.77$ \\
    \chainM                    & $2.27$ & $11.37$ & $10.52$ & $1.49$ & $13.56$ \\
    \chainL                    & $2.29$ & $10.48$ & $7.71$ & $1.53$ & $13.74$ \\
    \chainX                    & $2.26$ & $11.66$ & $7.21$ & $1.51$ & $13.63$ \\
    \bottomrule
    \end{tabular}
    }
\end{table}

We measure the latency and memory usage to assess the overhead of $\gcn$ when $\epsilon=1,K=20$ as an example. 

\subsubsection{Execution time}
For latency, we use the total execution time (training and testing) for EDP, and in Table~\ref{tab:reulst_table_running_time} we list the average of running $3$ times. 
We found $\gcn$ has less latency than baselines including $3.22$ seconds of $\gcn$ averaged among datasets, comparing to $61.8$ seconds of \dpdgc, $5.99$ seconds of \gap,  and $24.39$ seconds \mlpE, except for comparative $3.25$ seconds of \sage. Notable latency increase is observed on the \facebook\ dataset because its large number of nodes (26,406) and edges (2,117,924), and $\gcn$ shows a larger lead comparing the other private GNNs: $233.65$ and $6.36$ for \dpdgc\  and \gap, respectively. 

As for NDP, among all the evaluated models, $\gcn$ consistently achieves the lowest execution times across all datasets, with an average of only $2.77$ seconds per training-test run. This is substantially faster than the other methods, particularly in comparison to resource-intensive models like DPDGC and MLP, which have average execution times of $467.52$ seconds and $633.57$ seconds, respectively. $\gcn$’s ability to process both large and small datasets demonstrates its potential advantage for real-world applications where computational resources or time may be limited.  

\subsubsection{Memory costs}
For memory, we measure the max memory usage over the 3 runs and report the results in Table~\ref{tab:reulst_table_memory} under EDP. For chain-structured datasets, the memory consumption averaged across datasets of $\gcn$ (1.09MB) is close to \sage\ (1.12MB) and \mlpE\ (1.05MB), while much lower than \gap\ (1.55MB) and \dpdgc\ (1.55MB). On the other hand, we found that a higher memory consumption of $\gcn$ is introduced in the standard graph datasets, compared to the baseline methods due to the graph loading mechanism in \gap and controllable parameters for contractiveness. 
Still, the memory consumption of $\gcn$ is reasonable, which is smaller than our GPU memory limit (24GB) for all datasets. We argue the privacy, and utility benefits of $\gcn$ outweigh its memory costs here.
 
\begin{table}[H]
\label{tab:ndp_acc}
\small
    \caption{Max Memory Usage (MB) for EDP}\label{tab:reulst_table_memory}
    \vspace{1pt}
    \centering
    \setlength\tabcolsep{3pt}
    \adjustbox{max width=\textwidth}{
    \begin{tabular}{c|c|c|c|c|c}
    \toprule[1pt]
    \textbf{Dataset} & \textbf{$\gcn$} & \dpdgc& \gap & \sage & \mlpE \\
    \cmidrule[1pt]{1-6}
    \computers & $314.04$ & $1833.63$ & $138.40$ & $154.57$ & $129.27$ \\
    \facebook & $525.18$ & $3657.00$ & $315.84$ & $359.63$ & $251.24$ \\
    \pubmed & $272.70$ & $2594.26$ & $190.12$ & $167.64$ & $98.26$ \\
    \cora & $104.94$ & $377.78$ & $27.55$ & $32.85$ & $77.01$ \\
    \photo & $170.28$ & $1021.70$ & $77.54$ & $83.49$ & $80.14$ \\\hline
        \chainS & $1.08$ & $1.34$ & $1.41$ & $1.10$ & $1.03$ \\
    \chainM & $1.09$ & $1.42$ & $1.46$ & $1.11$ & $1.03$ \\
    \chainL & $1.09$ & $1.69$ & $1.58$ & $1.13$ & $1.03$ \\
    \chainX & $1.11$ & $2.42$ & $1.82$ & $1.17$ & $1.04$ \\
    \bottomrule
    \end{tabular}
    }
\end{table}

\begin{table}[H]
\label{tab:ndp_acc}
\small
    \caption{Max Memory Usage (MB) for NDP. Max node degree is 20.}\label{tab:reulst_table_memory_ndp}
    \vspace{1pt}
    \centering
    \setlength\tabcolsep{3pt}
    \adjustbox{max width=\textwidth}{
    \begin{tabular}{c|c|c|c|c|c}
    \toprule[1pt]
    \textbf{Dataset} & \textbf{$\gcn$} & \dpdgc& \gap & \sage & \mlpE \\
    \cmidrule[1pt]{1-6}
    \computers & $357.28$ & $1831.35$ & $319.80$ & $1621.09$ & $130.88$  \\
    \facebook & $531.76$ & $3484.93$ & $479.18$ & $1676.24$ & $223.17$  \\
    \pubmed & $314.45$ & $2597.45$ & $345.10$ & $577.42$ & $101.49$    \\
    \cora & $121.39$ & $469.11$ & $299.51$ & $890.55$ & $78.15$     \\
    \photo & $195.02$ & $1051.85$ & $239.85$ & $1352.86$ & $81.29$   \\\hline
     \chainS & $1.40$ & $1.37$ & $2.13$ & $1.10$ & $1.09$     \\
    \chainM & $1.41$ & $1.44$ & $2.32$ & $1.11$ & $1.10$     \\
    \chainL & $1.42$ & $1.71$ & $2.79$ & $1.13$ & $1.11$   \\
    \chainX & $1.46$ & $2.44$ & $3.72$ & $1.19$ & $1.13$  \\
    \bottomrule
    \end{tabular}
    }
\end{table}

\begin{table*}[!h]
\caption{Privacy Auditing via LinkTeller and G-MIA. AUC score is reported.}
\setlength\tabcolsep{9pt}
\begin{center}
\adjustbox{max width=\textwidth}{
\begin{tabular}{l|c|c|c|c|c|c|c|c|c|c|c|c|c|c}
  \toprule[1pt]
\multirow{2}{*}{\textbf{Dataset}} & \multicolumn{7}{c|}{\textbf{LinkTeller}} & \multicolumn{7}{c}{\textbf{G-MIA}} \\
 & \textbf{$\epsilon=\inf$} & \textbf{$\epsilon=1$} & \textbf{$\epsilon=2$} & \textbf{$\epsilon=4$} & \textbf{$\epsilon=8$} & \textbf{$\epsilon=16$} & \textbf{$\epsilon=32$} & \textbf{$\epsilon=\inf$} & \textbf{$\epsilon=1$} & \textbf{$\epsilon=2$} & \textbf{$\epsilon=4$} & \textbf{$\epsilon=8$} & \textbf{$\epsilon=16$} & \textbf{$\epsilon=32$} \\ \cmidrule[1pt]{1-15}
\facebook&$0.977$&$0.463$&$0.483$&$0.482$&$0.478$&$0.472$&$0.462$&$0.567$&$0.587$&$0.587$&$0.587$&$0.583$&$0.587$&$0.588$\\\hline
\pubmed&$0.981$&$0.446$&$0.442$&$0.441$&$0.445$&$0.442$&$0.449$&$0.600$&$0.599$&$0.598$&$0.598$&$0.601$&$0.601$&$0.605$\\\hline
\cora&$0.998$&$0.427$&$0.448$&$0.399$&$0.443$&$0.451$&$0.450$&$0.645$&$0.500$&$0.500$&$0.500$&$0.500$&$0.500$&$0.500$\\\hline
\photo&$0.962$&$0.475$&$0.404$&$0.421$&$0.428$&$0.417$&$0.434$&$0.678$&$0.677$&$0.682$&$0.682$&$0.678$&$0.672$&$0.676$\\  \hline
\computers&$0.860$&$0.367$&$0.364$&$0.372$&$0.363$&$0.361$&$0.384$&$0.702$&$0.701$&$0.661$&$0.707$&$0.708$&$0.711$&$0.701$\\ \bottomrule
\end{tabular}}
\end{center}
\label{tab:privacy_audit}
\end{table*}

Table~\ref{tab:reulst_table_memory_ndp} presents the maximum memory usage (in MB) for NDP across a variety of datasets with a maximum node degree of 20, highlighting the efficiency of our method, $\gcn$. Across all datasets, $\gcn$ demonstrates significantly lower memory consumption compared to \dpdgc, which incurs the highest memory usage, particularly on large-scale datasets like \facebook\  (3484.93 MB for \dpdgc\  vs. 531.75 MB for $\gcn$) and \pubmed\  (2597.45 MB for \dpdgc\  vs. 314.44 MB for $\gcn$). While baseline methods such as \gap\  and \mlpE\  show lower memory usage on certain datasets, $\gcn$ achieves a favorable balance of efficiency and scalability, maintaining competitive memory consumption across both small and large graphs. Notably, on the \cora\  and \photo\  datasets, $\gcn$ uses 104.51 MB and 165.86 MB, respectively, which is substantially less than \dpdgc and only moderately higher than the most memory-efficient baselines.
Notably,  \sage\ incurs much larger memory overhead under NDP than that of EDP settings due to its $O(|\mV|^2)$ computational complexity.
Overall, these results illustrate that $\gcn$ effectively controls memory usage, outperforming \dpdgc\  by a wide margin and providing robust scalability under NDP settings.

\subsection{More Results of EDP and NDP}
\label{app:endp_supple}

\subsubsection{Analysis of EDP} 
Regarding the standard graph datasets, for \computers, \pubmed, \cora\ and \photo, $\gcn$ can outperform all the other baselines in most cases with varying $\epsilon$. 
It is also worth noting that $\gcn$ achieves fairly stable accuracy with with varying $\epsilon$ (e.g., $92.0\%$ to $92.4\%$ for \computers\ with $\epsilon$ growing from $1$ to $\inf$), suggesting the convergent privacy design is indeed effective. For \facebook, we found $\gcn$ outperforms the other baselines with small $\epsilon$ (i.e., $\epsilon \leq 4$), but falls behind at larger $\epsilon$. 
Yet, we argue that performance at low privacy budget is more critical for privacy-sensitive datasets like \facebook, and stable accuracy across privacy budget $\epsilon$ might be more desired.

Now, we highlight the results on the chain-structured datasets. In summary, we observe that the non-private versions (i.e., $\epsilon=\inf$) of \gap\ and $\gcn$ can achieve perfect classification accuracy, 
while \dpdgc, \sage\ and \mlpE\ show relatively low accuracy. 
The root cause is that features of a distant node can only be learned through a sufficient number of hops of message passing, \ie, aggregating and passing features of the neighboring node one by one, but \dpdgc\ and \sage\ are not designed to leverage long-hop interactions. For \mlpE, as the graph structure is not utilized at all, the result is close to random guessing.
After adding DP noises to protect edge-level memberships, the accuracy of both \gap\ and $\gcn$ drop, but the accuracy loss of \gap\ is more prominent (from $100\%$ down to $57.8\%$-$62.5\%$). \gap\ is able to maintain reasonable accuracy even with fairly tight privacy budget $\epsilon=1$ on small chain-structured datasets like \chainS\ and \chainM. On the other hand, \chainL\ and \chainX\ have relatively low accuracy even for $\gcn$ (e.g., $70.0\%$ and $66.0\%$ when $\epsilon=1$).
This observed phenomenon is attributed to the increased difficulty of GNN classification with longer chains, as DP noise is injected at every layer, compounding its impact.




\subsubsection{Analysis of NDP} 
 In line with the EDP results, the NDP results also demonstrate that precision increases in general with larger $\epsilon$ values.
As established in Theorem~\ref{thm:node-sensitivity}, NDP requires injecting more noise compared to EDP under the same privacy budget, hence, the accuracy of NDP is often lower than EDP for standard datasets. 
This phenomenon does not hold for chain-structured datasets as the original max node degree is very small, \ie, $0,1,2$. 
Ablation study on configuring different max node degree for NDP is in Section~\ref{sec:abla_degreenode}.

Notably, $\gcn$ consistently outperforms all other methods across all privacy budgets $\epsilon$ and across all nine datasets, while \dpdgc\ shows the second-best performance overall but remains significantly behind $\gcn$. 
For instance, $\gcn$ achieves $91.91\%$, which is $35\%$+ higher than the next best \mlpE\  of $63.44\%$,  $54\%$ higher than \dpdgc\  of $56.72\%$, and more than $60\%$ higher than \sage with $29.51\%$.
\mlpE represents the baseline of learning node features independently without graph topology.
In particular, $\gcn$ is the only framework that can surpass \mlpE in many cases, showing effective GNN learning over structural graphs.
In addition, \gap, \sage, and \mlpE\  generally yield lower values as the realization of NDP is more challenging than the that of EDP. 
Thus, Table~\ref{tab:reulst_table_overall_acc_top1} indicates strong and stable advantage of $\gcn$ over state-of-the-art baselines in NDP settings.

\subsubsection{More Results}

Table~\ref{tab:reulst_table_overall_acc_mean} shows the mean accuracy compared between $\gcn$ and other baselines across all datasets. 
For both EDP and NDP tasks, $\gcn$ always achieves the highest or second-highest accuracy, often outperforming competing methods by a notable margin. 
This is especially apparent in the NDP task, where the performance gap widens under stricter privacy budgets (lower $\epsilon$), demonstrating $\gcn$’s robustness in privacy-sensitive regimes.
 For example, under NDP with $\epsilon = 1$ on the \cora\ dataset, $\gcn$ achieves an accuracy of $67.47\%$, while the next best method, DPDGC, reaches just $31.04\%$ with a difference of over $50\%$ drop.
In the chain-structured datasets, $\gcn$ maintains strong performance even as the privacy constraints increase, whereas other models often see substantial accuracy drops. 
Notably, in standard graph datasets and higher privacy budgets, $\gcn$'s accuracy advantage is frequently greater than $10\%$ compared to the next best approach, as indicated by green arrows in the table. 
Even in the non-private setting, $\gcn$ delivers near-optimal or best accuracy across all datasets. These results highlight $\gcn$'s overall reliability, adaptability, and superior utility for both private and non-private graph learning tasks.

Figure~\ref{fig:non-priv} demonstrates the learning behavior on chain-structured datasets, where increasing the number of aggregation hops ($K$) leads to continuous improvement in accuracy for both \gap\ and $\gcn$ in general. This indicates that deeper message passing is beneficial for capturing information on long chains.
For the first $15$ hops, $\gcn$ outperforms \gap, suggesting it is more effective to aggregate node features at smaller range.

\begin{figure}[H]
    \subfloat[Chain-S]{
		\includegraphics[width = 2.5cm]{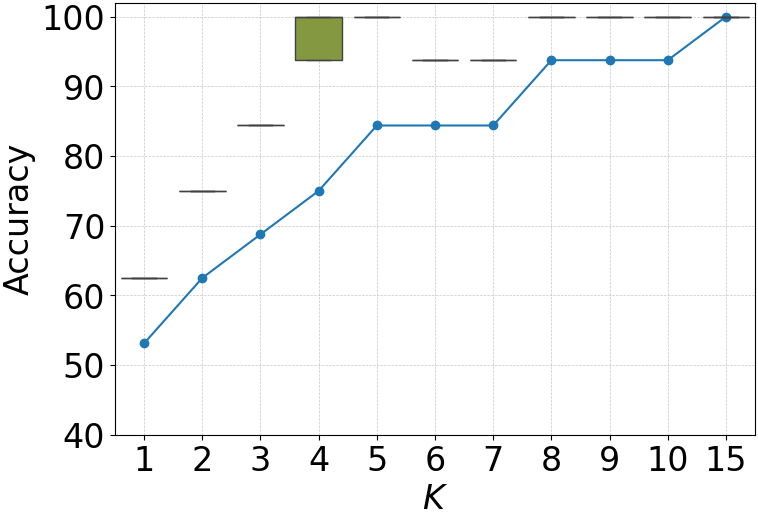}
		\label{img:acc5_chains1_inf}
		}
    \subfloat[Chain-M]{
		\includegraphics[width = 2.5cm]	{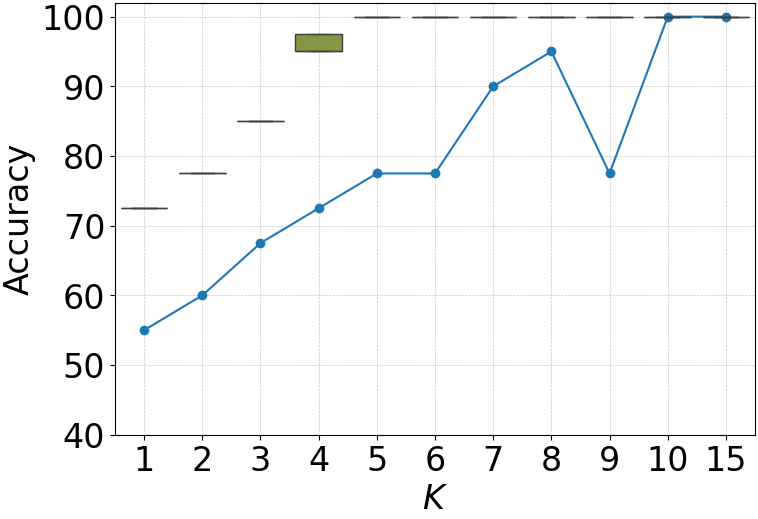}
	\label{img:acc5_chains2_inf}
		} 
    \subfloat[Chain-L]{
		\includegraphics[width = 2.5cm]{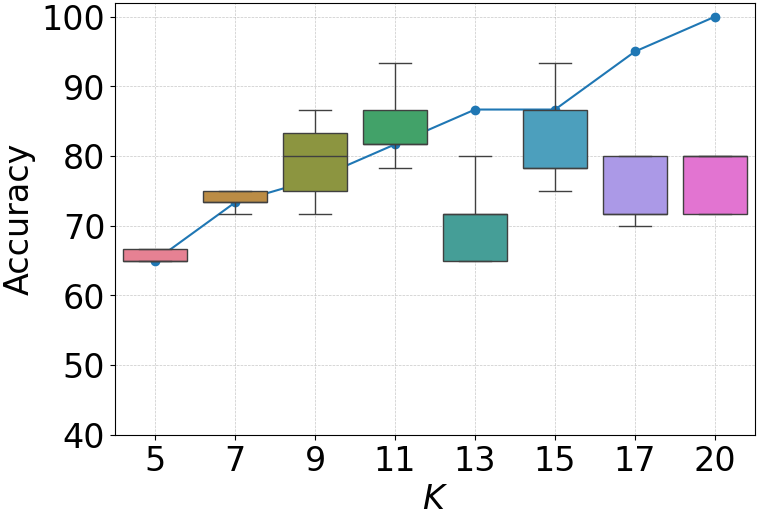}
	\label{img:acc5_chains3_inf}
		}  
	\caption{Accuracy of non-private $\gcn$ and \gap\ under different $K$.}
	\label{fig:non-priv}
\end{figure}

Figure~\ref{fig:k-study} and Figure~\ref{fig:K-study} provide the complete ablation study results on $K$ and $\epsilon$.

\begin{figure*}[h]
\subfloat[$\epsilon=1$ (\photo)]{
		\includegraphics[width = 2.6cm]{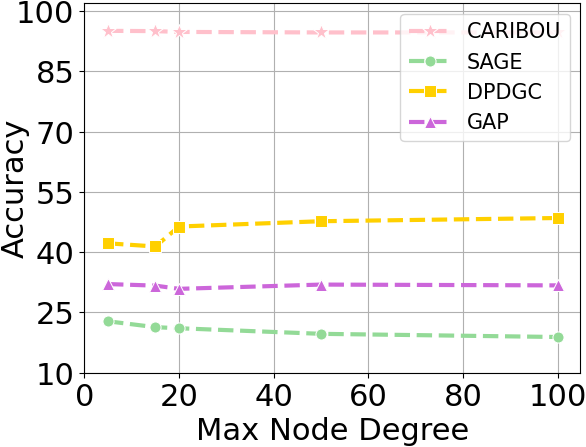}
		\label{fig:acc_nodedegree_e1_photo}
		}
\subfloat[$\epsilon=2$ (\photo)]{
		\includegraphics[width = 2.6cm]{./img/acc_nodedegree_e2_photo.png}
		\label{fig:acc_nodedegree_e2_photo}
		}
\subfloat[$\epsilon=4$ (\photo)]{
		\includegraphics[width = 2.6cm]{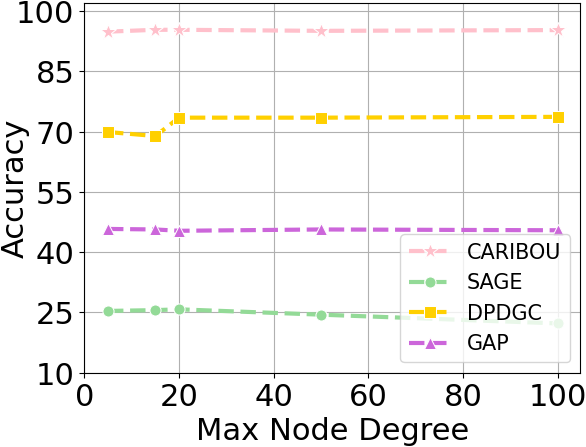}
		\label{fig:acc_nodedegree_e4_photo}
		}
\subfloat[$\epsilon=8$ (\photo)]{
		\includegraphics[width = 2.6cm]{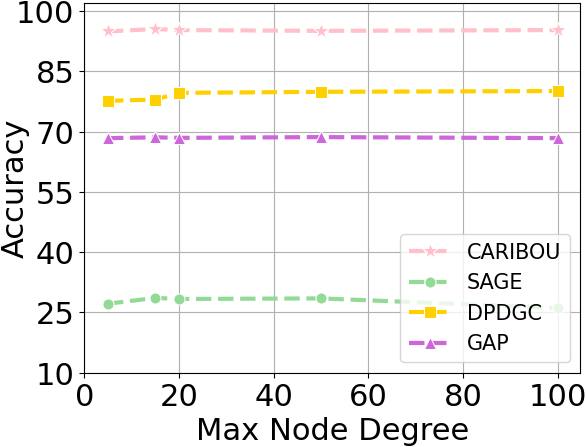}
		\label{fig:acc_nodedegree_e8_photo}
		} 
\subfloat[$\epsilon=16$ (\photo)]{
		\includegraphics[width = 2.6cm]{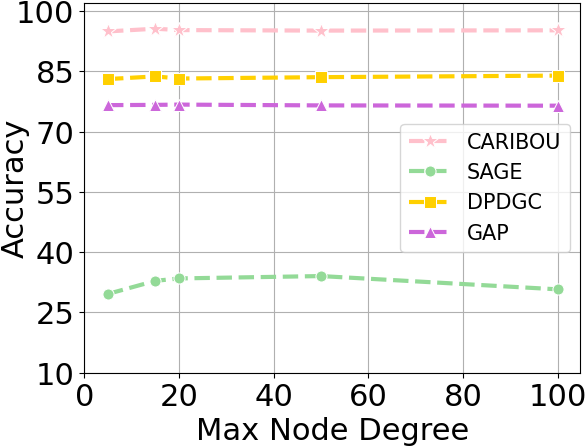}
		\label{fig:acc_nodedegree_e16_photo}
		}
\subfloat[$\epsilon=32$ (\photo)]{
		\includegraphics[width = 2.6cm]{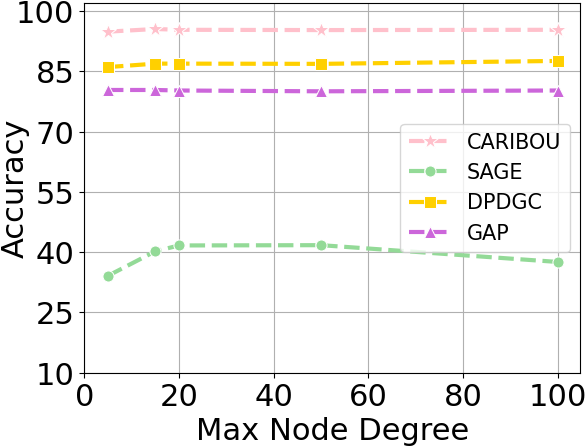}
		\label{fig:acc_nodedegree_e32_photo}
		}

\subfloat[$\epsilon=1$ (\chainS)]{
		\includegraphics[width = 2.6cm]{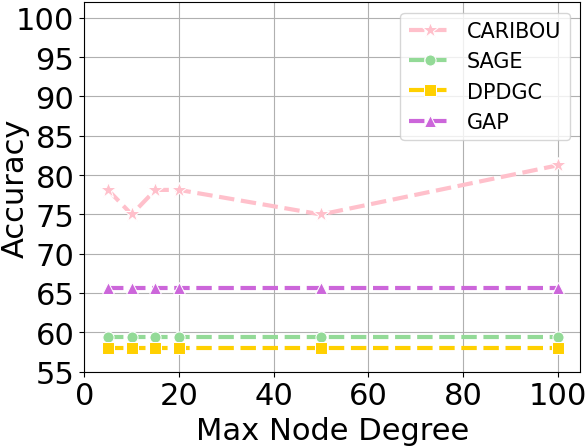}
		\label{fig:acc_nodedegree_e1_chains1}
		}
\subfloat[$\epsilon=2$ (\chainS)]{
		\includegraphics[width = 2.6cm]{./img/acc_nodedegree_e2_chains1.png}
		\label{fig:acc_nodedegree_e2_chains1}
		}
\subfloat[$\epsilon=4$ (\chainS)]{
		\includegraphics[width = 2.6cm]{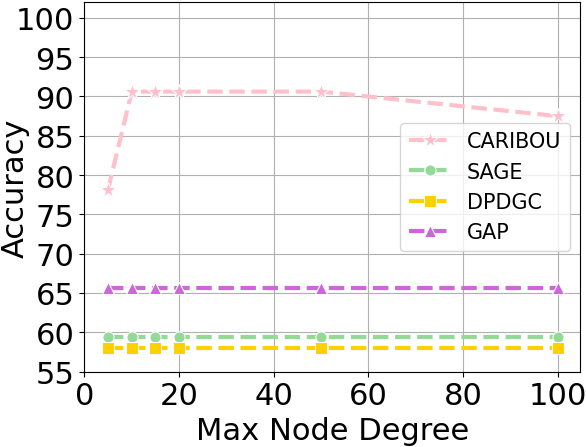}
		\label{fig:acc_nodedegree_e4_chains1}
		}
\subfloat[$\epsilon=8$ (\chainS)]{
		\includegraphics[width = 2.6cm]{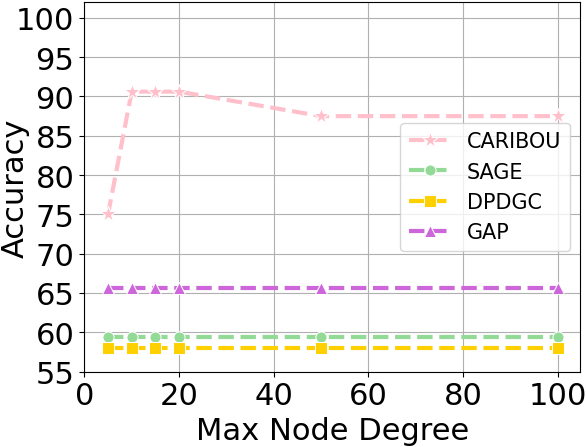}
		\label{fig:acc_nodedegree_e8_chains1}
		} 
\subfloat[$\epsilon=16$ (\chainS)]{
		\includegraphics[width = 2.6cm]{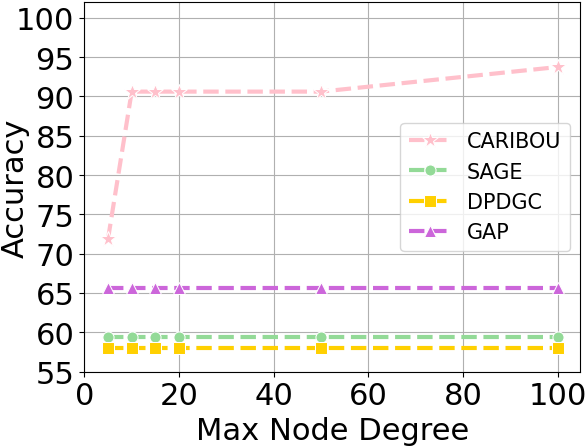}
		\label{fig:acc_nodedegree_e16_chains1}
		}
\subfloat[$\epsilon=32$ (\chainS)]{
		\includegraphics[width = 2.6cm]{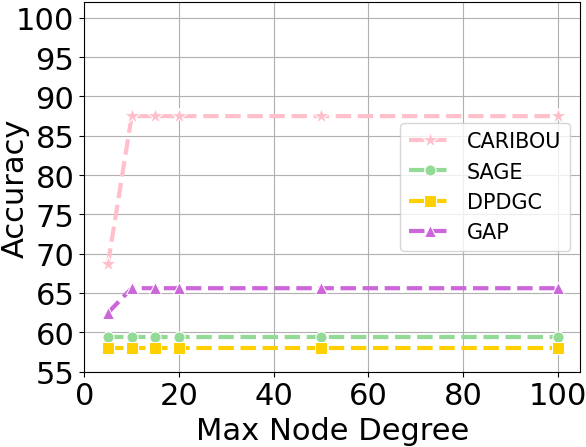}
		\label{fig:acc_nodedegree_e32_chains1}
		}
\caption{Ablation study of max node accuracy for NDP.}
        \label{fig:abla_nodedeg}
\end{figure*}

\begin{figure*}[h]
     \subfloat[$\epsilon=1$ (\chainS)]{
		\includegraphics[width = 2.6cm]{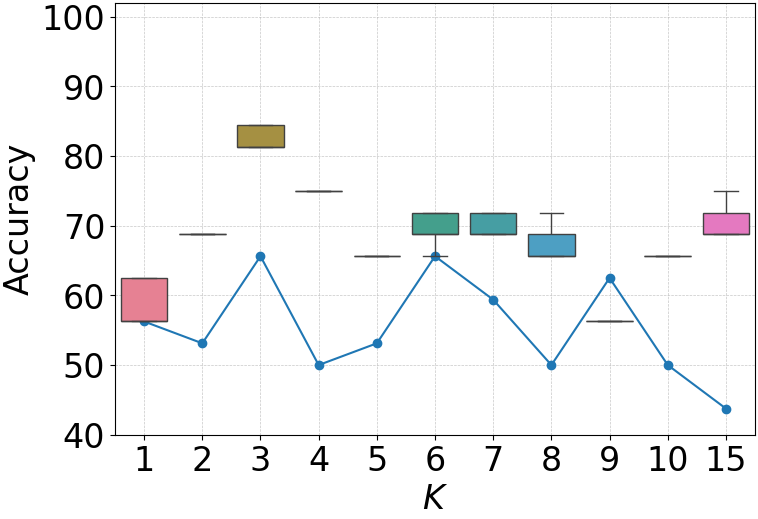}
		\label{fig:acc5_chains1_e1_1}
		}
        \subfloat[$\epsilon=2$ (\chainS)]{
		\includegraphics[width = 2.6cm]{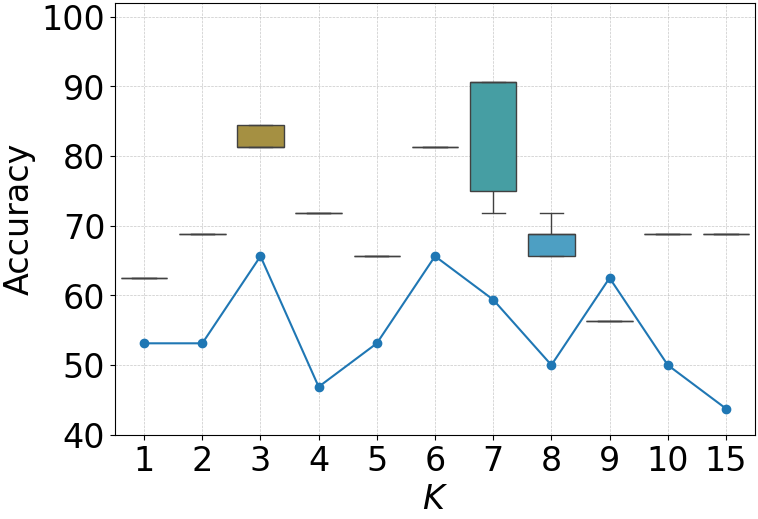}
		\label{fig:acc5_chains1_e2}
		}
		\subfloat[$\epsilon=4$ (\chainS)]{
		\includegraphics[width = 2.6cm]{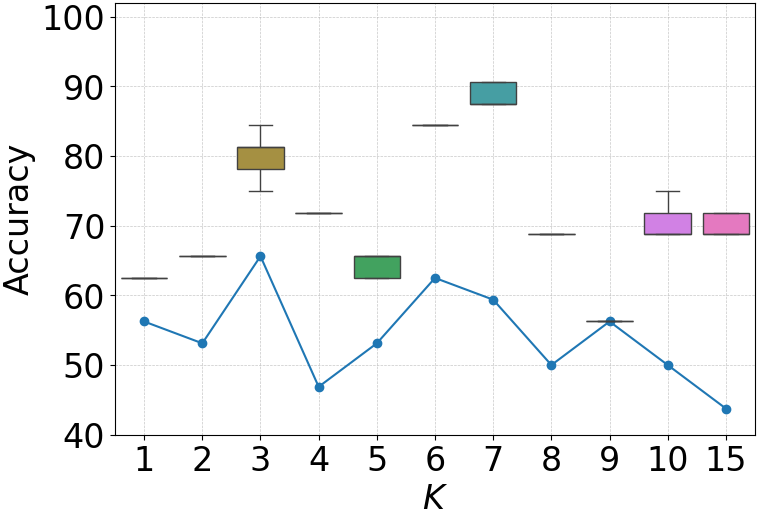}
		\label{fig:acc5_chains1_e4}
		}
		\subfloat[$\epsilon=8$ (\chainS)]{
		\includegraphics[width = 2.6cm]{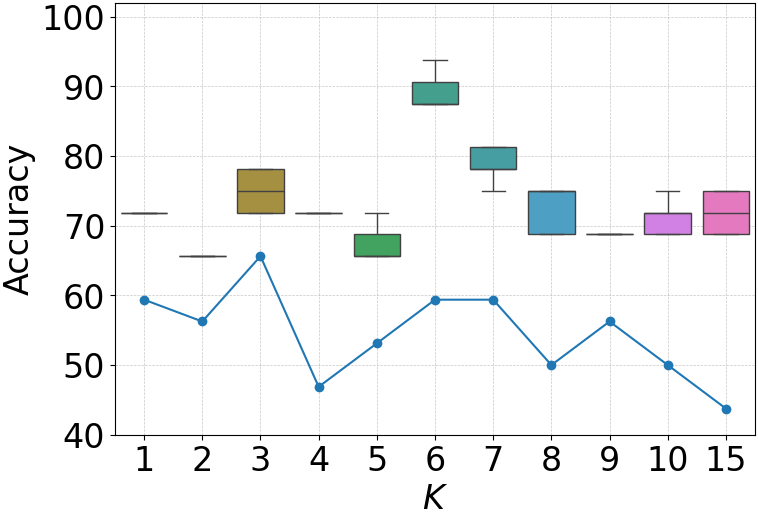}
		\label{fig:acc5_chains1_e8}
		}
		\subfloat[$\epsilon=16$ (\chainS)]{
		\includegraphics[width = 2.6cm]{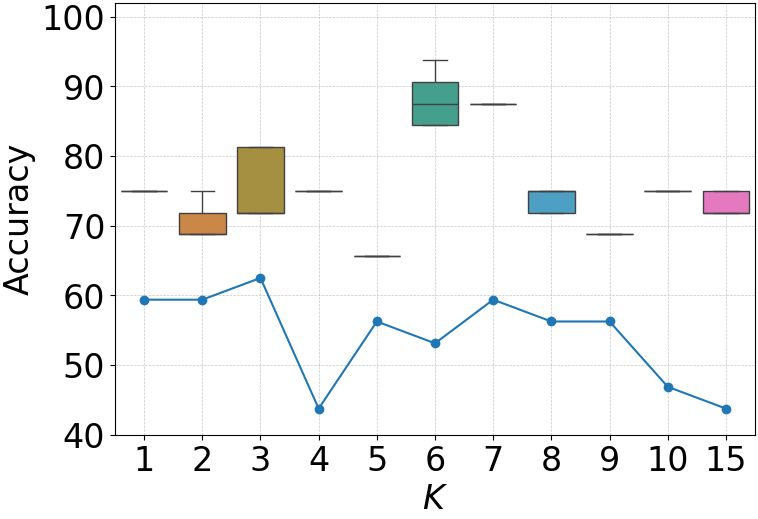}
		\label{fig:acc5_chains1_e16}
		}
		\subfloat[$\epsilon=32$ (\chainS)]{
		\includegraphics[width = 2.6cm]{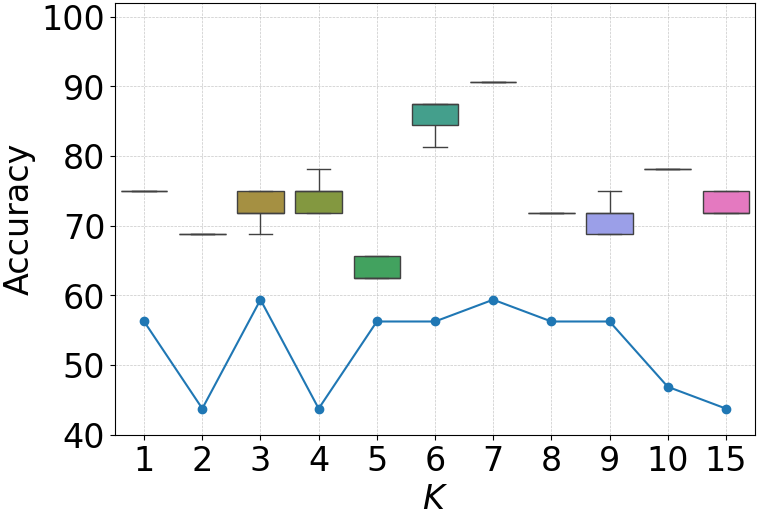}
		\label{fig:acc5_chains1_e32}
		}

    		\subfloat[$\epsilon=1$ (\chainM)]{
		\includegraphics[width = 2.6cm]{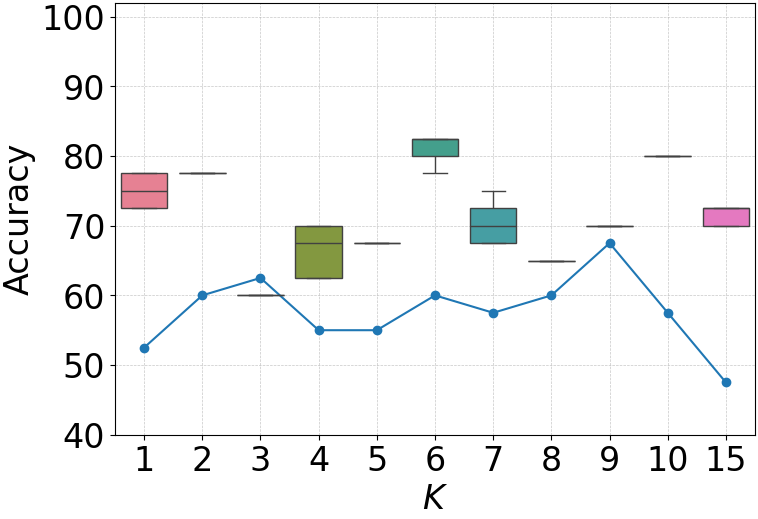}
		\label{fig:acc5_chains2_e1}
		}
		\subfloat[$\epsilon=2$ (\chainM)]{
		\includegraphics[width = 2.6cm]{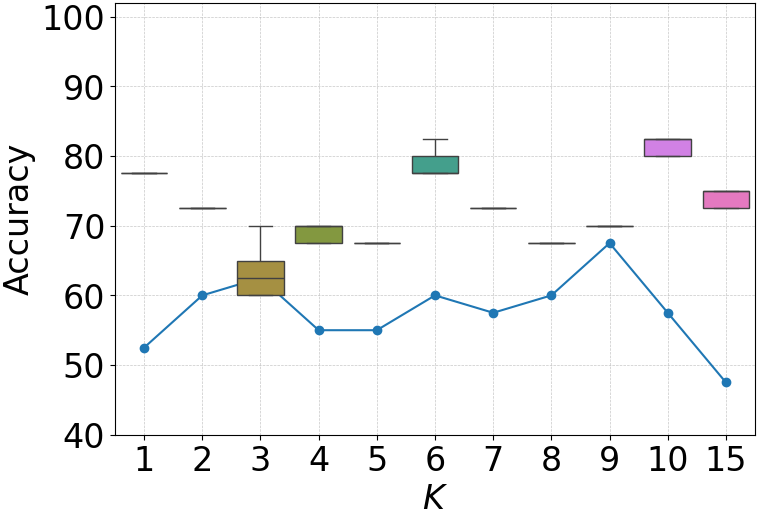}
		\label{fig:acc5_chains2_e2}
		}
		\subfloat[$\epsilon=4$ (\chainM)]{
		\includegraphics[width = 2.6cm]{./img/acc5_chains2_e4.png}
		\label{fig:acc5_chains2_e4}
		}
		\subfloat[$\epsilon=8$ (\chainM)]{
		\includegraphics[width = 2.6cm]{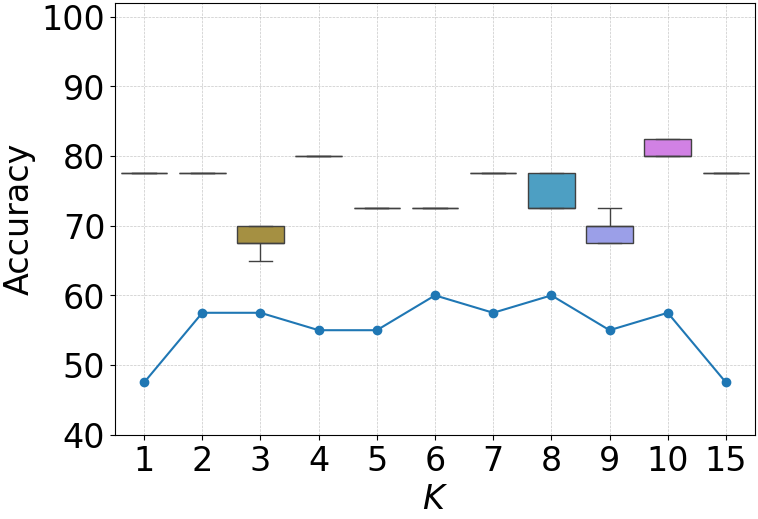}
		\label{fig:acc5_chains2_e8}
		}
		\subfloat[$\epsilon=16$ (\chainM)]{
		\includegraphics[width = 2.6cm]{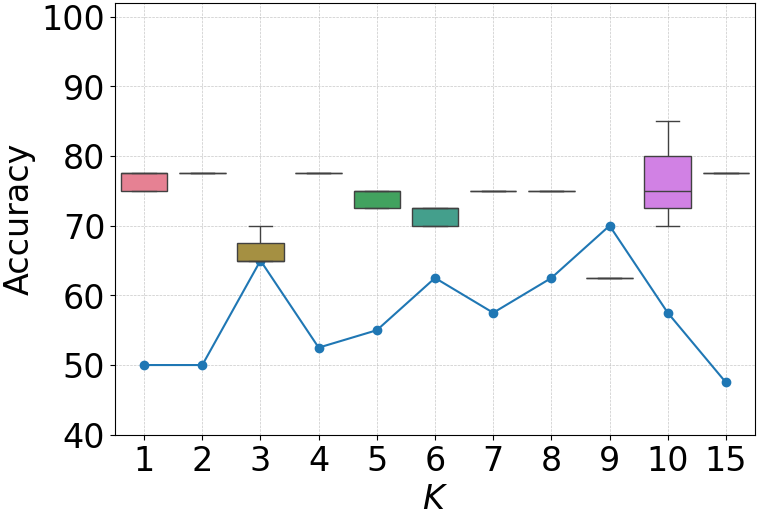}
		\label{fig:acc5_chains2_e16}
		}
		\subfloat[$\epsilon=32$ (\chainM)]{
		\includegraphics[width = 2.6cm]{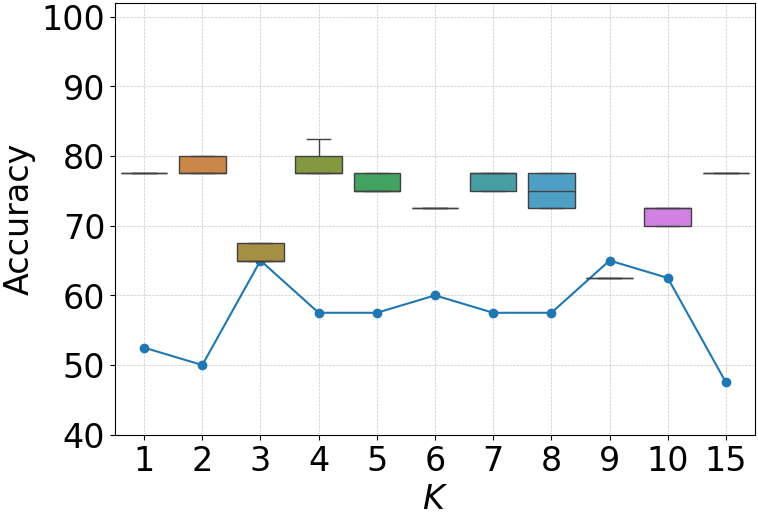}
		\label{fig:acc5_chains2_e32}
		}
        
		\subfloat[$\epsilon=1$ (\chainL)]{
		\includegraphics[width = 2.6cm]{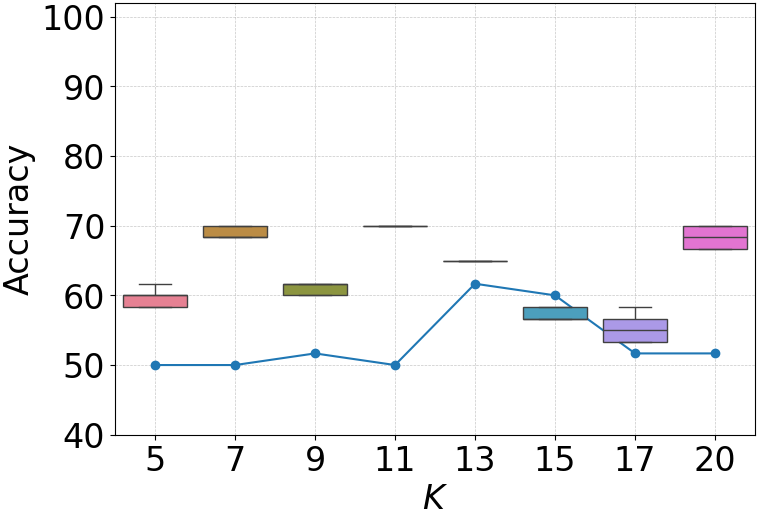}
		\label{fig:acc5_chains3_e1}
		}
		\subfloat[$\epsilon=2$ (\chainL)]{
		\includegraphics[width = 2.6cm]{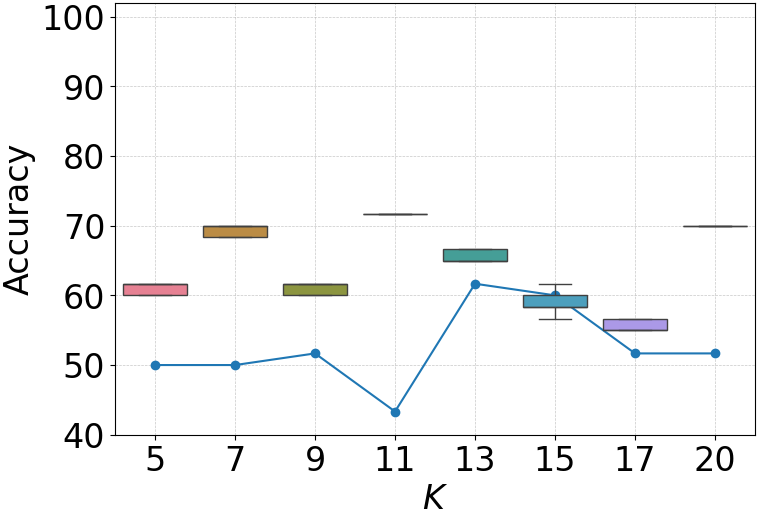}
		\label{fig:acc5_chains3_e2}
		}
		\subfloat[$\epsilon=4$ (\chainL)]{
		\includegraphics[width = 2.6cm]{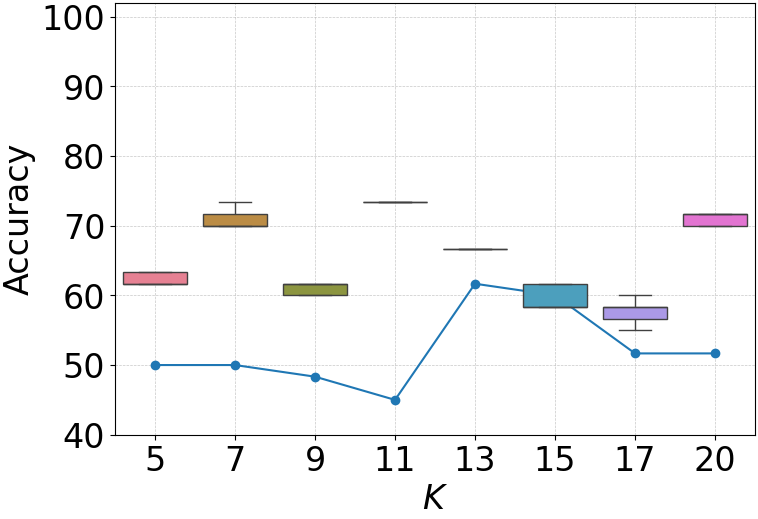}
		\label{fig:acc5_chains3_e4}
		}
		\subfloat[$\epsilon=8$ (\chainL)]{
		\includegraphics[width = 2.6cm]{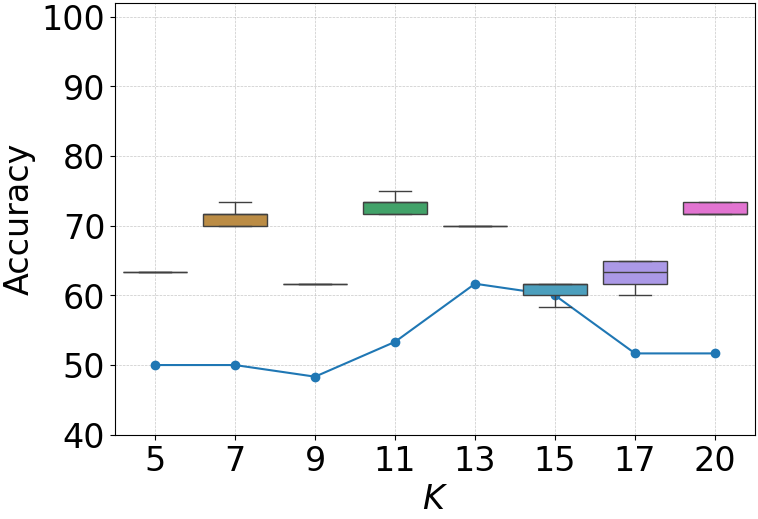}
		\label{fig:acc5_chains3_e8}
		}
		\subfloat[$\epsilon=16$ (\chainL)]{
		\includegraphics[width = 2.6cm]{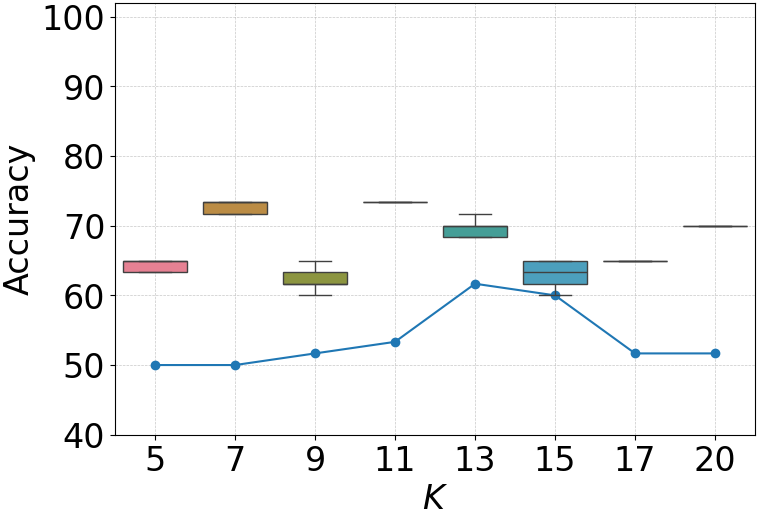}
		\label{fig:acc5_chains3_e16}
		}
		\subfloat[$\epsilon=32$ (\chainL)]{
		\includegraphics[width = 2.6cm]{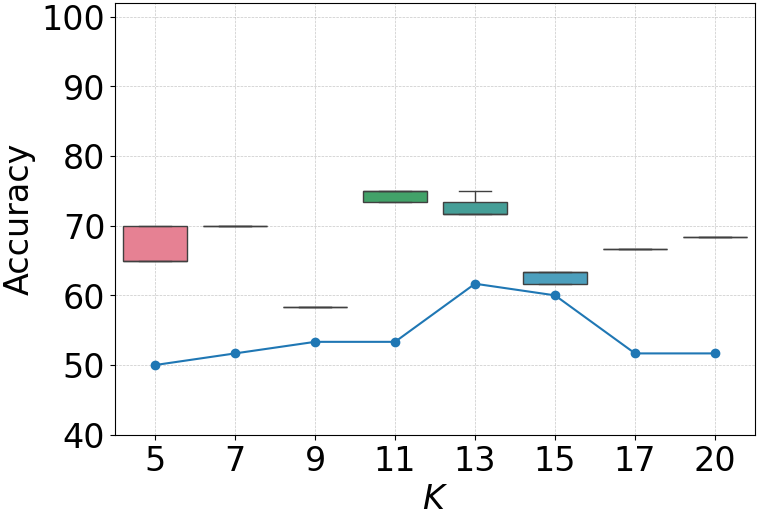}
		\label{fig:acc5_chains3_e32}
		}
             
		\caption{Ablation Study of $K$ on $\gcn$ (colored boxes) and \gap\ (blue lines). For each pair of $\epsilon$ and dataset (\chainS, \chainM\  and \chainL), different $K$ are used. }
        \label{fig:k-study}
	\end{figure*}

\begin{figure*}[h]
\begin{center}
           \subfloat[$K\hspace{-1mm}=\hspace{-1mm}1$  (\chainS)]{
		\includegraphics[width = 2.6cm]{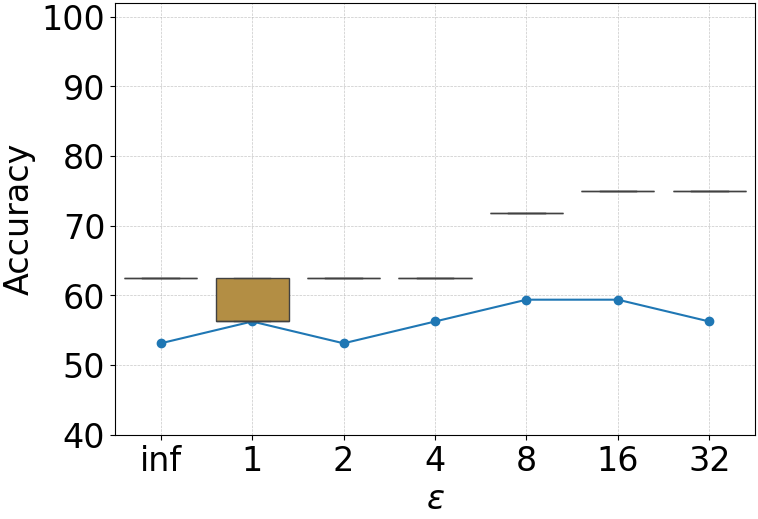}
		\label{fig:acc5_chains1_hops1}
		}
		\subfloat[$K\hspace{-1mm}=\hspace{-1mm}2$ (\chainS)]{
		\includegraphics[width = 2.6cm]{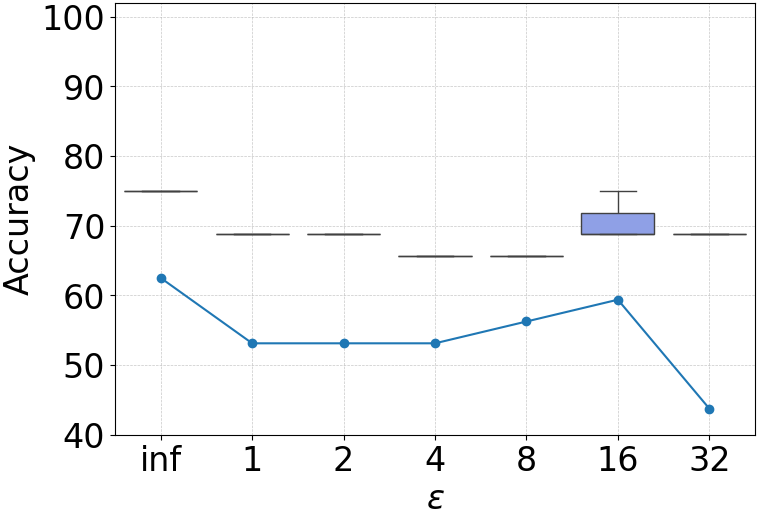}
		\label{fig:acc5_chains1_hops2}
		}
		\subfloat[$K\hspace{-1mm}=\hspace{-1mm}3$ (\chainS)]{
		\includegraphics[width = 2.6cm]{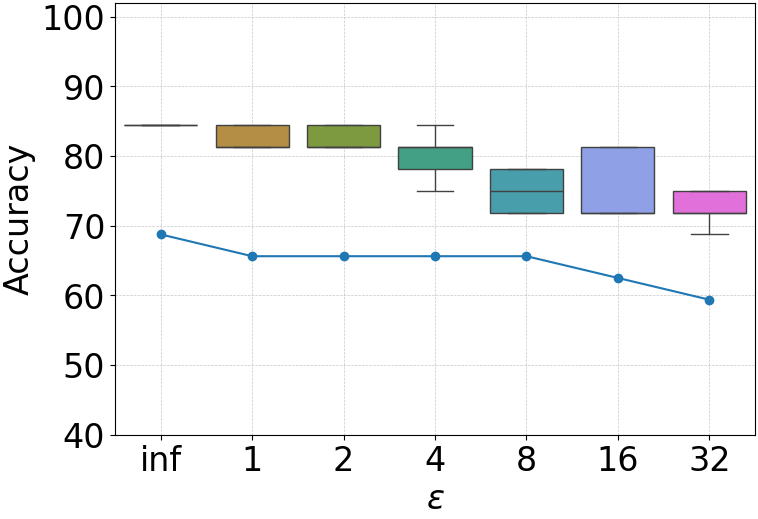}
		\label{fig:acc5_chains1_hops3}
		}
		\subfloat[$K\hspace{-1mm}=\hspace{-1mm}4$ (\chainS)]{
		\includegraphics[width = 2.6cm]{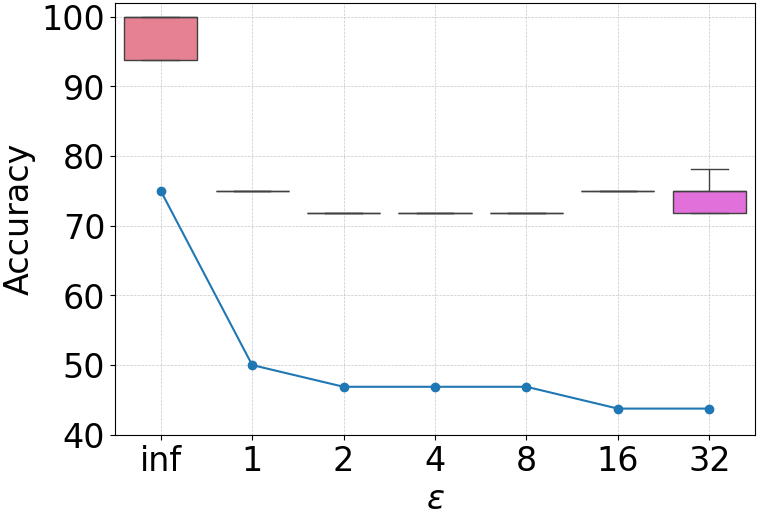}
		\label{fig:acc5_chains1_hops4}
		}
        \subfloat[$K\hspace{-1mm}=\hspace{-1mm}5$ (\chainS)]{
		\includegraphics[width = 2.6cm]{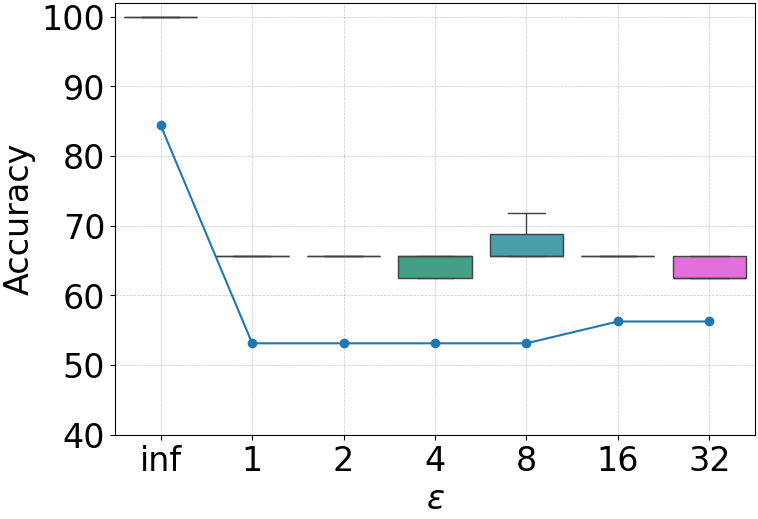}
		\label{fig:acc5_chains1_hops5}
		}
        
        \subfloat[$K\hspace{-1mm}=\hspace{-1mm}6$ (\chainS)]{
		\includegraphics[width = 2.6cm]{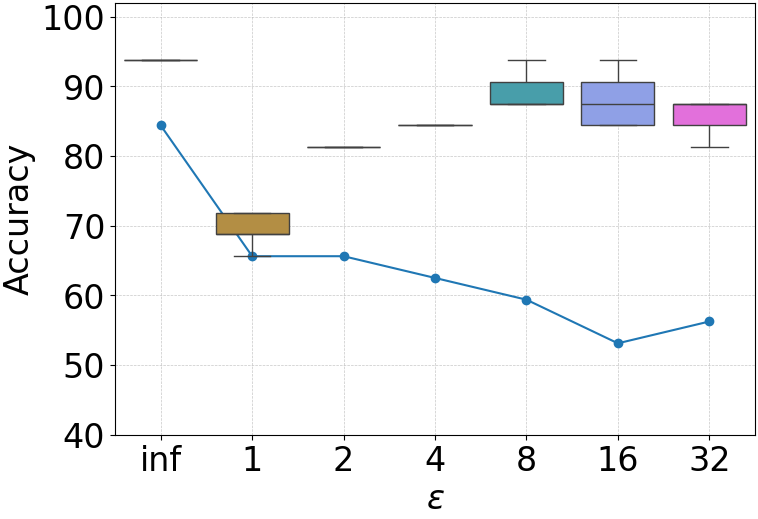}
		\label{fig:acc5_chains1_hops6}
		}
        \subfloat[$K\hspace{-1mm}=\hspace{-1mm}7$ (\chainS)]{
		\includegraphics[width = 2.6cm]{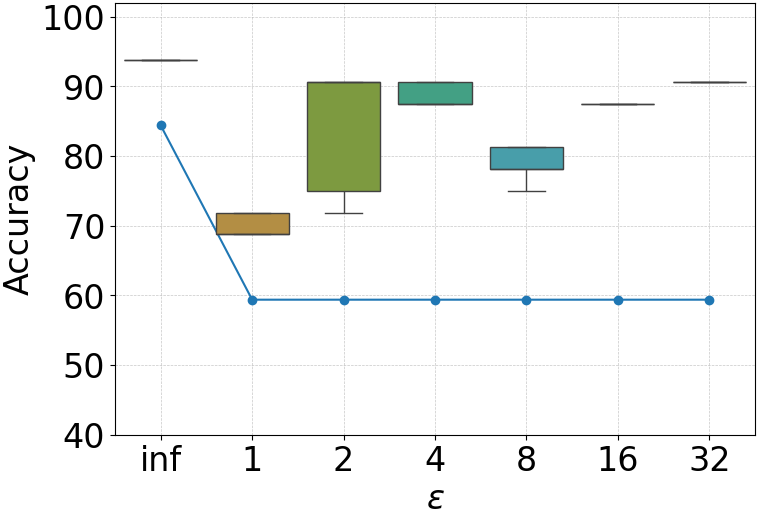}
		\label{fig:acc5_chains1_hops7}
		}
		\subfloat[$K\hspace{-1mm}=\hspace{-1mm}8$ (\chainS)]{
		\includegraphics[width = 2.6cm]{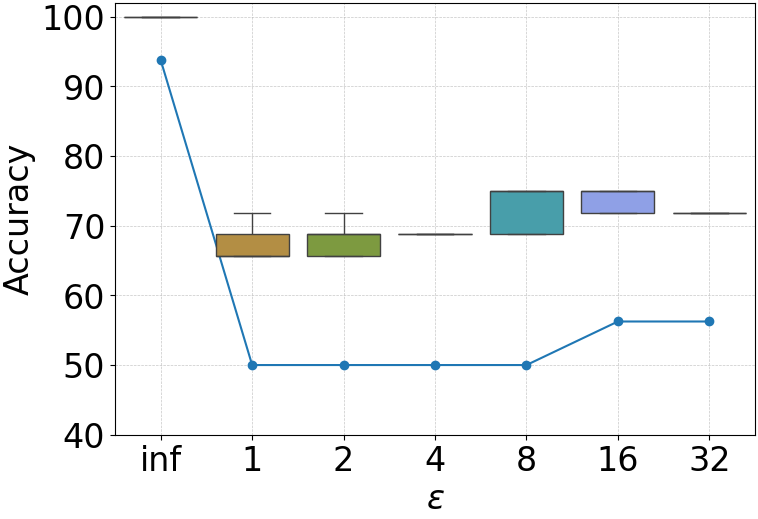}
		\label{fig:acc5_chains1_hops8}
		}
        \subfloat[$K\hspace{-1mm}=\hspace{-1mm}9$ (\chainS)]{
		\includegraphics[width = 2.6cm]{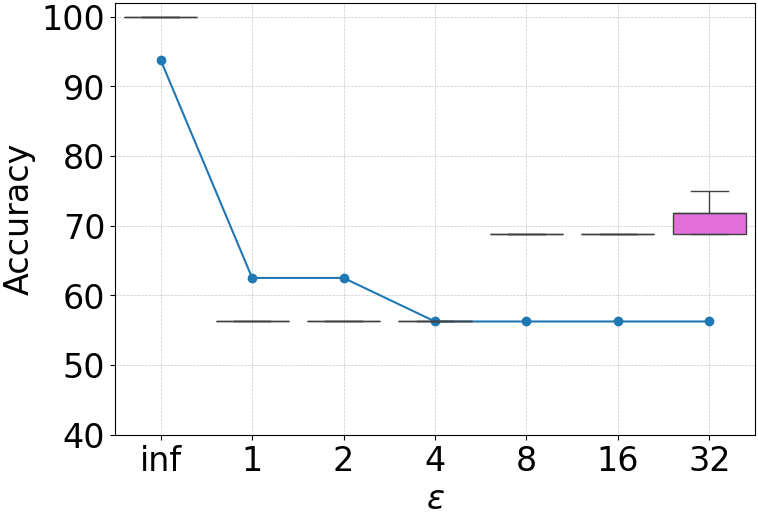}
		\label{fig:acc5_chains1_hops9}
		}
		\subfloat[$K\hspace{-1mm}=\hspace{-1mm}15$ (\chainS)]{
		\includegraphics[width = 2.6cm]{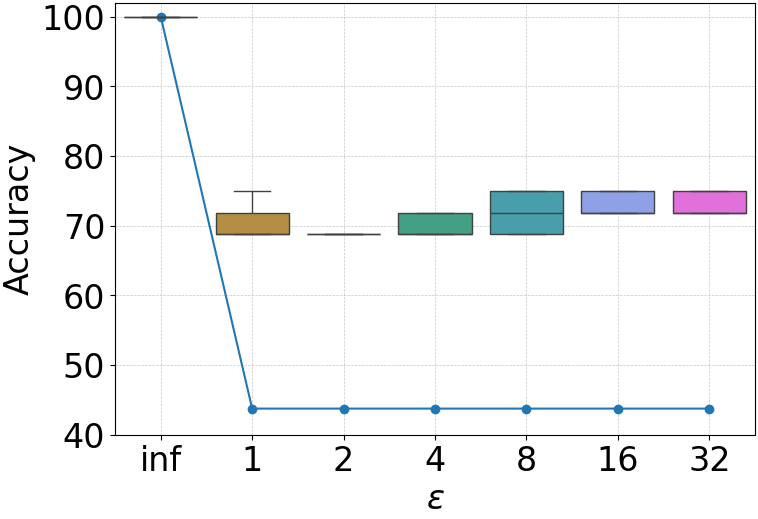}
		\label{fig:acc5_chains1_hops15}
		}

        \subfloat[$K\hspace{-1mm}=\hspace{-1mm}1$ (\chainM)]{
		\includegraphics[width = 2.6cm]{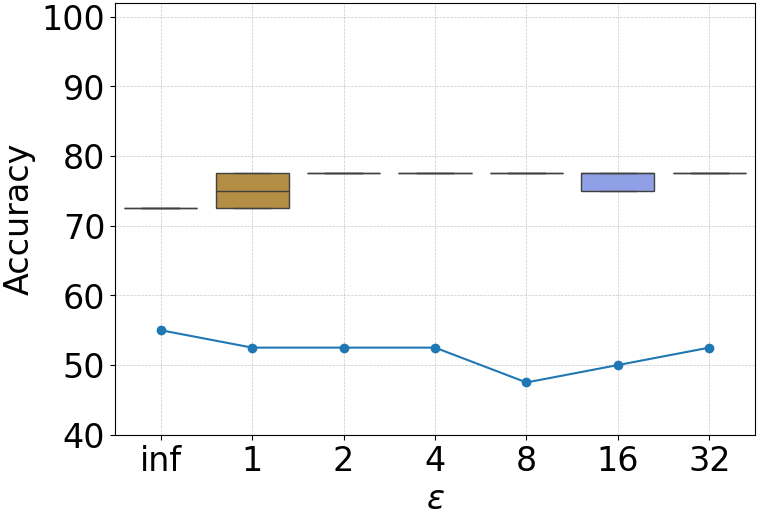}
		\label{fig:acc5_chains2_hops1}
		}
		\subfloat[$K\hspace{-1mm}=\hspace{-1mm}2$ (\chainM)]{
		\includegraphics[width = 2.6cm]{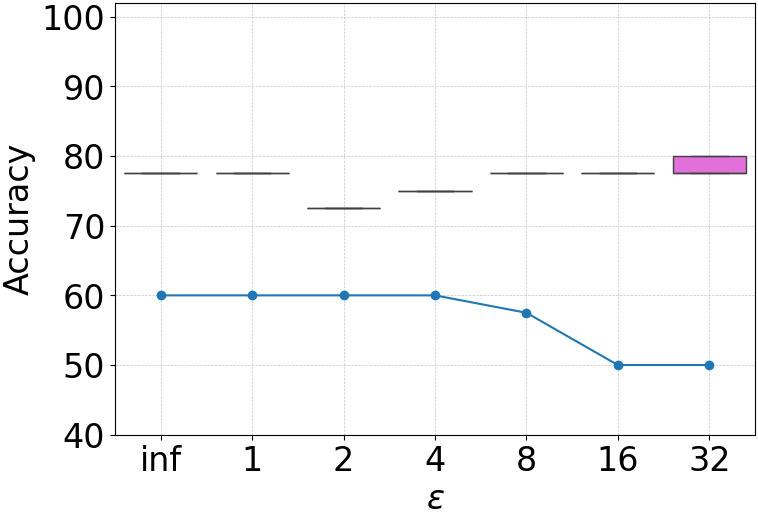}
		\label{fig:acc5_chains2_hops2}
		}
        \subfloat[$K\hspace{-1mm}=\hspace{-1mm}3$ (\chainM)]{
		\includegraphics[width = 2.6cm]{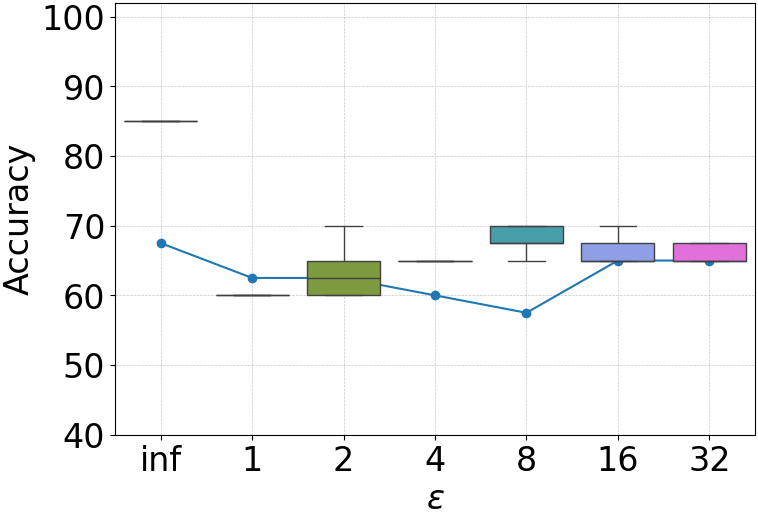}
		\label{fig:acc5_chains2_hops3}
		}
		\subfloat[$K\hspace{-1mm}=\hspace{-1mm}4$ (\chainM)]{
		\includegraphics[width = 2.6cm]{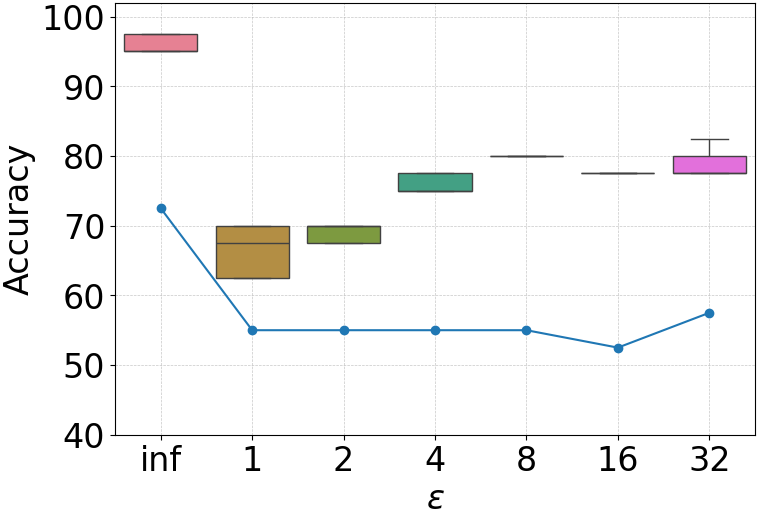}
		\label{fig:acc5_chains2_hops4}
		}
        \subfloat[$K\hspace{-1mm}=\hspace{-1mm}5$ (\chainM)]{
		\includegraphics[width = 2.6cm]{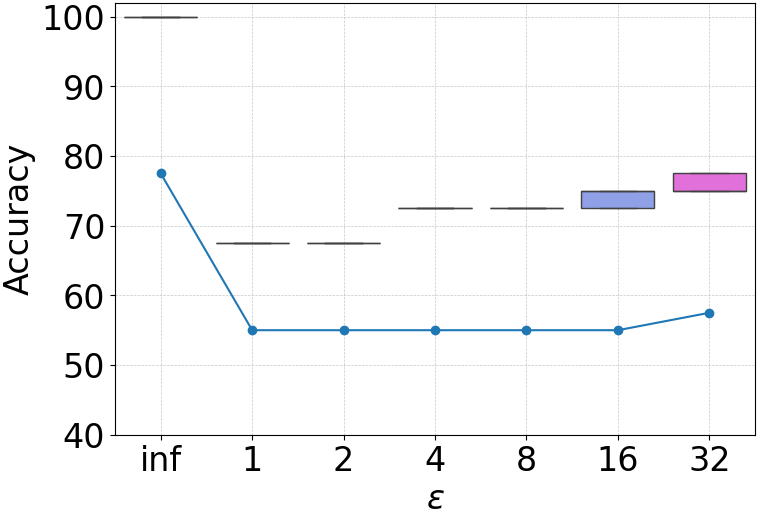}
		\label{fig:acc5_chains2_hops5}
		}
        

        \subfloat[$K\hspace{-1mm}=\hspace{-1mm}7$ (\chainM)]{
		\includegraphics[width = 2.6cm]{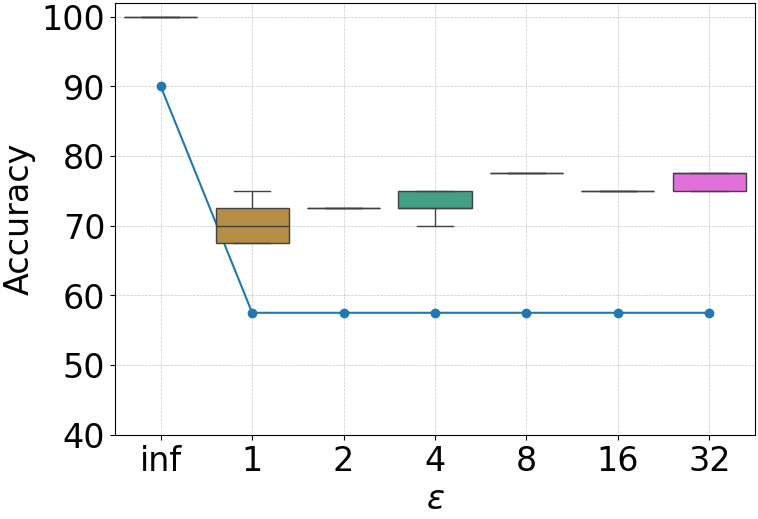}
		\label{fig:acc5_chains2_hops7}
		}
        \subfloat[$K\hspace{-1mm}=\hspace{-1mm}8$ (\chainM)]{
		\includegraphics[width = 2.6cm]{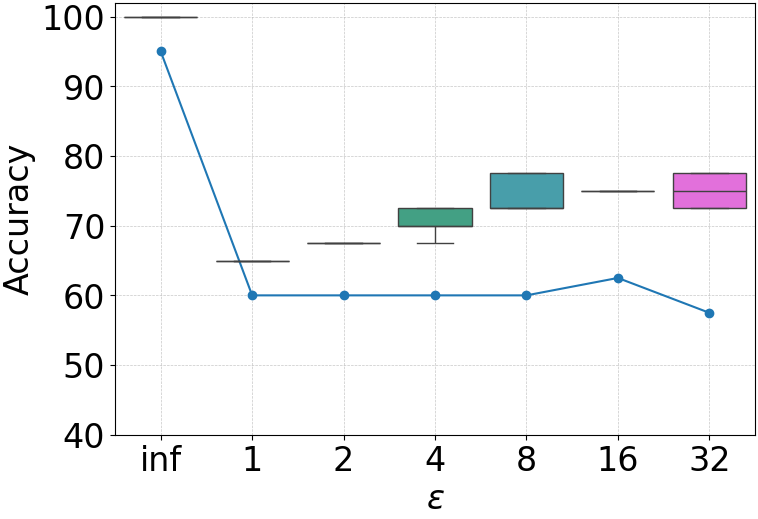}
		\label{fig:acc5_chains2_hops8}
        }
       \subfloat[$K\hspace{-1mm}=\hspace{-1mm}9$ (\chainM)]{
		\includegraphics[width = 2.6cm]{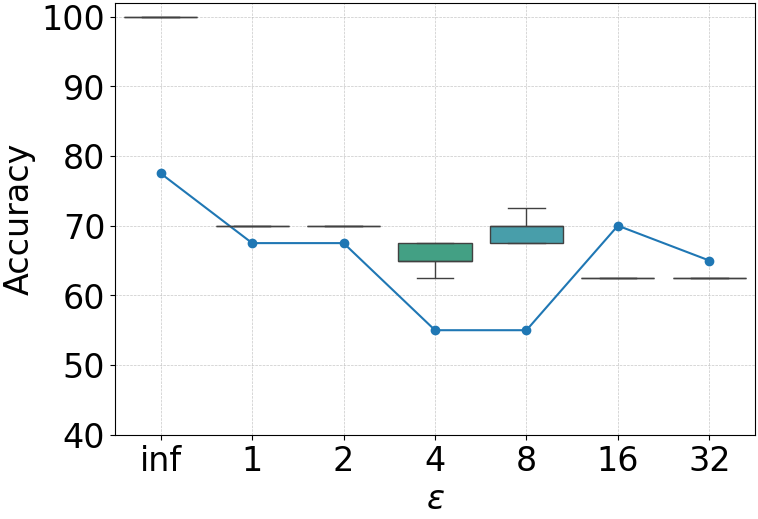}
		\label{fig:acc5_chains2_hops9}
		}
        \subfloat[$K\hspace{-1mm}=\hspace{-1mm}10$ (\chainM)]{
		\includegraphics[width = 2.6cm]{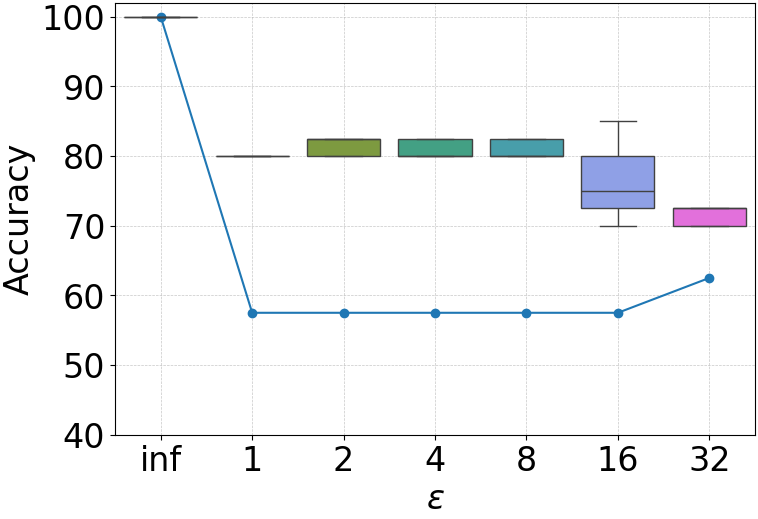}
		\label{fig:acc5_chains2_hops10}
		}
        \subfloat[$K\hspace{-1mm}=\hspace{-1mm}15$ (\chainM)]{
		\includegraphics[width = 2.6cm]{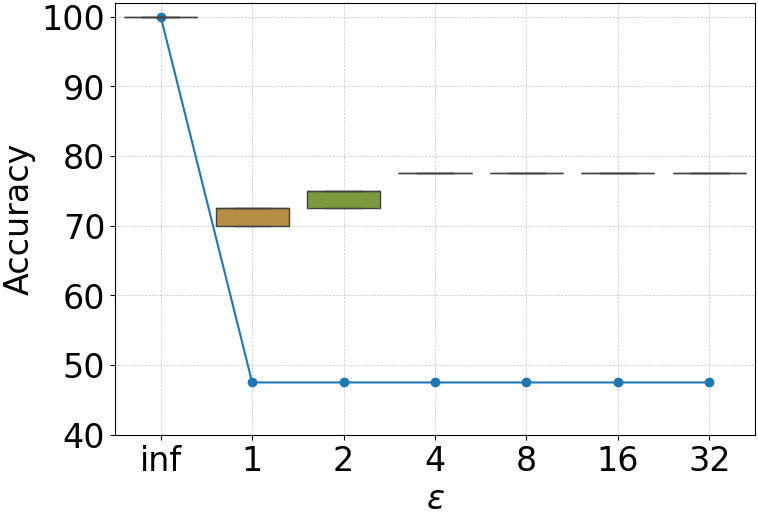}
		\label{fig:acc5_chains2_hops15}
		} 

		\subfloat[$K\hspace{-1mm}=\hspace{-1mm}7$ (\chainL)]{
		\includegraphics[width = 2.6cm]{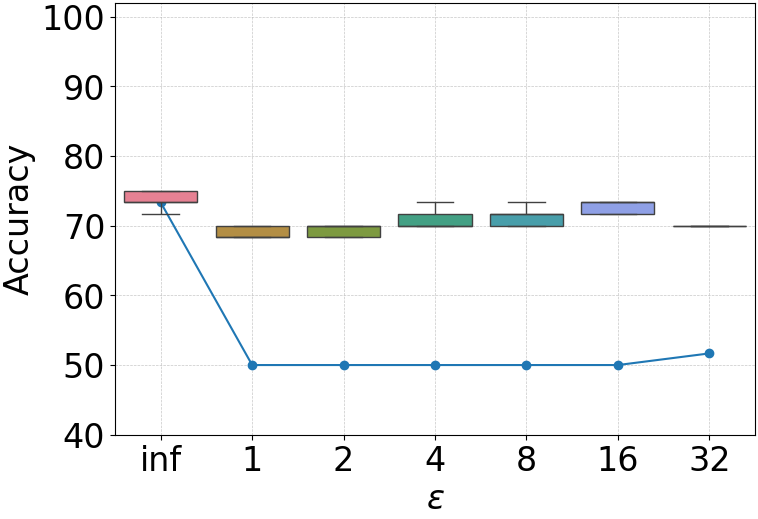}
		\label{fig:acc5_chains3_hops7}
		}
		\subfloat[$K\hspace{-1mm}=\hspace{-1mm}9$ (\chainL)]{
		\includegraphics[width = 2.6cm]{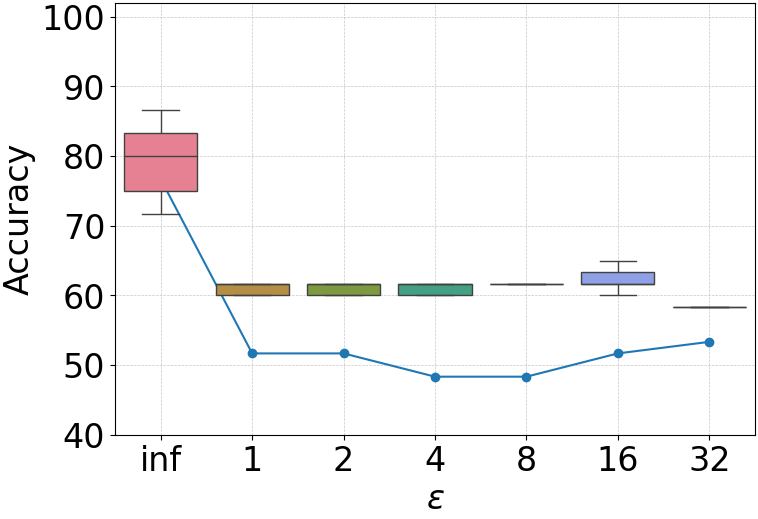}
		\label{fig:acc5_chains3_hops9}
		}
		\subfloat[$K\hspace{-1mm}=\hspace{-1mm}11$ (\chainL)]{
		\includegraphics[width = 2.6cm]{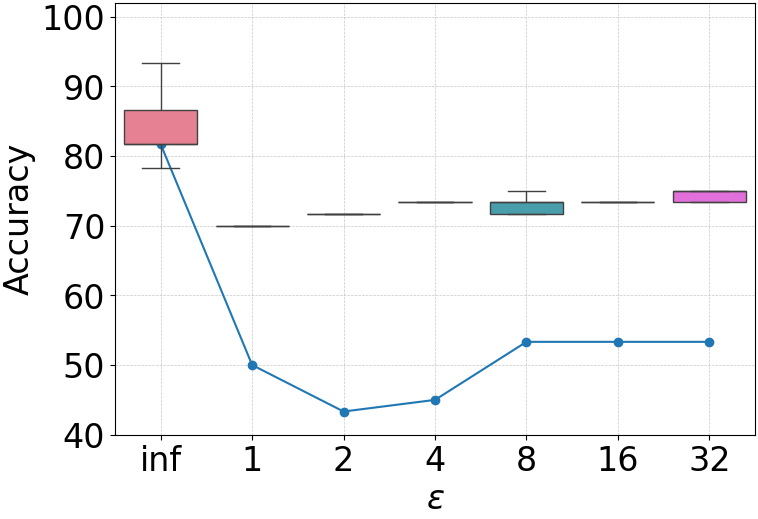}
		\label{fig:acc5_chains3_hops11}
		}
		\subfloat[$K\hspace{-1mm}=\hspace{-1mm}13$ (\chainL)]{
		\includegraphics[width = 2.6cm]{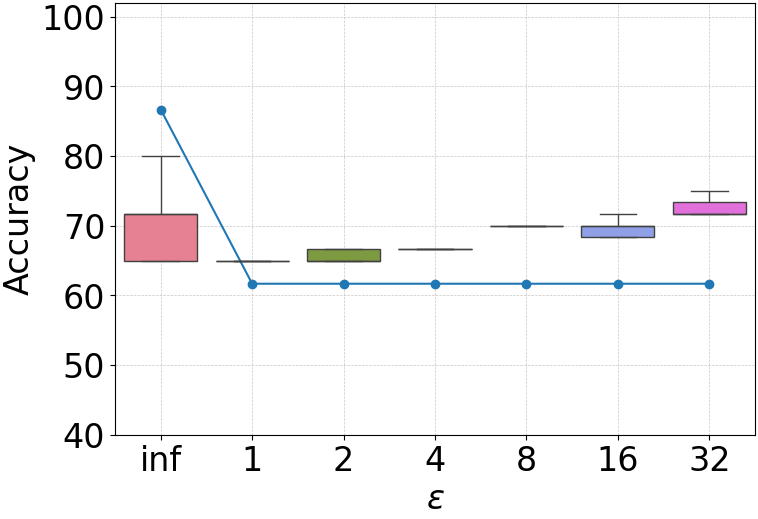}
		\label{fig:acc5_chains3_hops13}
		}
		\subfloat[$K\hspace{-1mm}=\hspace{-1mm}20$ (\chainL)]{
		\includegraphics[width = 2.6cm]{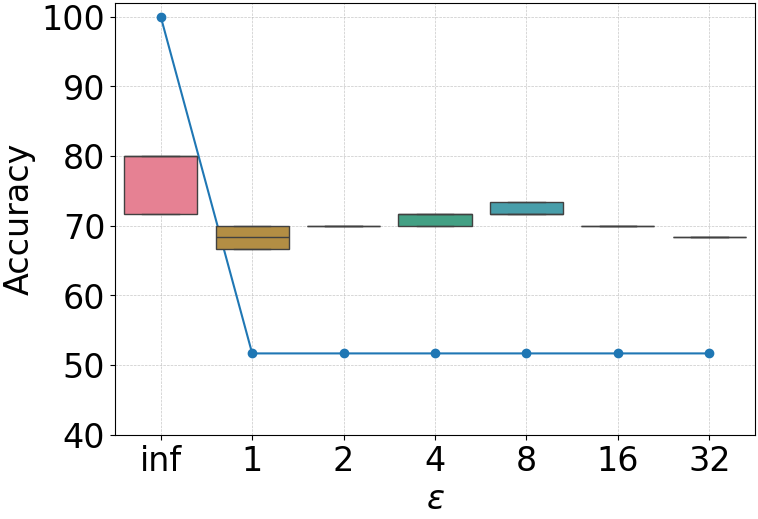}
		\label{fig:acc5_chains3_hops20}
		}
\end{center}
		\caption{Ablation Study of $\epsilon$ on $\gcn$ (colored boxes) and \gap\ (blue lines). For each pair of $K$ and dataset (\chainS, \chainM\  and \chainL), different $\epsilon$ are used. }
        \label{fig:K-study}
		\end{figure*}

\begin{table*}[t]
\small
    \caption{Top Accuracy over 3 Runs for NDP.  Maximum node degree is 10 for relatively large datasets (\pubmed, \facebook) and 5 for other datasets.  The \colorbox{colorbest}{\color{colortext}{best accuracy}} and the \colorbox{colorsecond}{\color{colortext}second-best} accuracy are highlighted, respectively. }\label{tab:reulst_table_overall_acc_top_node_5}
    \vspace{1pt}
    \centering
    \setlength\tabcolsep{12pt}
    \adjustbox{max width=\textwidth}{
    \begin{tabular}{l|c|c|c|c|c|c|c|c|c|c}
    \toprule[1pt]     & \textbf{Dataset}    &  \textbf{Computers} &  \textbf{Facebook}&    \textbf{PubMed} &  \textbf{Cora}&  \textbf{Photo}& \textbf{Chain-S}  & \textbf{Chain-M}& \textbf{Chain-L} &  \textbf{Chain-X}\\
     \cmidrule[1pt]{1-11}
 \multirow{4}{*}{$\epsilon=1$}     
    & $\gcn$ & \cellcolor{colorbest}\color{colortext}$91.91\%$ & \cellcolor{colorbest}\color{colortext}$60.36\%$ & \cellcolor{colorbest}\color{colortext}$73.85\%$ & \cellcolor{colorbest}\color{colortext}$80.63\%$ & \cellcolor{colorbest}\color{colortext}$95.10\%$ & \cellcolor{colorbest}\color{colortext}$78.12\%$ & \cellcolor{colorbest}\color{colortext}$70.0\%$ & \cellcolor{colorbest}\color{colortext}$61.67\%$ & \cellcolor{colorbest}\color{colortext}$61.0\%$ \\
    & \dpdgc & \cellcolor{colorsecond}\color{colortext}$57.65\%$ & \cellcolor{colorsecond}\color{colortext}$40.93\%$ & \cellcolor{colorsecond}\color{colortext}$58.91\%$ & $33.21\%$ & \cellcolor{colorsecond}\color{colortext}$42.21\%$ & $58.06\%$ & $57.5\%$ & $57.63\%$ & $54.6\%$ \\
    & \gap & $36.71\%$ & $35.66\%$ & $53.06\%$ & \cellcolor{colorsecond}\color{colortext}$34.13\%$ & $32.07\%$ & \cellcolor{colorsecond}\color{colortext}$65.62\%$ & $55.0\%$ & \cellcolor{colorsecond}\color{colortext}$58.33\%$ & \cellcolor{colorsecond}\color{colortext}$59.0\%$ \\
    & \sage & $34.41\%$ & $20.49\%$ & $39.89\%$ & $21.03\%$ & $22.8\%$ & $59.38\%$ & \cellcolor{colorsecond}\color{colortext}$60.0\%$ & $56.67\%$ & $55.0\%$ \\
    \cmidrule[1pt]{1-11}
  \multirow{4}{*}{$\epsilon=2$}     
    & $\gcn$ & \cellcolor{colorbest}\color{colortext}$92.28\%$ & \cellcolor{colorbest}\color{colortext}$64.95\%$ & \cellcolor{colorbest}\color{colortext}$80.4\%$ & \cellcolor{colorbest}\color{colortext}$83.76\%$ & \cellcolor{colorbest}\color{colortext}$94.76\%$ & \cellcolor{colorbest}\color{colortext}$78.12\%$ & \cellcolor{colorbest}\color{colortext}$72.5\%$ & \cellcolor{colorbest}\color{colortext}$63.33\%$ & \cellcolor{colorbest}\color{colortext}$61.0\%$ \\
    & \dpdgc & \cellcolor{colorsecond}\color{colortext}$66.11\%$ & \cellcolor{colorsecond}\color{colortext}$46.95\%$ & \cellcolor{colorsecond}\color{colortext}$70.94\%$ & $31.37\%$ & \cellcolor{colorsecond}\color{colortext}$53.41\%$ & $58.06\%$ & \cellcolor{colorsecond}\color{colortext}$57.5\%$ & $57.63\%$ & $54.5\%$ \\
    & \gap & $47.48\%$ & $39.80\%$ & $68.93\%$ & \cellcolor{colorsecond}\color{colortext}$32.47\%$ & $37.31\%$ & \cellcolor{colorsecond}\color{colortext}$65.62\%$ & $55.0\%$ & \cellcolor{colorsecond}\color{colortext}$58.33\%$ & \cellcolor{colorsecond}\color{colortext}$59.0\%$ \\
    & \sage & $35.63\%$ & $20.33\%$ & $40.02\%$ & $24.35\%$ & $23.92\%$ & $59.38\%$ & $60.0\%$ & $56.67\%$ & $55.0\%$ \\
    \cmidrule[1pt]{1-11}
  \multirow{4}{*}{$\epsilon=4$}     
    & $\gcn$ & \cellcolor{colorbest}\color{colortext}$92.28\%$ & \cellcolor{colorbest}\color{colortext}$68.19\%$ & \cellcolor{colorbest}\color{colortext}$84.81\%$ & \cellcolor{colorbest}\color{colortext}$85.98\%$ & \cellcolor{colorbest}\color{colortext}$94.9\%$ & \cellcolor{colorbest}\color{colortext}$78.12\%$ & \cellcolor{colorbest}\color{colortext}$72.5\%$ & \cellcolor{colorbest}\color{colortext}$63.33\%$ & \cellcolor{colorbest}\color{colortext}$62.0\%$ \\
    & \dpdgc & \cellcolor{colorsecond}\color{colortext}$72.20\%$ & \cellcolor{colorsecond}\color{colortext}$48.77\%$ & \cellcolor{colorsecond}\color{colortext}$80.02\%$ & $32.66\%$ & \cellcolor{colorsecond}\color{colortext}$69.91\%$ & $58.06\%$ & $57.5\%$ & $57.63\%$ & \cellcolor{colorsecond}\color{colortext}$59.6\%$ \\
    & \gap & $61.77\%$ & $46.72\%$ & $79.51\%$ & \cellcolor{colorsecond}\color{colortext}$33.58\%$ & $45.79\%$ & \cellcolor{colorsecond}\color{colortext}$65.62\%$ & $55.0\%$ & \cellcolor{colorsecond}\color{colortext}$58.33\%$ & $59.0\%$ \\
    & \sage & $35.56\%$ & $21.09\%$ & $40.78\%$ & $26.57\%$ & $25.38\%$ & $59.38\%$ & \cellcolor{colorsecond}\color{colortext}$60.0\%$ & $56.67\%$ & $55.0\%$ \\
    \cmidrule[1pt]{1-11}
  \multirow{4}{*}{$\epsilon=8$}     
    & $\gcn$ & \cellcolor{colorbest}\color{colortext}$92.24\%$ & \cellcolor{colorbest}\color{colortext}$70.11\%$ & \cellcolor{colorbest}\color{colortext}$87.34\%$ & \cellcolor{colorbest}\color{colortext}$87.45\%$ & \cellcolor{colorbest}\color{colortext}$94.96\%$ & \cellcolor{colorbest}\color{colortext}$75.00\%$ & \cellcolor{colorbest}\color{colortext}$72.5\%$ & \cellcolor{colorbest}\color{colortext}$66.67\%$ & \cellcolor{colorbest}\color{colortext}$63.0\%$ \\
    & \dpdgc & \cellcolor{colorsecond}\color{colortext}$76.39\%$ & \cellcolor{colorsecond}\color{colortext}$49.91\%$ & \cellcolor{colorsecond}\color{colortext}$83.49\%$ & \cellcolor{colorsecond}\color{colortext}$43.91\%$ & \cellcolor{colorsecond}\color{colortext}$77.67\%$ & $58.06\%$ & $57.5\%$ & $57.63\%$ & $58.6\%$ \\
    & \gap & $68.52\%$ & $48.24\%$ & $82.17\%$ & $31.73\%$ & $68.39\%$ & \cellcolor{colorsecond}\color{colortext}$65.62\%$ & $55.0\%$ & \cellcolor{colorsecond}\color{colortext}$58.33\%$ & \cellcolor{colorsecond}\color{colortext}$59.0\%$ \\
    & \sage & $34.41\%$ & $22.25\%$ & $42.28\%$ & $28.97\%$ & $27.17\%$ & $59.38\%$ & \cellcolor{colorsecond}\color{colortext}$60.0\%$ & $56.67\%$ & $55.0\%$ \\
    \cmidrule[1pt]{1-11}
  \multirow{4}{*}{$\epsilon=16$}     
    & $\gcn$ & \cellcolor{colorbest}\color{colortext}$92.50\%$ & \cellcolor{colorbest}\color{colortext}$70.60\%$ & \cellcolor{colorbest}\color{colortext}$88.49\%$ & \cellcolor{colorbest}\color{colortext}$87.82\%$ & \cellcolor{colorbest}\color{colortext}$94.96\%$ & \cellcolor{colorbest}\color{colortext}$71.88\%$ & \cellcolor{colorbest}\color{colortext}$72.5\%$ & \cellcolor{colorbest}\color{colortext}$73.33\%$ & \cellcolor{colorbest}\color{colortext}$65.0\%$ \\
    & \dpdgc & \cellcolor{colorsecond}\color{colortext}$78.88\%$ & \cellcolor{colorsecond}\color{colortext}$50.11\%$ & \cellcolor{colorsecond}\color{colortext}$84.88\%$ & \cellcolor{colorsecond}\color{colortext}$56.46\%$ & \cellcolor{colorsecond}\color{colortext}$83.10\%$ & $58.06\%$ & $57.5\%$ & $57.63\%$ & $54.5\%$ \\
    & \gap & $73.83\%$ & $49.91\%$ & $83.77\%$ & $37.27\%$ & $76.61\%$ & \cellcolor{colorsecond}\color{colortext}$65.62\%$ & $55.0\%$ & \cellcolor{colorsecond}\color{colortext}$58.33\%$ & \cellcolor{colorsecond}\color{colortext}$59.0\%$ \\
    & \sage & $33.18\%$ & $22.15\%$ & $44.00\%$ & $30.44\%$ & $29.62\%$ & $59.38\%$ & \cellcolor{colorsecond}\color{colortext}$60.0\%$ & $56.67\%$ & $55.0\%$ \\
    \cmidrule[1pt]{1-11}
  \multirow{4}{*}{$\epsilon=32$}     
    & $\gcn$ & \cellcolor{colorbest}\color{colortext}$92.46\%$ & \cellcolor{colorbest}\color{colortext}$70.96\%$ & \cellcolor{colorbest}\color{colortext}$88.97\%$ & \cellcolor{colorbest}\color{colortext}$87.82\%$ & \cellcolor{colorbest}\color{colortext}$94.90\%$ & \cellcolor{colorbest}\color{colortext}$68.75\%$ & \cellcolor{colorbest}\color{colortext}$72.5\%$ & \cellcolor{colorbest}\color{colortext}$75.00\%$ & \cellcolor{colorbest}\color{colortext}$64.0\%$ \\
    & \dpdgc & \cellcolor{colorsecond}\color{colortext}$81.40\%$ & $50.64\%$ & \cellcolor{colorsecond}\color{colortext}$86.33\%$ & \cellcolor{colorsecond}\color{colortext}$64.21\%$ & \cellcolor{colorsecond}\color{colortext}$86.08\%$ & $58.06\%$ & $57.5\%$ & \cellcolor{colorsecond}\color{colortext}$57.63\%$ & $54.5\%$ \\
    & \gap & $77.13\%$ & \cellcolor{colorsecond}\color{colortext}$50.66\%$ & $85.21\%$ & $57.38\%$ & $80.38\%$ & \cellcolor{colorsecond}\color{colortext}$62.50\%$ & $55.0\%$ & $56.67\%$ & \cellcolor{colorsecond}\color{colortext}$59.0\%$ \\
    & \sage & $36.45\%$ & $22.67\%$ & $46.97\%$ & $31.73\%$ & $34.13\%$ & $59.38\%$ & \cellcolor{colorsecond}\color{colortext}$60.0\%$ & $56.67\%$ & $55.0\%$ \\
    \bottomrule
    \end{tabular}
    }
\end{table*}

\begin{table*}[t]
\small
    \caption{Mean Accuracy over 3 Runs for EDP and NDP. The \colorbox{colorbest}{\color{colortext}{best accuracy}} and the \colorbox{colorsecond}{\color{colortext}second-best} accuracy are highlighted, respectively.  The symbol \greenup\ represents that the best accuracy improves the second-best accuracy by more than $10\%$. The symbol \reddown\  represents the accuracy less than $55\%$, close to random guess on the chain-structured datasets.}\label{tab:reulst_table_overall_acc_mean}
    \vspace{1pt}
    \centering
    \setlength\tabcolsep{12pt}
    \adjustbox{max width=\textwidth}{
    \begin{tabular}{l|c|c|c|c|c|c|c|c|c|c}
    \toprule[1pt]     & \textbf{Dataset}   &  \textbf{Computers} &  \textbf{Facebook}&    \textbf{PubMed} &  \textbf{Cora}&  \textbf{Photo}& \textbf{Chain-S}  & \textbf{Chain-M}& \textbf{Chain-L} &  \textbf{Chain-X} \\
    \cmidrule[1pt]{1-11}
     \multicolumn{11}{c}{\textbf{EDP}} \\\cmidrule[1pt]{1-11}
  \multirow{4}{*}{$\epsilon=1$}   &   $\gcn$ &  \cellcolor{colorbest}\color{colortext}$92.0\%$&  \cellcolor{colorbest}\color{colortext}$74.0\%$&  \cellcolor{colorsecond}\color{colortext}$87.9\%$&  \cellcolor{colorbest}\color{colortext}$84.3\%$&  \cellcolor{colorbest}\color{colortext}$95.6\%$&  \cellcolor{colorbest}\color{colortext}$78.1\%$ &  \cellcolor{colorbest}\color{colortext}$80.0\%$ &  \cellcolor{colorbest}\color{colortext}$70.0\%$&  \cellcolor{colorbest}\color{colortext}$65.0\%$\\
    &\dpdgc    & \cellcolor{colorsecond}\color{colortext}$88.0\%$ & $60.6\%$ &\cellcolor{colorbest}\color{colortext}$88.3\%$&$75.5\%$&$92.5\%$& $43.8\%$    & $50.0\%$ 
     &$51.7\%$ & $39.0\%$ \\
     &\gap  &  $87.0\%$&  \cellcolor{colorsecond}\color{colortext}$68.3\%$&  $87.2\%$&  \cellcolor{colorsecond}\color{colortext}$76.0\%$&  \cellcolor{colorsecond}\color{colortext}$92.8\%$&  \cellcolor{colorsecond}\color{colortext}$65.6\%$&  \cellcolor{colorsecond}\color{colortext}$67.5\%$&  \cellcolor{colorsecond}\color{colortext}$61.7\%$&  \cellcolor{colorsecond}\color{colortext}$55.5\%$\\
   & \sage    &$77.8\%$ & $48.2\%$ &$85.0\%$&$60.0\%$&$82.4\%$  &  $53.1\%$  & $45.0\%$ 
     &$53.3\%$  &$51.0\%$ \\
     \cmidrule[1pt]{1-11}
      \multirow{4}{*}{$\epsilon=2$}   &   $\gcn$  &  \cellcolor{colorbest}\color{colortext}$92.0\%$&  \cellcolor{colorbest}\color{colortext}$73.8\%$&  \cellcolor{colorbest}\color{colortext}$89.1\%$&  \cellcolor{colorbest}\color{colortext}$86.0\%$&  \cellcolor{colorbest}\color{colortext}$95.7\%$&  \cellcolor{colorbest}\color{colortext}$90.6\%$ &  \cellcolor{colorbest}\color{colortext}$82.5\%$ &  \cellcolor{colorbest}\color{colortext}$71.7\%$ &  \cellcolor{colorbest}\color{colortext}$68.0\%$ \\
    &\dpdgc    & \cellcolor{colorsecond}\color{colortext}$88.2\%$ & $66.7\%$&\cellcolor{colorsecond}\color{colortext}$88.2\%$& \cellcolor{colorsecond}\color{colortext}$77.5\% $&\cellcolor{colorsecond}\color{colortext}$93.3\%$ & $43.8\%$   & $50.0\%$ 
     &$51.7\%$  &$39.0\%$ \\
     &\gap  &  $88.0\%$&  \cellcolor{colorsecond}\color{colortext}$71.7\%$&  $87.3\%$&  $76.8\%$&  \cellcolor{colorsecond}\color{colortext}$93.3\%$ &  \cellcolor{colorsecond}\color{colortext}$65.6\%$&  \cellcolor{colorsecond}\color{colortext}$67.5\%$&  \cellcolor{colorsecond}\color{colortext}$61.7\%$&  \cellcolor{colorsecond}\color{colortext}$55.0\%$ \\
   & \sage   &$76.1\%$ & $48.1\%$  &$84.6\%$&$60.1\%$&$82.4\%$ &  $43.8\%$   & $45.0\%$ 
     &$53.3\%$ &$51.0\%$   \\
     \cmidrule[1pt]{1-11} 
       \multirow{4}{*}{$\epsilon=4$}   &   $\gcn$  &  \cellcolor{colorbest}\color{colortext}$92.2\%$&  \cellcolor{colorbest}\color{colortext}$73.9\%$&  \cellcolor{colorbest}\color{colortext}$89.5\%$&  \cellcolor{colorbest}\color{colortext}$86.7\%$&  \cellcolor{colorbest}\color{colortext}$95.8\%$&  \cellcolor{colorbest}\color{colortext}$90.6\%$ &  \cellcolor{colorbest}\color{colortext}$82.5\%$ &  \cellcolor{colorbest}\color{colortext}$71.7\%$ &  \cellcolor{colorbest}\color{colortext}$65.0\%$\\
    &\dpdgc   & \cellcolor{colorsecond}\color{colortext}$88.9\%$ & $73.3\%$ &\cellcolor{colorsecond}\color{colortext}$88.4\%$&$75.1\%$& \cellcolor{colorsecond}\color{colortext}$94.2\%$ & $43.8\%$   & $50.0\%$ 
     & $51.7\%$ &$39.0\%$ \\
     &\gap    &  $88.8\%$&  \cellcolor{colorsecond}\color{colortext}$73.6\%$&  $87.7\%$&  \cellcolor{colorsecond}\color{colortext}$76.9\%$&  $93.8\%$&  \cellcolor{colorsecond}\color{colortext}$62.5\%$&  \cellcolor{colorsecond}\color{colortext}$60.0\%$&  \cellcolor{colorsecond}\color{colortext}$61.7\%$&  \cellcolor{colorsecond}\color{colortext}$55.5\%$\\
   & \sage    & $79.1\%$ & $50.3\%$ &$85.8\%$& $63.3\%$& $85.7\%$ &  $50.0\%$   & $47.5\%$ 
     &$51.7\%$  &$54.0\%$ \\
     \cmidrule[1pt]{1-11}
       \multirow{4}{*}{$\epsilon=8$}   &   $\gcn$ &  \cellcolor{colorbest}\color{colortext}$92.2\%$&  $74.2\%$&  \cellcolor{colorbest}\color{colortext}$89.7\%$&  \cellcolor{colorbest}\color{colortext}$87.6\%$&  \cellcolor{colorbest}\color{colortext}$95.8\%$ &  \cellcolor{colorbest}\color{colortext}$81.2\%$ &  \cellcolor{colorbest}\color{colortext}$82.5\%$ &  \cellcolor{colorbest}\color{colortext}$73.3\%$ &  \cellcolor{colorbest}\color{colortext}$68.0\%$\\
    &\dpdgc    & $89.4\%$ & \cellcolor{colorbest}\color{colortext}$78.6\%$ &\cellcolor{colorsecond}\color{colortext}$88.6\%$&$76.0\%$&\cellcolor{colorsecond}\color{colortext}$94.6\%$&  $43.8\%$  & $50.0\%$ 
     &$51.7\%$  &$39.0\%$ \\
     &\gap   &  \cellcolor{colorsecond}\color{colortext}$89.6\%$&  $75.0\%$&  $88.0\%$&  \cellcolor{colorsecond}\color{colortext}$78.2\%$&  \cellcolor{colorsecond}\color{colortext}$94.6\%$ &  \cellcolor{colorsecond}\color{colortext}$59.4\%$&  \cellcolor{colorsecond}\color{colortext}$60.0\%$&  \cellcolor{colorsecond}\color{colortext}$61.7\%$&  $55.5\%$\\
   & \sage    & $87.9\%$ & \cellcolor{colorsecond}\color{colortext}$75.6\%$&$84.8\%$&$ 75.7\%$& $92.2\%$&  $43.8\%$   & $47.5\%$ 
     &$51.7\%$  &\cellcolor{colorsecond}\color{colortext}$61.0\%$ \\
     \cmidrule[1pt]{1-11} 
     \multirow{4}{*}{$\epsilon=16$}   &   $\gcn$  &  \cellcolor{colorbest}\color{colortext}$92.1\%$&  $74.3\%$&  \cellcolor{colorbest}\color{colortext}$89.8\%$&  \cellcolor{colorbest}\color{colortext}$88.0\%$&  \cellcolor{colorbest}\color{colortext}$95.9\%$&  \cellcolor{colorbest}\color{colortext}$87.5\%$ &  \cellcolor{colorbest}\color{colortext}$85.0\%$ &  \cellcolor{colorbest}\color{colortext}$70.0\%$&  \cellcolor{colorbest}\color{colortext}$68.0\%$ \\
    &\dpdgc   & $90.4\%$ &  \cellcolor{colorbest}\color{colortext}$81.2\%$&\cellcolor{colorsecond}\color{colortext}$88.9\%$&$77.7\%$&$93.8\%$ & $43.8\%$   & $50.0\%$ 
     &$51.7\%$  &$40.0\%$ \\
     &\gap   &  $90.0\%$&  $76.0\%$&  $88.5\%$&  $80.3\%$&  \cellcolor{colorsecond}\color{colortext}$94.6\%$&  \cellcolor{colorsecond}\color{colortext}$59.4\%$&  \cellcolor{colorsecond}\color{colortext}$70.0\%$&  \cellcolor{colorsecond}\color{colortext}$61.7\%$&  \cellcolor{colorsecond}\color{colortext}$54.3\%$ \\
   & \sage    & \cellcolor{colorsecond}\color{colortext}$90.9\%$ &\cellcolor{colorsecond}\color{colortext}$79.5\%$&$87.6\%$&\cellcolor{colorsecond}\color{colortext}$84.7\%$& $94.1\%$&  $43.8\%$   & $47.5\%$ 
     & $51.7\%$ & $51.0\%$  \\
     \cmidrule[1pt]{1-11} 
      \multirow{4}{*}{$\epsilon=32$}   &   $\gcn$ &  \cellcolor{colorbest}\color{colortext}$92.2\%$&  $73.9\%$&  \cellcolor{colorbest}\color{colortext}$89.8\%$&  \cellcolor{colorbest}\color{colortext}$88.2\%$&  \cellcolor{colorbest}\color{colortext}$95.8\%$&  \cellcolor{colorbest}\color{colortext}$90.6\%$ &  \cellcolor{colorbest}\color{colortext}$82.5\%$ &  \cellcolor{colorbest}\color{colortext}$78.3\%$ &  \cellcolor{colorbest}\color{colortext}$67.0\%$ \\
    &\dpdgc   & \cellcolor{colorsecond}\color{colortext}$91.3\%$ & \cellcolor{colorbest}\color{colortext}$82.9\% $ &\cellcolor{colorsecond}\color{colortext}$88.8\%$&$79.9\%$&$94.3\%$ & $43.8\%$   & $50.0\%$ 
     & $51.7\%$ &$40.0\%$ \\
     &\gap   &  $90.2\%$&  $76.5\% $&  $88.7\%$&  $82.2\%$&  \cellcolor{colorsecond}\color{colortext}$94.6\%$ &  \cellcolor{colorsecond}\color{colortext}$59.4\% $&  \cellcolor{colorsecond}\color{colortext}$65.0\% $&  \cellcolor{colorsecond}\color{colortext}$61.7\% $&  \cellcolor{colorsecond}\color{colortext}$54.3\%$ \\
   & \sage   & $90.6\%$ & \cellcolor{colorsecond}\color{colortext}$79.9\% $ &$86.9\%$&\cellcolor{colorsecond}\color{colortext}$85.2\%$&$94.4\%$ &  $43.8\%$  & $47.5\%$ 
     & $51.7\%$ &$51.0\%$ \\
     \cmidrule[1pt]{1-11} 
      \multicolumn{11}{c}{\textbf{NDP (max node degree $=20$)}} \\\cmidrule[1pt]{1-11}
  \multirow{4}{*}{$\epsilon=1$}     
    & $\gcn$ & \cellcolor{colorbest}\color{colortext}$83.10\%$ & \cellcolor{colorbest}\color{colortext}$50.32\%$ & \cellcolor{colorbest}\color{colortext}$62.05\%$ & \cellcolor{colorbest}\color{colortext}$67.17\%$ & \cellcolor{colorbest}\color{colortext}$87.75\%$ & \cellcolor{colorsecond}\color{colortext}$53.95\%$ & \cellcolor{colorsecond}\color{colortext}$53.67\%$ & $51.80\%$ & \cellcolor{colorsecond}\color{colortext}$51.68\%$ \\
    & \dpdgc & \cellcolor{colorsecond}\color{colortext}$54.42\%$ & \cellcolor{colorsecond}\color{colortext}$36.73\%$ & $51.34\%$ & $30.04\%$ & \cellcolor{colorsecond}\color{colortext}$42.60\%$ & $48.06\%$ & $47.25\%$ & \cellcolor{colorsecond}\color{colortext}$52.20\%$ & $48.99\%$ \\
    & \gap & $36.71\%$ & $35.07\%$ & \cellcolor{colorsecond}\color{colortext}$55.06\%$ & \cellcolor{colorsecond}\color{colortext}$33.95\%$ & $30.88\%$ & \cellcolor{colorbest}\color{colortext}$65.62\%$ & \cellcolor{colorbest}\color{colortext}$55.00\%$ & \cellcolor{colorbest}\color{colortext}$58.33\%$ & \cellcolor{colorbest}\color{colortext}$59.00\%$ \\
    & \sage & $23.91\%$ & $18.58\%$ & $37.32\%$ & $13.62\%$ & $17.10\%$ & $51.88\%$ & $51.25\%$ & $49.83\%$ & $50.60\%$ \\
    \cmidrule[1pt]{1-11}
  \multirow{4}{*}{$\epsilon=2$}     
    & $\gcn$ & \cellcolor{colorbest}\color{colortext}$84.60\%$ & \cellcolor{colorbest}\color{colortext}$53.03\%$ & \cellcolor{colorsecond}\color{colortext}$64.22\%$ & \cellcolor{colorbest}\color{colortext}$69.63\%$ & \cellcolor{colorbest}\color{colortext}$88.89\%$ & \cellcolor{colorsecond}\color{colortext}$54.49\%$ & \cellcolor{colorsecond}\color{colortext}$54.32\%$ & \cellcolor{colorsecond}\color{colortext}$52.23\%$ & \cellcolor{colorsecond}\color{colortext}$52.14\%$ \\
    & \dpdgc & $64.05\%$ & $42.68\%$ & \cellcolor{colorbest}\color{colortext}$65.20\%$ & $30.92\%$ & $53.76\%$ & $48.06\%$ & $47.25\%$ & $52.20\%$ & $48.99\%$ \\
    & \gap & $42.32\%$ & $37.52\%$ & $61.44\%$ & $33.67\%$ & $33.46\%$ & \cellcolor{colorbest}\color{colortext}$65.62\%$ & \cellcolor{colorbest}\color{colortext}$55.00\%$ & \cellcolor{colorbest}\color{colortext}$58.33\%$ & \cellcolor{colorbest}\color{colortext}$59.00\%$ \\
    & \sage & $31.32\%$ & $20.39\%$ & $38.08\%$ & $14.6\%$ & $19.90\%$ & $51.88\%$ & $51.25\%$ & $49.83\%$ & $50.60\%$\\
    \cmidrule[1pt]{1-11}
  \multirow{4}{*}{$\epsilon=4$}     
    & $\gcn$ & \cellcolor{colorbest}\color{colortext}$85.95\%$ & \cellcolor{colorbest}\color{colortext}$55.73\%$ & \cellcolor{colorsecond}\color{colortext}$67.15\%$ & \cellcolor{colorbest}\color{colortext}$71.97\%$ & \cellcolor{colorbest}\color{colortext}$89.89\%$ & \cellcolor{colorsecond}\color{colortext}$54.71\%$ & \cellcolor{colorbest}\color{colortext}$55.35\%$ & \cellcolor{colorsecond}\color{colortext}$52.83\%$ & \cellcolor{colorsecond}\color{colortext}$52.89\%$ \\
    & \dpdgc & \cellcolor{colorsecond}\color{colortext}$70.79\%$ & \cellcolor{colorsecond}\color{colortext}$47.52\%$ & \cellcolor{colorbest}\color{colortext}$78.82\%$ & $31.22\%$ & \cellcolor{colorsecond}\color{colortext}$71.17\%$ & $48.06\%$ & $47.25\%$ & $52.20\%$ & $48.99\%$ \\
    & \gap & $48.82\%$ & $40.69\%$ & $67.40\%$ & \cellcolor{colorsecond}\color{colortext}$33.58\%$ & $37.42\%$ & \cellcolor{colorbest}\color{colortext}$65.62\%$ & \cellcolor{colorsecond}\color{colortext}$55.00\%$ & \cellcolor{colorbest}\color{colortext}$58.33\%$ & \cellcolor{colorbest}\color{colortext}$59.00\%$ \\
    & \sage & $35.14\%$ & $22.50\%$ & $39.03\%$ & $16.40\%$ & $23.48\%$ & $51.88\%$ & $51.25\%$ & $49.83\%$ & $50.60\%$ \\
    \cmidrule[1pt]{1-11}
  \multirow{4}{*}{$\epsilon=8$}     
    & $\gcn$ & \cellcolor{colorbest}\color{colortext}$87.15\%$ & \cellcolor{colorbest}\color{colortext}$58.05\%$ & $70.17\%$ & \cellcolor{colorbest}\color{colortext}$73.80\%$ & \cellcolor{colorbest}\color{colortext}$90.73\%$ & \cellcolor{colorsecond}\color{colortext}$55.71\%$ & \cellcolor{colorbest}\color{colortext}$56.09\%$ & \cellcolor{colorsecond}\color{colortext}$53.57\%$ & \cellcolor{colorsecond}\color{colortext}$53.50\%$ \\
    & \dpdgc & \cellcolor{colorsecond}\color{colortext}$75.32\%$ & \cellcolor{colorsecond}\color{colortext}$49.34\%$ & \cellcolor{colorbest}\color{colortext}$82.65\%$ & \cellcolor{colorsecond}\color{colortext}$37.47\%$ & \cellcolor{colorsecond}\color{colortext}$78.54\%$ & $48.06\%$ & $47.25\%$ & $52.20\%$ & $49.49\%$ \\
    & \gap & $53.75\%$ & $42.60\%$ & \cellcolor{colorsecond}\color{colortext}$71.14\%$ & $33.07\%$ & $45.18\%$ & \cellcolor{colorbest}\color{colortext}$65.62\%$ & \cellcolor{colorsecond}\color{colortext}$55.00\%$ & \cellcolor{colorbest}\color{colortext}$58.33\%$ & \cellcolor{colorbest}\color{colortext}$59.00\%$ \\
    & \sage & $36.82\%$ & $24.05\%$ & $40.90\%$ & $19.33\%$ & $27.18\%$ & $51.88\%$ & $51.25\%$ & $49.83\%$ & $50.60\%$ \\
    \cmidrule[1pt]{1-11}
  \multirow{4}{*}{$\epsilon=16$}     
    & $\gcn$ & \cellcolor{colorbest}\color{colortext}$88.28\%$ & \cellcolor{colorbest}\color{colortext}$60.06\%$ & $72.76\%$ & \cellcolor{colorbest}\color{colortext}$75.24\%$ & \cellcolor{colorbest}\color{colortext}$91.53\%$ & \cellcolor{colorsecond}\color{colortext}$55.99\%$ & \cellcolor{colorbest}\color{colortext}$56.99\%$ & \cellcolor{colorsecond}\color{colortext}$54.11\%$ & \cellcolor{colorsecond}\color{colortext}$54.28\%$ \\
    & \dpdgc & \cellcolor{colorsecond}\color{colortext}$77.85\%$ & \cellcolor{colorsecond}\color{colortext}$50.00\%$ & \cellcolor{colorbest}\color{colortext}$84.27\%$ & \cellcolor{colorsecond}\color{colortext}$53.71\%$ & \cellcolor{colorsecond}\color{colortext}$81.84\%$ & $48.06\%$ & $47.25\%$ & $52.20\%$ & $50.00\%$ \\
    & \gap & $57.79\%$ & $44.07\%$ & \cellcolor{colorsecond}\color{colortext}$73.65\%$ & $33.95\%$ & $51.49\%$ & \cellcolor{colorbest}\color{colortext}$65.62\%$ & \cellcolor{colorsecond}\color{colortext}$55.00\%$ & \cellcolor{colorbest}\color{colortext}$58.33\%$ & \cellcolor{colorbest}\color{colortext}$59.00\%$ \\
    & \sage & $39.57\%$ & $25.95\%$ & $43.45\%$ & $22.14\%$ & $31.64\%$ & $51.88\%$ & $51.25\%$ & $49.83\%$ & $50.60\%$ \\
    \cmidrule[1pt]{1-11}
  \multirow{4}{*}{$\epsilon=32$}     
    & $\gcn$ & \cellcolor{colorbest}\color{colortext}$89.35\%$ & \cellcolor{colorbest}\color{colortext}$61.87\%$ & $74.83\%$ & \cellcolor{colorbest}\color{colortext}$76.61\%$ & \cellcolor{colorbest}\color{colortext}$92.36\%$ & \cellcolor{colorsecond}\color{colortext}$56.23\%$ & \cellcolor{colorbest}\color{colortext}$57.72\%$ & \cellcolor{colorsecond}\color{colortext}$54.64\%$ & \cellcolor{colorsecond}\color{colortext}$54.77\%$ \\
    & \dpdgc & \cellcolor{colorsecond}\color{colortext}$80.53\%$ & \cellcolor{colorsecond}\color{colortext}$50.65\%$ & \cellcolor{colorbest}\color{colortext}$85.34\%$ & \cellcolor{colorsecond}\color{colortext}$61.18\%$ & \cellcolor{colorsecond}\color{colortext}$85.53\%$ & $48.06\%$ & $47.25\%$ & $52.20\%$ & $49.59\%$ \\
    & \gap & $60.98\%$ & $45.17\%$ & \cellcolor{colorsecond}\color{colortext}$75.54\%$ & $37.79\%$ & $56.28\%$ & \cellcolor{colorbest}\color{colortext}$65.62\%$ & \cellcolor{colorsecond}\color{colortext}$55.00\%$ & \cellcolor{colorbest}\color{colortext}$58.33\%$ & \cellcolor{colorbest}\color{colortext}$59.00\%$ \\
    & \sage & $45.63\%$ & $29.30\%$ & $46.38\%$ & $25.44\%$ & $39.40\%$ & $51.88\%$ & $51.25\%$ & $49.83\%$ & $50.60\%$ \\  
   \cmidrule[1pt]{1-11}
      \multicolumn{11}{c}{\textbf{Non-Private}} \\
         \cmidrule[1pt]{1-11} 
     \multirow{5}{*}{Plain}   &   $\gcn$  &  \cellcolor{colorsecond}\color{colortext}$92.2\%$&  $78.1\% $&  \cellcolor{colorbest}\color{colortext}$89.9\%$&  \cellcolor{colorbest}\color{colortext}$88.4\%$&  \cellcolor{colorsecond}\color{colortext}$95.8\%$&  \cellcolor{colorbest}\color{colortext}$100.0\% $&  \cellcolor{colorbest}\color{colortext}$100.0\% $&  \cellcolor{colorbest}\color{colortext}$100.0\% $&  \cellcolor{colorbest}\color{colortext}$100.0\%$\\
    &\dpdgc   & \cellcolor{colorbest}\color{colortext}$92.8\%$ & \cellcolor{colorbest}\color{colortext}$86.4\% $ &$88.1\%$&$83.9\%$& \cellcolor{colorbest}\color{colortext}$96.2\%$  & $59.4\% $  & $77.5\% $ 
     &$63.3\% $ &$73.0\% $\\
     &\gap    &  $90.9\%$&  $78.3\% $&  \cellcolor{colorsecond}\color{colortext}$89.1\%$&  \cellcolor{colorsecond}\color{colortext}$85.1\%$&  $95.2\%$&  \cellcolor{colorsecond}\color{colortext}$97.9\%$&  \cellcolor{colorbest}\color{colortext}$100.0\%$&  \cellcolor{colorsecond}\color{colortext}$93.3\%$&  \cellcolor{colorbest}\color{colortext}$100.0\%$\\
   & \sage    & $91.6\%$ & \cellcolor{colorsecond}\color{colortext}$79.7\%$&$87.0\%$&$82.1\%$&$94.2\%$ &$59.4\% $    & $55.0\%$ 
     & $60.0\% $&$58.0\% $ \\
   &  \mlpE   & $83.9\%$ &$51.1\%$  &$84.9\%$&$72.9\%$&$91.1\%$ &  $40.6\%$  & $47.5\%$ 
     & $46.7\%$ &$51.0\%$ \\
    \bottomrule
    \end{tabular}
    }
\end{table*}

\begin{figure}[!t]
		\subfloat[Chain-S]{
		\includegraphics[width = 4cm]{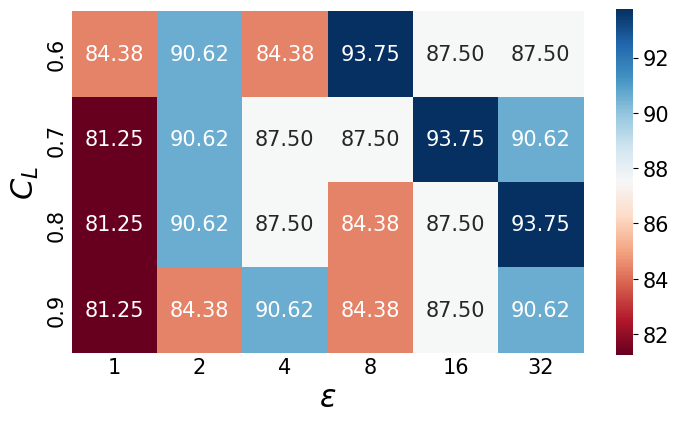}
		\label{fig:heatmap_chains1_epsilon_cl}
		}
      \subfloat[Cora]{
		\includegraphics[width = 4cm]{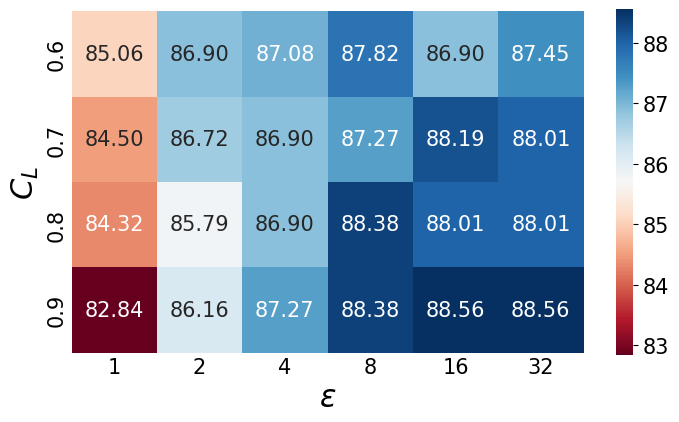}
		\label{fig:heatmap_cora_epsilon_cl}
		}

        \subfloat[Chain-S]{
		\includegraphics[width = 4cm]{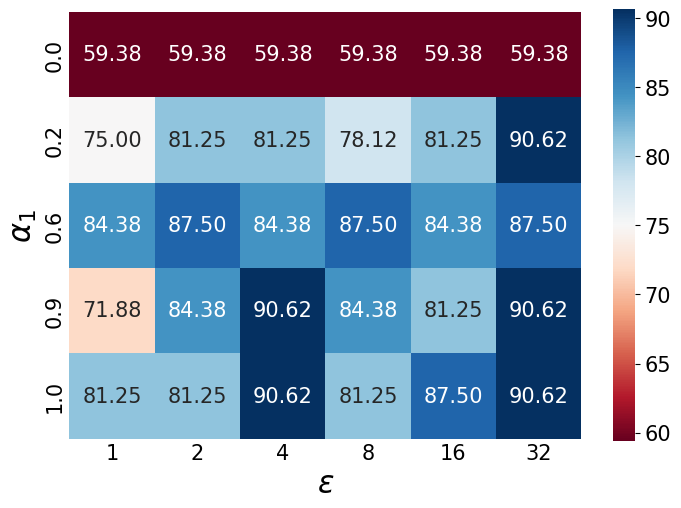}
		\label{fig:heatmap_chains1_epsilon_alpha}
		}
      \subfloat[Cora]{
		\includegraphics[width = 4cm]{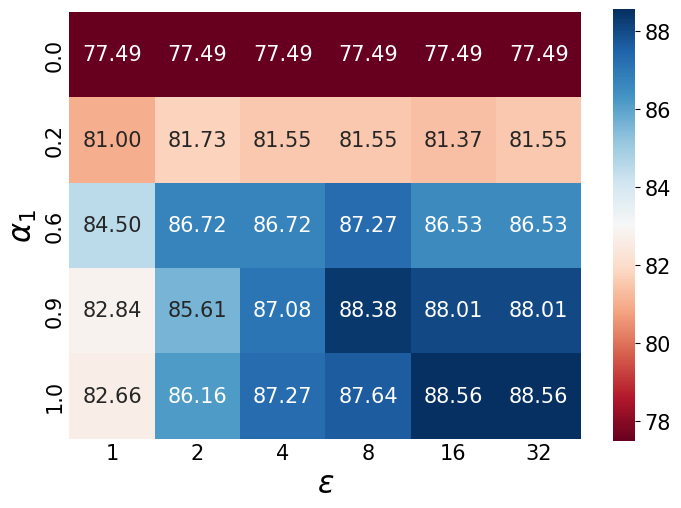}
		\label{fig:heatmap_cora_epsilon_alpha}
		}

        	\subfloat[Chain-S]{
		\includegraphics[width = 4cm]{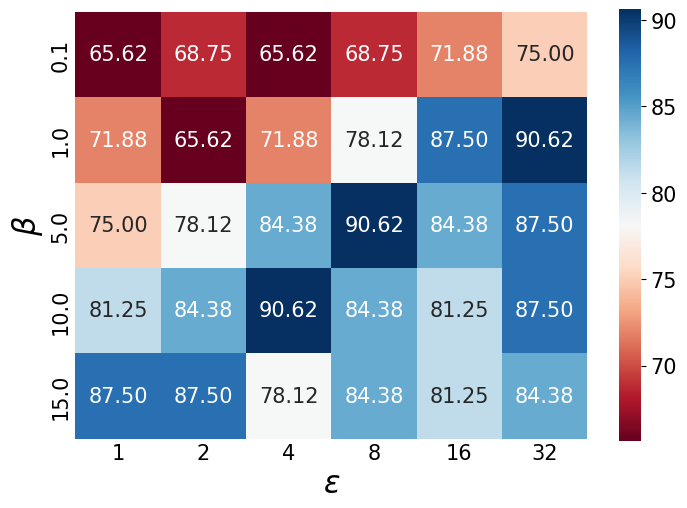}
		\label{fig:heatmap_chains1_epsilon_beta}
		}
      \subfloat[Cora]{
		\includegraphics[width = 4cm]{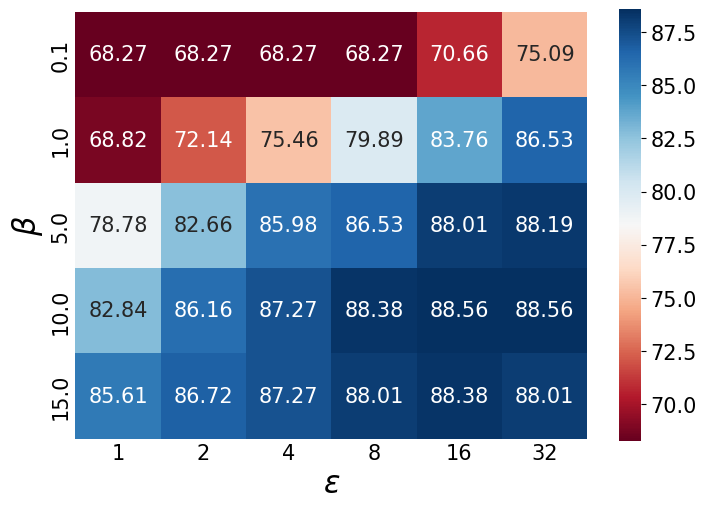}
		\label{fig:heatmap_cora_epsilon_beta}
		}
		\caption{Classification Accuracy under $\Clip, \alpha_1,\beta$ of CGL.}
        \label{fig:heatmaps}
        \vspace{-6mm}
		\end{figure}
{
\subsection{Larger Datasets Evaluation}
\label{subsec:larger}

To further assess the scalability on large real-world graphs, we additionally evaluate $\gcn$  on the Reddit2 dataset~\cite{zeng2019graphsaint}. Reddit2 contains over 232K nodes and 23M edges, representing a significantly larger and structurally richer benchmark. 
Table~\ref{tab:reddit2} summarizes the top-1 accuracy under both EDP and NDP across privacy budgets $\epsilon \in \{1,2,4,8,16,32\}$.

\begin{table}[H]
\small
    \caption{Evaluating EDP and NDP over Reddit2 Dataset} \label{tab:reddit2}
    \vspace{1pt}
    \centering
    \setlength\tabcolsep{5pt}
    \adjustbox{max width=\textwidth}{
    \begin{tabular}{c|c|c|c|c|c|c}
    \toprule[1pt]
    $\mathbf{\epsilon}$ & \textbf{1} & \textbf{2}& \textbf{4}& \textbf{8}& \textbf{16} & \textbf{32} \\
    \cmidrule[1pt]{1-7}
    EDP & $70.6\%$ & $73.8\%$ & $76.4\%$ & $78.6\%$ & $78.7\%$ & $80.9\%$ \\
    NDP & $63.7\%$ & $68.5\%$ & $73.2\%$ & $77.3\%$ & $78.4\%$ & $79.6\%$ \\
    \bottomrule
    \end{tabular}
    }
\end{table}

We observe:
1) $\gcn$ achieves $70.6\%$ (EDP) and $63.7\%$ (NDP) accuracy even at $\epsilon = 1$, improving steadily as privacy budgets relax, reaching $80.9\%$ (EDP) and $79.6\%$ (NDP) at $\epsilon = 32$. This demonstrates that $\gcn$ maintains stable utility even on million-edge graphs. 2) Due to the dataset’s large maximum degree and degree heterogeneity, NDP naturally incurs higher sensitivity. This is reflected in a larger EDP–NDP gap at small $\epsilon$, which narrows as noise reduces with increasing privacy budgets. 3) Despite the challenging properties of Reddit2, $\gcn$ avoids the noise explosion potentially observed in non-contractive private GNNs. The monotonic accuracy increase across $\epsilon$ validates that $\gcn$'s contractive operator bounds perturbation amplification over multiple hops.
These results collectively demonstrate that CARIBOU extends effectively to large-scale real-world social networks, providing additional empirical support for the generality of $\gcn$'s convergent privacy framework.
\subsection{Failure Case on Divergent Privacy}
\label{app:exp_failure}
Table~\ref{tab:failure-case} reports the accuracy of \gap\  under a fixed privacy budget $\epsilon=4$ as the 
depth $K$ increases. Since \gap\  relies on standard linear RDP composition, the required 
noise variance grows unboundedly with $K$, leading to a \emph{divergent} noise allocation 
across layers. This phenomenon is clearly reflected in the performance trend: while GAP 
achieves $77.10\%$ accuracy at $K=2$, the accuracy steadily degrades as the GNN becomes 
deeper, dropping to $71.03\%$ at $K=64$.
Figure~\ref{fig:k-study} also confirms that $\gcn$ shows better model accuracy than \gap across different $K$. 
Moreover, \gap further drops to $51.11\%$ at $K=128$, approaching random-guessing performance.
This study confirms that  non-convergent privacy accounting 
inevitably leads to excessive noise at large depths, which severely impairs model utility. 

\begin{table}[H]
\centering
\caption{
Accuracy with Divergent Noise Allocation. We  set $\epsilon=4$ and instantiate \gap\ on the \cora dataset as an example.
}
\label{tab:failure-case}
\setlength\tabcolsep{2pt}
\begin{tabular}{c|c|c|c|c|c|c|c}
\toprule
$K$ & $2$   & $4$ & $8$ & $16$ & $32$&$64$&$128$ \\
\midrule
Accuracy 
& $77.10\%$   & $75.65\%$ & $74.91\%$ & $73.43\%$ & $74.17\%$ & $71.03\%$ & $51.11\%$\\
\bottomrule
\end{tabular}
\end{table}

}

\subsection{More Related Works}
\label{app:related}


\subsubsection{Differential privacy for graph structures}
Euclidean data clearly states which data points are associated with a particular individual, for example, tabular data. 
Unlike tabular data, graph association can be interpreted into the combination by edges and nodes.
Corresponding to instance-level DP, the ``instance'' of graph data can be an edge or a node,  naturally called edge DP and node DP. 
Early stage works~\cite{tcc/GehrkeLP11,icdm/MirW09,pakdd/WangWW13,vldb/ChenFYD14,kdd/LuM14} started by private statistics estimation or counting, or release of private graphs under edge DP. 
Node DP protects against revealing the presence or absence of an individual node along with all its adjacent edges. 
Prior works consider various node-privacy algorithms~\cite{tcc/KasiviswanathanNRS13,pakdd/IftikharW21,sigmod/DayLL16,sigmod/XiaCKHT021,uss/ImolaMC21} tailored to specific graph statistics constrained by bounded node degrees.
Node DP typically offers a stronger level of protection compared to edge privacy, as inclusion of a node implies the inclusion of some particular edges.
Node DP is also considered harder to achieve than edge DP if maintaining reasonable utility loss.

\subsubsection{Differentially private GNNs}
Sharing model updates can lead to privacy risks, as adversaries may reconstruct training data. Such risks~\cite{sp/CarliniCN0TT22,ccs/0001MMBS22} stem from overfitting, where neural networks memorize training data~\cite{zhang2021understanding}. Graph structures, which encode relational information, are widely used in applications like intrusion detection~\cite{eurosp/ApruzzeseLS23,sp/SommerP10}, social recommendation~\cite{icde/XiaSH0XP23}, and drug discovery~\cite{kdd/ZhongM23}. Graph neural networks (GNNs) have emerged as a key approach for learning over graph-structured data, but their message-passing mechanisms can leak sensitive information about nodes and their neighbors.
To address these risks, several works~\cite{arxiv/sajadmanesh2020differential,ccs/sajadmanesh2021locally,uss/sajadmanesh2023gap,tcc/GuptaRU12} propose differentially private GNNs~\cite{arxiv/DaigavaneNode21,corr/abs-2109-08907,uss/sajadmanesh2023gap}. 
This necessitates injecting calibrated noise after each message passing layer, where the noise scale is proportional to $K$ for a fixed privacy budget $\epsilon$.
For instance, Wu \etal~\cite{sp/Wulinkteller22} achieve edge-level privacy by adding Laplace noise to adjacency matrix entries under a small DP budget. Kolluri \etal~\cite{ccs/KolluriBHS22} improve privacy-utility trade-offs by decoupling graph structure from the neural network architecture, querying the graph only when necessary. However, leveraging contractive hidden states in graph learning for enhanced privacy analysis remains an underexplored avenue.

\subsubsection{Privacy analysis for iterations}
Prior works analyze cumulative privacy costs by composing iteration-wise privacy guarantees using calibrated noise and privacy composition theorems. 
These theorems enable modular analysis of complex algorithms by controlling the total privacy budget. Feldman et al.~\cite{focs/FeldmanMTT18,stoc/FeldmanKT20} address privacy analysis for gradient computations over a single training epoch under smooth and convex loss functions. 
Recent advancements~\cite{nips/ChourasiaYS21,nips/AltschulerT22,siamcomp/AltschulerBT24} improve privacy analysis for convex and strongly convex losses by eliminating dependence on infinite iterations, leveraging the R\'enyi DP framework.
While effective, R\'enyi DP is lossy. 
$f$-DP framework offers tighter
analysis  through hypothesis testing curves.

For multi-hop message passing GNNs (e.g., \gap~\cite{uss/sajadmanesh2023gap}), the privacy guarantee scales linearly with the number of hops $K$. 
Foundational questions about privacy costs remain challenging, even for smooth convex losses over bounded domains~\cite{nips/AltschulerT22}. Recent works~\cite{nips/AltschulerT22,nips/0001S22} reveal that privacy leakage does not increase indefinitely with iterations. 
They demonstrate that after a brief burn-in period, additional iterations of SGD do not significantly impact privacy.

{
%

\subsubsection{Overview of privacy attacks}
Numerous privacy attacks have been developed against machine learning models, including membership inference, attribute inference, and reconstruction attacks. In the graph domain, these attacks are adapted to structured data, where the protected units are nodes, edges, or local subgraphs~\cite{uss/000100S022,sp/0011L0022,tpsisa/OlatunjiNK21}. MIA determine whether a specific node or edge appeared in the training graph, with variants studied in general ML models and specifically in GNNs. 
In $\gcn$, privacy auditing module aims to empirically estimate the privacy leakage through an 
attacking algorithm~\cite{nips/0002NJ23,ccs/0001MMBS22}. 
Attribute inference attacks~\cite{olatunji2023does} aim to recover sensitive node or edge attributes from predictions. More recent graph-specific attacks further exploit message-passing behaviors to infer private structural information~\cite{uss/000100S022}.

\textit{Attacking algorithms in $\gcn$.} The objective of $\gcn$ is not to benchmark the full space of graph attacks, but to conduct a  mechanism-level privacy audit of perturbed message passing under the \emph{black-box} membership threat model~\cite{sp/CarliniCN0TT22,tpsisa/OlatunjiNK21}  that underlies prior DP-GNN frameworks. 
In contrast to the upper bound obtained from theoretical privacy analysis, privacy auditing via MIA achieves the experimental true positive rate (TPR) and false positive rate (FPR) and then presents a lower bound on the privacy budget $\epsilon$~\cite{sp/CarliniCN0TT22}.
Within this setting, LinkTeller~\cite{sp/0011L0022} and the node-level GNN membership attack of Olatunji et al.~[42] (G-MIA) are canonical choices: both are tailored to GNNs, operate in the transductive setting we consider, require only query access, and have open-source implementations. Using these two attacks provides a clean and reproducible evaluation of edge- and node-level leakage, enabling direct comparison with private GNNs. 

Other attack families, such as attribute inference~\cite{olatunji2023does}  or reconstruction attacks~\cite{zhang2022model,ijcai/ZhangLHWLLC21} require white-box knowledge.
They aim to address different privacy notions or stronger threat models, and are therefore complementary rather than directly comparable to our membership-centric audit.
Considering a practical and generic adversary, we do not include these attacks that require white-box access to internal embeddings or gradients, or that focus on DP notions such as attribute inference or graph-level reconstruction.
Those attacks assume a strictly stronger adversary or target different sensitive objects than the edge- and node-level memberships that our DP guarantees protect.
Extending our empirical evaluation to such attacks in other application scenarios is an interesting direction for future work.
 }
 

\clearpage
\newpage

\section*{Artifact Evaluation Appendix}
\setcounter{subsection}{0}
\subsection{Artifact Summary}

 $\gcn$ is a privacy-preserving GNN framework with a convergent privacy analysis, enabling deeper message passing under both edge-DP and node-DP without sacrificing utility.
This artifact reproduces all experiments in Section~\ref{sec:exp}, as summarized in Table~\ref{tab:eval-index}. It provides one-command scripts and pinned environments to regenerate: (i) privacy–utility results under EDP/NDP across nine datasets; (ii) $\epsilon$-curves, $K$-hops curves,  $D$-curves, and heatmaps;  and (iii) overhead measurements (time and memory).
Overall, this artifact demonstrates that $\gcn$ achieves superior privacy–utility balance and efficiency among private GNNs, with results verifiable through automated scripts and publicly available datasets.

\begin{table}[H]
\caption{Summary of Artifact \& Evaluation Index. }
\label{tab:eval-index}
\begin{center}
\setlength\tabcolsep{4pt}
\begin{tabular}{c|l|l}
\toprule
     Category    & Experiments  & Evaluation \\
\midrule
 \multirow{2}{*}{PU} &PU-1 (Table~\ref{tab:reulst_table_overall_acc_top1}), PU-2 (Table~\ref{tab:reulst_table_overall_acc_top_node_5}), & \multirow{2}{*}{Sec.~\ref{ae:e1}} \\
  &  PU-3 (Table~\ref{tab:reulst_table_overall_acc_mean}) &  \\\hline
 \multirow{2}{*}{CRV} &CRV-$\epsilon$ (Figure~\ref{fig:k-study}), CRV-$K$ (Figure~\ref{fig:K-study})  & \multirow{2}{*}{Sec.~\ref{ae:e2}} \\
                  & CRV-$D$,  (Figure~\ref{fig:abla_nodedeg}),CRV-H (Figure~\ref{fig:heatmaps}) & \\\hline
 OV      & OV-E (Tables~\ref{tab:reulst_table_running_time},\ref{tab:reulst_table_memory}), OV-N (Tables~\ref{tab:reulst_table_running_time_ndp},\ref{tab:reulst_table_memory_ndp}) & Sec.~\ref{ae:e3} \\
\bottomrule
\end{tabular}
\end{center}
``PU'': privacy-utility experiments; ``CRV'': curves for accuracy as $\epsilon,K,D_{\max}$ varies; ``AUD'': privacy auditing experiments; and ``OV'': overhead measurements.
\end{table}

\subsection{Description \& Requirements}

This section provides all the information necessary to recreate the experimental setup to run our artifacts. All experiments can run on a commodity desktop machine.  

\subsubsection{How to access}
The’s main GitHub repository $\gcn$ can be found at \url{https://github.com/yuzhengcuhk/caribou-public}, 
and the exact version of the code for the evaluation of artifacts is also available at  
\url{https://doi.org/10.5281/zenodo.17539660} on the Zenodo platform.
$\gcn$ is released under the \textit{MIT License}.

\subsubsection{Hardware dependencies}
Our artifacts can be run on the Ubuntu 20.04.2 LTS server, with AMD Ryzen Threadripper 3970X 32-core CPUs of 256 GB CPU memory and NVIDIA GeForce RTX 3090 of 24GB memory. To ensure that all artifacts run correctly, it is recommended to use a machine with a similar configuration.

\subsubsection{Python environments} 
We provide a reproducible \texttt{conda} environment for evaluating all experiments, built on \texttt{Python~3.9.20}. Core system libraries are managed with \texttt{conda}, while most packages are pinned and installed via \texttt{pip}. 

\subsubsection{Datasets}
$\gcn$ can be tested over nine datasets, including \photo\ and \computers\ \citep{corr/abs-1811-05868}, \cora\ and \pubmed\ \citep{sen2008collective}, \facebook\ \citep{corr/abs-1102-2166}, \chainS, \chainM, \chainL\ and \chainX~\cite{nips/GuC0SG20}.

\subsubsection{Artifact Architecture}
 \texttt{core/} is the primary implementation directory, including:
\begin{scriptsize}
\begin{verbatim}
  - methods/ # DP-GNN implementations
  - models/ # neural network architectures
  - modules/ #  GNN building blocks
  - privacy/ # privacy mechanisms
  - datasets/ # data loading
  - trainer/ # training and test loop
  - args/ # CLI parsing and configuration
  - utils.py # general utilities
\end{verbatim}
\end{scriptsize}

\subsection{Artifact Installation \& Configuration}
\label{ae:aic}
\subsubsection{Download the code} To download the artifact, we recommend cloning via HTTPS. Execute the following commands to fetch the repository, and  then  enter the $\gcn$'s directory. 

\begin{center}
\begin{minipage}{0.95\linewidth}
\begin{lstlisting}[language=Python, label={lst:example}]
git clone https://github.com/yuzhengcuhk/caribou-public.git
cd caribou-public/
\end{lstlisting} 
\end{minipage}
\end{center}

\subsubsection{Build the environment}
Inside \texttt{caribou-public/}, we provide a script \texttt{setup\_minimal\_env.sh}  to quickly and conveniently build $\gcn$'s environment for evaluators. The shell script \texttt{setup\_minimal\_env.sh} creates a
\texttt{conda} environment (\texttt{caribou-minimal}) installs the core libraries from the CUDA~11.7 wheel index.

\begin{center}
\begin{minipage}{0.95\linewidth}
\begin{lstlisting}[language=Python, label={lst:example}]
chmod +x ./setup_minimal_env.sh
./setup_minimal_env.sh
conda activate caribou-minimal
python train.py mlp-dp --dataset cora --epsilon 2
\end{lstlisting}
\end{minipage}
\end{center}
For users who prefer \texttt{pip}-based installs, we also provide \texttt{requirements\_minimal\_caribou.txt}
that lists the same minimal set and helps evaluators to install relevant packages. 
After the environment is successfully provisioned, execute the example run via the third command. This launches a \emph{DP-MLP} training sanity check and prints the results.

\subsubsection{\texttt{train.py} entry point}
\texttt{train.py} serves as the unified launcher for all experiments. The \texttt{--method} flag selects the pipeline family and its privacy regime, including $\gcn$ and all baseline works in this paper: \texttt{mlp}, \texttt{mlp-dp}, \texttt{gap-inf}, \texttt{gap-edp}, \texttt{gap-ndp}, \texttt{sage-inf}, \texttt{sage-edp}, \texttt{sage-ndp}, \texttt{dpdgc-inf}, \texttt{dpdgc-edp}, \texttt{dpdgc-ndp},  \texttt{caribou-inf}, \texttt{caribou-edp}, \texttt{caribou-ndp}.  Internally, the entry point dispatches to the corresponding implementation in \texttt{core/methods/}.  
 
\subsection{Major Claims}
\label{ae:claim}
\begin{itemize}
    \item (\textbf{C1}): $\gcn$ can outperform all other baselines in most cases with varying $\epsilon$, no matter whether it is EDP or NDP. $\gcn$ achieves a more favorable privacy-utility trade-off.
    \item (\textbf{C2}): $\gcn$ consumes relatively low computational overhead in execution time and max memory usage. 
\end{itemize}

\subsection{Evaluation}
This section enumerates the end-to-end steps and experiments required to evaluate our artifact and validate the paper’s claims (\textbf{C1}, \textbf{C2} in Section~\ref{ae:claim}) for $\gcn$. The full evaluation can be completed in \emph{approximately 20 human minutes} and \emph{about 8 compute hours}, subject to hardware. 
We assume the host machine is correctly configured with the dependencies described in Section~\ref{ae:aic}. 
The same instructions can also be found in the \texttt{caribou-public/}, \ie, top-level \texttt{README} to provide a consistent and reproducible workflow. 

\subsubsection{Evaluation (\textbf{E1})} 
\label{ae:e1}
[PU-1 (5 human-minutes, 3 computer-hours) + PU-2 (5 human-minutes, 3 computer-hours) + PU-3 (5 human-minutes, 3 computer-hours)]

Under \texttt{AE/PU/}, we provide two scripts for quickly reproducing privacy–utility experiments.
\texttt{run\_scripts\_all.sh} executes the \emph{full} experiments across the configured datasets and privacy budgets to regenerate all three accuracy tables (Table~\ref{tab:reulst_table_overall_acc_top1}, Table~\ref{tab:reulst_table_overall_acc_top_node_5}, Table~\ref{tab:reulst_table_overall_acc_mean}) reported in the paper. It is intended for a complete reproduction run on a machine with
sufficient compute. For convenience to evaluators, \texttt{run\_scripts\_computers.sh} is a \emph{targeted} script
that launches the exactly same pipelines on the \texttt{computers} dataset, serving as a quick
sanity check to verify accuracy results.
To reduce the workload, we only compare $\gcn$ with the generally strongest baseline \gap\ in artifact evaluation.
It creates an output directory (\texttt{AE\_outputs/PU/}) for storing results.
Within the \texttt{caribou\_public} directory, execute:

\begin{center}
\begin{minipage}{0.95\linewidth}
\begin{lstlisting}[language=Python, label={lst:example}]
chmod +x AE/PU/pu1_run_scripts_computers.sh
./AE/PU/pu1_run_scripts_computers.sh # PU-1

chmod +x AE/PU/pu2_run_scripts_computers.sh
./AE/PU/pu2_run_scripts_computers.sh # PU-2

chmod +x AE/PU/pu3_run_scripts_computers.sh
./AE/PU/pu3_run_scripts_computers.sh # PU-3
\end{lstlisting} 
\end{minipage}
\end{center}

Each evaluation item (i.e., each table entry) emits a plain-text report. Evaluators can open the
corresponding \texttt{.txt} file to read either the top accuracy or the mean accuracy
over three independent runs. Because training seeds are randomized and i.i.d noise is sampled every time, some  fluctuations  across
executions are expected. 
The \texttt{Chain} datasets are particularly noise–sensitive, and thus can
exhibit comparatively larger variation. 
The \texttt{Chain} datasets are particularly noise-sensitive and thus can exhibit comparatively larger variation. 
This sensitivity is due to their structure: non-zero features are present only at the first node of each chain. Information must propagate from this source, and it can be degraded by noise accumulation during propagation. This task is already challenging with small number of nodes, and its difficulty can be unwittingly amplified by the random train/val/test split. 
To \textit{keep runtime modest} as required by artifact evaluation, the provided scripts sweep only a \emph{subset} of the hyper-parameters used
for the paper’s final tables. As a result, the top accuracy observed by evaluators during
reproduction sometimes may   be  lower than the maxima reported in the paper. This is expected as the paper
results reflect a broader search grid. 
Despite these effects, $\gcn$ exceeds the
accuracy of competing baselines in the majority settings (Claim~\textbf{C1}).

\subsubsection{Evaluation (\textbf{E2})} 
\label{ae:e2}
[CRV-$\epsilon$ (5 human-minutes, 0.5 computer-hours) + CRV-$K$ (5 human-minutes, 0.5 computer-hours) + CRV-$D$ (5 human-minutes, 0.5 computer-hours) + + CRV-H (5 human-minutes, 0.5 computer-hours)]
Before re-generating Figures~\ref{fig:k-study},\ref{fig:K-study},\ref{fig:abla_nodedeg},\ref{fig:heatmaps} we need to install \texttt{jupyter} and register this environment as a Jupyter kernel.
Notably, each of the following commands is a \textit{single line}. Some PDF viewers soft-wrap lines, which can corrupt copy-and-paste. To prevent errors, paste each command into a plain-text (\texttt{.txt}) file, remove any spurious whitespace, then paste and execute it in the terminal. 
To reproduce all figures, execute the following commands within the \texttt{caribou\_public} directory:
\begin{center}
\begin{minipage}{0.95\linewidth}
\begin{lstlisting}[language=Python, label={lst:example}]
python -m pip install jupyter 
python -m pip install seaborn==0.13.2
 
python -m ipykernel install --user --name caribou-minimal --display-name "Python (caribou-minimal)"
 
jupyter nbconvert --to notebook --execute --inplace   --ExecutePreprocessor.kernel_name=caribou-minimal   ./AE/CRV/eps_hop_plots.ipynb # CRV-eps, CRV-K

jupyter nbconvert --to notebook --execute --inplace   --ExecutePreprocessor.kernel_name=caribou-minimal   ./AE/CRV/degree_plots.ipynb # CRV-D

jupyter nbconvert --to notebook --execute --inplace   --ExecutePreprocessor.kernel_name=caribou-minimal   ./AE/CRV/heatmap.ipynb # CRV-H


\end{lstlisting} 
\end{minipage}
\end{center}
The evaluated results are embedded within the corresponding \texttt{.ipynb} notebooks and can be inspected directly after execution.
In addition, the generated figures are stored in \texttt{AE\_outputs/CRV/}.
\subsubsection{Evaluation (\textbf{E3})} 
\label{ae:e3}
[OV-E (5 human-minutes, 2 computer-hours) + OV-N (5 human-minutes, 2 computer-hours)]
We provided two scripts to execute the overhead benchmarks for EDP and NDP, respectively. 
Within the \texttt{caribou\_public} directory, execute:

\begin{center}
\begin{minipage}{0.95\linewidth}
\begin{lstlisting}[language=Python, label={lst:example}]
chmod +x AE/OV/ove_run_scripts.sh
./AE/OV/ove_run_scripts.sh

chmod +x AE/OV/ovn_run_scripts.sh
./AE/OV/ovn_run_scripts.sh
\end{lstlisting} 
\end{minipage}
\end{center}
Upon completion of each run, the corresponding computational overhead results are available in \texttt{AE\_outputs/OV/}. 
This directory contains the actual execution time and the maximum memory usage results referenced in Tables~\ref{tab:reulst_table_running_time},\ref{tab:reulst_table_memory},\ref{tab:reulst_table_running_time_ndp},\ref{tab:reulst_table_memory_ndp}.
Note that absolute values are hardware–dependent; therefore, the numbers obtained by the evaluators may differ from those exactly reported in the paper. This variation is expected. Provided the machine offers sufficient compute resources, the qualitative trends and relative comparisons should remain consistent with our results (Claim \textbf{C2}).

\end{document}